\documentclass[11pt,letterpaper,abstract=true]{scrartcl}
\usepackage[margin=1in]{geometry}
\usepackage[T1]{fontenc}
\usepackage{lmodern}
\usepackage{microtype}
\usepackage[UKenglish]{babel}
\usepackage{mathtools}
\usepackage{amssymb,amsthm}
\usepackage{thm-restate}
\usepackage{graphicx}
\usepackage{fontawesome5}
\usepackage{booktabs}
\usepackage{acronym}
\usepackage{xcolor}
\usepackage{bm}
\newcommand{\vect}[1]{\bm{#1}}
\usepackage{tabularx}
\usepackage{etoolbox}
\usepackage[hidelinks,unicode]{hyperref}
\usepackage[capitalise,noabbrev]{cleveref}
\usepackage[
    backend=biber,
    style=alphabetic,
    minalphanames=3,
    maxalphanames=4,
    maxbibnames=99,
    doi=true,
    giveninits=true
]{biblatex}
\hypersetup{colorlinks=true,linkcolor=blue!40!black,citecolor=blue!40!black,urlcolor=blue!40!black,
    pdftitle={A Constructive FPT Algorithm for Integer Cone Optimization with Convex Costs},
    pdfsubject={Integer cones, fixed-parameter tractability, convex optimization, high-multiplicity bin packing}
}

\numberwithin{equation}{section}

\theoremstyle{plain}
\newtheorem{theorem}{Theorem}
\newtheorem{lemma}[theorem]{Lemma}
\newtheorem{corollary}[theorem]{Corollary}
\newtheorem{proposition}[theorem]{Proposition}
\theoremstyle{remark}
\newtheorem{remark}[theorem]{Remark}
\newtheorem*{claim*}{Claim}

\acrodef{ETH}[\textsc{ETH}]{Exponential Time Hypothesis}
\acrodef{FPT}[\textsc{FPT}]{Fixed-Parameter Tractable}
\acrodef{EPTAS}[\textsc{EPTAS}]{Efficient Polynomial-Time Approximation Scheme}
\acrodef{PTAS}[\textsc{PTAS}]{Polynomial-Time Approximation Scheme}

\newcommand{\Z}{\mathbb Z}
\newcommand{\bZ}{\mathbb{Z}}
\newcommand{\R}{\mathbb R}

\newcommand{\Q}{\mathbb Q}

\newcommand{\enc}{\operatorname{enc}}

\DeclarePairedDelimiterX\setImpl[1]{\{}{\}}{#1}
\NewDocumentCommand\set{sO{}m>{\TrimSpaces}o}{
  \IfValueTF{#4}{
    \setImpl[#2]{#3 : \IfBooleanTF{#1}{\text{#4}}{#4}}
  }{
    \setImpl[#2]{#3}
  }
}

\theoremstyle{definition}
\newtheorem{assumption}[theorem]{Assumption}
\crefname{assumption}{assumption}{assumptions}
\Crefname{assumption}{Assumption}{Assumptions}
\theoremstyle{plain}

\title{A Constructive FPT Algorithm for Integer Cone Optimization with Convex Costs}
\author{%
    Klaus Jansen\thanks{Kiel University, kj@informatik.uni-kiel.de, \url{https://orcid.org/0000-0001-8358-6796}}
    \and
    Felix Ohnesorge\thanks{Kiel University, foh@informatik.uni-kiel.de, \url{https://orcid.org/0009-0003-8023-3380}}
    \and
    Björn Schumacher\thanks{Kiel University, bsch@informatik.uni-kiel.de, \url{https://orcid.org/0009-0002-2448-4794}}
}
\date{}
\AfterEndEnvironment{abstract}{%
    \par\smallskip
    {\small\noindent\textbf{Keywords.} Integer cones, fixed-parameter tractability, convex optimization.\par
    \smallskip\noindent\textbf{Classification.} Theory of computation: fixed parameter tractability; integer programming.\par
    \smallskip\noindent\textbf{Funding.} Deutsche Forschungsgemeinschaft (DFG, German Research Foundation) -- Project number 453769249.}%
}

\begin{document}
\maketitle
\begin{abstract}
    We give a deterministic fixed-parameter algorithm for integer cone optimization with convex costs.
    Given $\tau$ rational polytopes in $\R^{d+k}$, we combine generators whose first $d$ coordinates are integral and whose remaining $k$ coordinates may be fractional.
    Each generator has a convex cost depending only on its integer coordinates.
    We minimize the total cost of nonnegative integer combinations subject to arbitrary rational linear constraints on the full typewise aggregates and generator counts.
    Our algorithm runs in time
    \[
        2^{\tau\log(\tau+1)\,2^{O(d)}}\cdot |I|^{O(1)},
    \]
    where $|I|$ is the total binary input length.
    It constructs an optimal representation using at most $2^{2d+1}$ distinct generators per type.
    The number $k$ of continuous coordinates contributes only to the polynomial input factor.

    For a single generator type, the running time becomes $2^{2^{O(d)}}\cdot |I|^{O(1)}$.
    Applications include constructive FPT algorithms for Cone and Polytope Intersection, high-multiplicity bin packing, vector bin packing, and makespan minimization on uniform and unrelated machines.
    The double-exponential dependence on $d$ is tight for Point-in-Cone and both bin-packing problems assuming $\mathrm{NP}\nsubseteq\mathrm{SUBEXP}$.
    The mixed-integer formulation also yields additive EPTASs for makespan minimization, machine covering, and envy minimization on identical machines, and additive schemes fixed-parameter tractable in $1/\varepsilon+\tau$ on uniform machines.
\end{abstract}

\newpage
\section{Introduction}\label{sec:introduction}

Integer cones describe how a collection of feasible local generators can be combined into a global solution.
A generator may represent the contents of one bin, the jobs assigned to one machine, or a feasible allocation of several resources; its nonnegative integer coefficient records how often that generator is used.
The \emph{Cone and Polytope Intersection} (CAPI) problem captures this common structure: the generators are the integer points of a bounded rational polyhedron $P = \{\vect x : A \vect x \leq \vect b\}$, and their sum must lie in a second rational polyhedron $Q = \{\vect y : B \vect y \leq \vect c\}$.
Formally, the integer cone generated by the lattice points of $P$ is
\begin{equation*}
    \operatorname{int{.}cone}(P \cap \Z^d)
    \coloneqq  \left\{ \sum_{\vect x \in P \cap \Z^d}
    \lambda_{\vect x}\vect x :
    \lambda_{\vect x} \in \Z_{\ge 0} \right\}.
\end{equation*}
The task is to determine whether
\begin{equation}\label{eq:generalized-integer-cone}\tag{CAPI}
    \operatorname{int{.}cone}(P \cap \Z^d) \cap Q \ne \emptyset,
\end{equation}
and, if so, to find a vector $\vect y$ in this intersection together with a representation $\vect y = \sum_{\vect x} \lambda_{\vect x}\vect x$.

A central difficulty is that the generating set $P \cap \Z^d$ is given implicitly and may be exponentially large in the input length, even in fixed dimension.
Nevertheless, its convex structure guarantees sparse representations.
Eisenbrand and Shmonin~\cite{eisenbrand2006} proved that every vector in $\operatorname{int{.}cone}(P \cap \Z^d)$ admits a representation using at most $2^d$ distinct generators.
This bounds the number of generators needed, but does not identify them: both the generators and their multiplicities remain unknown and an exponential support can be necessary even for special cases~\cite{kowalik2024,jansen2025esa}.

Goemans and Rothvoss~\cite{goemans2020} established polynomial-time solvability of~\eqref{eq:generalized-integer-cone} for every fixed dimension.
Their algorithm finds a feasible vector and a certificate supported on at most $2^{2d+1}$ generators in time
\begin{equation*}
    \enc(P)^{2^{O(d)}} \cdot \enc(Q)^{O(1)}.
\end{equation*}
Thus, the Cone and Polytope Intersection Problem belongs to XP parameterized by $d$.
Jansen and Klein~\cite{jansenKlein2020} subsequently refined the underlying structure theorem using the vertices $V$ of the integer hull $\operatorname{conv}(P \cap \Z^d)$, obtaining running time $|V|^{2^{O(d)}} \cdot {|I|}^{O(1)}$, where $|I|$ is the binary input length of the problem.
Their result improves the dependence on the input when the integer hull has few vertices, but $|V|$ is not bounded by a function of $d$ alone.
These results therefore left open whether an \ac{FPT} algorithm with running time $f(d)\cdot {|I|}^{O(1)}$ exists.

This problem encompasses the \emph{Point-in-Cone} problem in which $Q=\{\vect y\}$ for a prescribed integer vector $\vect y$ and the \emph{High-Multiplicity Bin Packing} problem in which $P$ is a knapsack polytope as well as other scheduling problems.

Recently, Koana and Kumabe~\cite{koana2026} obtained an \ac{FPT} algorithm for high-multiplicity bin packing, with deterministic running time $2^{2^{O(d\log d)}} \cdot {|I|}^{O(1)}$, where $d$ is the number of item types.

\begin{figure}[ht]
    \centering
    \includegraphics[width=0.9\columnwidth]{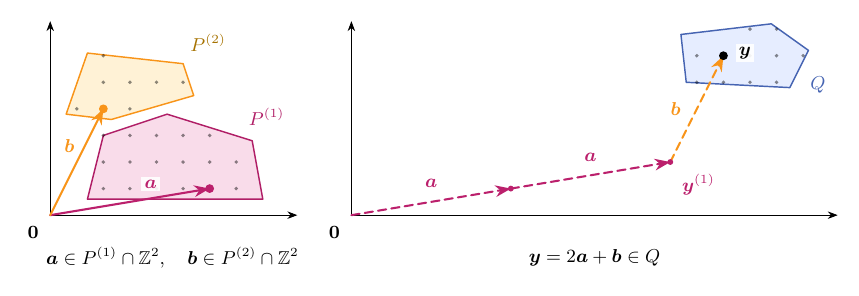}%
    \caption{Illustration of the integer-generator case ($k=0$) for two generator polytopes ($\tau=2$) in dimension $d=2$.
        The left panel shows the integer generators in $P^{(1)}$ and $P^{(2)}$.
        The right panel shows two copies of $\vect a\in P^{(1)}\cap\Z^2$ and one copy of $\vect b\in P^{(2)}\cap\Z^2$ added successively: the partial sum is $\vect y^{(1)}=2\vect a$, and the total sum $\vect y=\vect y^{(1)}+\vect b$ lies in $Q$.}
    \label{fig:problem}
\end{figure}

We study a generalization of~\eqref{eq:generalized-integer-cone} that allows several generator types, convex costs, and fractional coordinates.
Let $d\ge1$ and $k\ge0$ be integers, and let $P^{(1)},\ldots,P^{(\tau)}\subseteq\R^{d+k}$ be rational polytopes.
Write a generator as $(\vect x,\vect u)$, where $\vect x\in\Z^d$ records integer choices and $\vect u\in\R^k$ records fractional allocations.
We write
\[
    P_{\mathrm{x}}^{(t)}
    \coloneqq \{\vect x\in\R^d:\text{there exists }\vect u\in\R^k
    \text{ with }(\vect x,\vect u)\in P^{(t)}\}
\]
for the domain of the integer-coordinate block, and define
\[
    S^{(t)}\coloneqq P^{(t)}\cap(\Z^d\times\R^k),
    \qquad
    X^{(t)}\coloneqq P_{\mathrm{x}}^{(t)}\cap\Z^d.
\]
For each type $t\in[\tau]\coloneqq \{1,\ldots,\tau\}$, the cost function $f^{(t)}:P_{\mathrm{x}}^{(t)}\to\R$ depends only on this block: the cost of $(\vect x,\vect u)$ is $f^{(t)}(\vect x)$.
Throughout, each $f^{(t)}$ is an explicitly given polynomial in these $d$ variables with rational coefficients, in dense encoding by total degree, as in Heinz~\cite{heinz2005} and Hildebrand and K\"oppe~\cite[Section~5.1]{hildebrandKoeppe2013}; convexity is required on $P_{\mathrm{x}}^{(t)}$.\footnote{For the integer-generator case $k=0$, Appendix~\ref{sec:general-convex-costs} extends the result to efficiently approximable convex costs with a certified gap between sparse objective values.}

Given rational polyhedra $Q\subseteq\R^{d+k}$ and $\mathcal A\subseteq\R^{\tau(d+k+1)}$, we seek finitely supported nonnegative integer multiplicities solving
\begin{equation}\label{eq:main-problem}\tag{$\clubsuit$}
    \begin{aligned}
        \min_{\vect\lambda}\quad
         & \sum_{t=1}^{\tau}\sum_{(\vect x,\vect u)\in S^{(t)}}
        \lambda^{(t)}_{(\vect x,\vect u)}f^{(t)}(\vect x)                     \\
        \text{s.t.}\quad
         & \sum_{t=1}^{\tau}(\vect y^{(t)},\vect v^{(t)})\in Q,               \\
         & ((\vect y^{(t)},\vect v^{(t)},N^{(t)}))_{t=1}^{\tau}\in\mathcal A, \\
         & \lambda^{(t)}_{(\vect x,\vect u)}\in\Z_{\geq 0}
        \qquad(t\in[\tau],\ (\vect x,\vect u)\in S^{(t)}),
    \end{aligned}
\end{equation}
where the full typewise aggregate and generator count are
\[
    (\vect y^{(t)},\vect v^{(t)})
    \coloneqq \sum_{(\vect x,\vect u)\in S^{(t)}}\lambda^{(t)}_{(\vect x,\vect u)}(\vect x,\vect u),
    \qquad
    N^{(t)}\coloneqq \sum_{(\vect x,\vect u)\in S^{(t)}}\lambda^{(t)}_{(\vect x,\vect u)}.
\]
Finite support makes these sums well-defined even when $S^{(t)}$ is infinite.
When $k=0$, we omit the empty continuous blocks.
All coordinates may occur in the aggregate constraints, while only the integer coordinates enter the costs.
For example, the additive scheduling scheme in \Cref{sec:additive-scheduling} uses integer job-class counts together with fractional assignments of individual jobs.

Figure~\ref{fig:problem} illustrates the special case $k=0$.
The formulation extends the Multitype Integer Monoid Optimization (MIMO) problem~\cite{knop2019mimo} and includes the usual \emph{configuration IP} when $k=0$ and $Q=\{\vect y\}$.
In particular, $\mathcal A$ allows prescribing exact generator counts $N^{(t)}=m^{(t)}$, with $m^{(t)}\in\Z_{\ge0}$, separately for every type.
Without additional typewise constraints, set $\mathcal A=\R^{\tau(d+k+1)}$.

As an intermediate step, we solve the problem under an auxiliary bound
\begin{equation}\label{eq:count-bound}
    \sum_{t=1}^{\tau}N^{(t)}\le N_{\max}.
\end{equation}
For prescribed counts, one can take $N_{\max}=\sum_t m^{(t)}$.
Afterwards, we show how to remove this bound.

We write $\enc(P)$ and $\enc(Q)$ for the binary encoding lengths of their defining inequality systems, and set $|I|\coloneqq \enc(P)+\enc(Q)$ for the original single-type feasibility problem.
For optimization, $|I|$ includes the descriptions of all $P^{(t)}$, $f^{(t)}$, $Q$, and $\mathcal A$.
The parameters are $d$ and $\tau$; the number $k$ of continuous coordinates, the coefficient magnitudes, and the number of inequalities are not bounded in terms of these parameters.
The original \eqref{eq:generalized-integer-cone}~problem~\cite{goemans2020} is recovered for $k=0$, $\tau=1$, $P^{(1)}=P$, $f^{(1)}=0$, and $\mathcal A=\R^{d+1}$.

We extend the approach of Koana and Kumabe~\cite{koana2026} to~\eqref{eq:main-problem}, improving the parameter dependence and explicitly constructing sparse optimal certificates.
The main challenges are to reconstruct generators without increasing their convex cost and to handle arbitrarily many continuous coordinates while keeping their contribution to the running time polynomial.
We address these challenges through a residue decomposition that controls the cost and support bounds that depend only on the integer dimension.
\Cref{thm:main} establishes fixed-parameter tractability in $d+\tau$.

\begin{theorem}\label{thm:main}
    The problem~\eqref{eq:main-problem} can be solved deterministically in time
    \begin{equation*}
        2^{\tau\log(\tau+1)\,2^{O(d)}}\cdot {|I|}^{O(1)}.
    \end{equation*}
    The algorithm reports infeasibility or unboundedness below, or returns an optimal aggregate $(\vect y,\vect v) = \sum_{t=1}^{\tau} (\vect y^{(t)}, \vect v^{(t)}) \in Q$ satisfying $((\vect y^{(t)}, \vect v^{(t)}, N^{(t)}))_{t = 1}^\tau \in \mathcal{A}$, and a certificate consisting of rational generators and integer multiplicities, with at most $2^{2d+1}$ distinct generators per type.
\end{theorem}

The applications in \Cref{sec:applications} include Point-in-Cone, high-multiplicity Bin Packing, Vector Bin Packing, and Bin Covering, each with deterministic running time $2^{2^{O(d)}}\cdot |I|^{O(1)}$ and a sparse solution certificate.
Under the \ac{ETH}, the dependence on $d$ is tight for point-in-cone and the bin-packing applications (see~\cite{kowalik2024,jansen2026}).

\begin{remark}[A weaker hardness assumption]\label{rem:weaker-hardness-assumption}
    For the above lower bounds against deterministic algorithms, the \ac{ETH} can be replaced by the weaker assumption $\mathrm{NP}\nsubseteq\mathrm{SUBEXP}$, where $\mathrm{SUBEXP}\coloneqq\bigcap_{\varepsilon>0}\mathrm{DTIME}(2^{L^\varepsilon})$ and $L$ denotes input length.
    Indeed, the constructions in~\cite{kowalik2024,jansen2026} yield polynomially many polynomial-size instances with $d=O(\log L)$ from a 3-SAT input of length $L$.
    A $2^{2^{o(d)}}\cdot |I|^{O(1)}$-time algorithm would therefore solve 3-SAT in time $2^{L^{o(1)}}$, implying $\mathrm{NP}\subseteq\mathrm{SUBEXP}$.
\end{remark}

For scheduling on uniform or unrelated machines (with $\tau$ types) with fixed-power convex load costs, and for cutting stock with several stock types, we obtain constructive \ac{FPT} algorithms parameterized by $d+\tau$.
The dependence on $\tau$ cannot in general be removed~\cite{kouteckyZink2025}.
Further, we obtain a \ac{FPT} algorithm for the generalized-mean Welfare Problem, solving the open question in~\cite{nguyenRothe2026,nguyenElbassioni2026}.
\Cref{tab:application-parameters} compares these parameterizations with known algorithms.
For vector bin packing, $r$ denotes the number of resource dimensions and $\Delta$ the largest item-size coefficient after clearing denominators in each resource constraint.
For huge $n$-fold optimization, $r$ is the number of global equations, $s$ the maximum number of local equations per type, and $\Delta$ the maximum absolute matrix entry.
For scheduling, $m$ is the total number of machines and $p_{\max}$ is the largest integral processing requirement, before speed scaling in the uniform model.
\begin{table}[htbp]
    \caption{Applications and parameterizations. Numerical data and multiplicities are binary encoded, except that the additive schemes use explicitly listed jobs and machines.
        All our algorithms are constructive. Our exact algorithms take $2^{\tau\log(\tau+1)\,2^{O(d)}}|I|^{O(1)}$ time; rows parameterized by $d$ use $\tau=1$.
        The welfare order $p<1$, $p\ne0$, and the load exponent $q$ are fixed.
        The additive schemes have the same running-time form with $d$ replaced by $1/\varepsilon$.}
    \label{tab:application-parameters}
    \centering
    \small
    \setlength{\tabcolsep}{4pt}
    \renewcommand{\arraystretch}{1.15}
    \begin{tabularx}{\linewidth}{@{}>{\raggedright\arraybackslash}p{0.29\linewidth}>{\raggedright\arraybackslash}X>{\raggedright\arraybackslash}p{0.22\linewidth}@{}}
        \toprule
        Application & Known Algorithms                                                                                                                            & Our Result \\
        \midrule
        Point-in-Cone
                    & XP in $d$~\cite{goemans2020}; \ac{FPT} in $|V| + d$~\cite{jansenKlein2020}
                    & \ac{FPT} in $d$                                                                                                                                          \\
        Bin Packing
                    & \ac{FPT} in $d$: $2^{2^{O(d\log d)}}|I|^{O(1)}$~\cite{koana2026}
                    & \ac{FPT} in $d$                                                                                                                                          \\
        Vector Bin Packing
                    & XP in $d$~\cite{goemans2020}; \ac{FPT} in $d+r+\Delta$~\cite{knop2021}
                    & \ac{FPT} in $d$                                                                                                                                          \\
        Bin Covering
                    & XP in $d$ via the cone theorem of~\cite{goemans2020}
                    & \ac{FPT} in $d$                                                                                                                                          \\
        Several Bin Types / Cutting Stock
                    & XP in $d+\tau$~\cite{goemans2020}
                    & \ac{FPT} in $d+\tau$                                                                                                                                     \\
        Huge $n$-Fold Optimization
                    & \ac{FPT} in $r+s+\Delta$ for integer bricks and separable convex costs~\cite{knop2023huge}
                    & \ac{FPT} in $d+\tau$; mixed-integer bricks                                                                                                               \\
        Uniform/Unrelated Makespan
                    & \ac{FPT} in $d+p_{\max}$ or $d+m$~\cite{knop2019mimo,kouteckyZink2025,rohwedder2025uniform,jansen2025combinatorial}
                    & \ac{FPT} in $d+\tau$                                                                                                                                     \\
        Uniform/Unrelated $\sum_i L_i^q$
                    & \ac{FPT} in $d+p_{\max}$ or $d+m$~\cite{knop2019mimo,knop2023huge}
                    & \ac{FPT} in $d+\tau$                                                                                                                                     \\
        Generalized-Mean Welfare $M_p$
                    & \ac{FPT}AS for fixed $d$; \acs{PTAS} for fixed $\tau$~\cite{nguyenElbassioni2026}
                    & Exact \ac{FPT} in $d+\tau$                                                                                                                               \\
        \midrule
        Identical-Machine Makespan / Covering / Envy
                    & Additive \acs{PTAS}, error $\varepsilon p_{\max}$~\cite{DBLP:conf/icalp/BuchemRVW21}
                    & Additive \acs{EPTAS}, same error                                                                                                                         \\
        Uniform-Machine Makespan (\(\tau\) speeds) / Covering / Envy
                    & Target-interval method after speed scaling: $m^{\tau+1}n^{O(1/\varepsilon)}|I|^{O(1)}$~\cite[Theorem~12]{DBLP:journals/corr/abs-2007-09333}
                    & \ac{FPT} in $1/\varepsilon+\tau$, error $\varepsilon p_{\max}/s_{\min}$                                                                                  \\
        \bottomrule
    \end{tabularx}
\end{table}

\subsection{Algorithmic Ideas}
Our algorithm builds on the residue decomposition of Koana and Kumabe~\cite{koana2026}.
The starting idea is to group generators and describe each group by its total contribution and its number of generators, rather than choosing every generator individually.
We group generators according to their integer coordinates modulo $2^d$.
Although there are $2^{d^2}$ residue classes per type, a sparsification argument shows that every feasible solution can be replaced, without increasing its cost, by one using at most $2^d$ classes per type.
We therefore enumerate these possible sets of active classes.

For each guessed set of active residue classes, we formulate a mixed-integer linear program.
For each selected class, we introduce $d$ integer variables for the aggregate integer coordinates, one integer variable for the generator count, and $k$ continuous variables for the aggregate continuous coordinates.
Thus, the total number of integer variables is at most $D=\tau(d+1)2^d$.
The constraints require these aggregates to satisfy the original problem's constraints and to admit a representation using fractional multiplicities of generators from their respective classes.
We handle the continuous variables through linear programming and projection, leaving an integer-feasibility problem whose dimension depends only on $d$ and $\tau$.

The central challenge is to turn such a fractional representation into actual generators with integer multiplicities.
Our fractional cost oracle finds a minimum-cost representation using at most $2^d$ generators per class, independently of $k$.
We retain the integer parts of their multiplicities and reconstruct the remaining fractional mass by averaging.
The residue structure makes this reconstruction possible while preserving every aggregate coordinate and the generator count; convexity ensures that the cost does not increase.

We first carry out this procedure under a bound on the total number of generators.
A uniform multiplicity bound, together with a test for unboundedness, then allows us to solve the problem without a supplied count bound.

\subsection{An Application: Additive Scheduling Schemes}
A concrete consequence of \Cref{thm:main} is an additive \ac{EPTAS} for makespan minimization and machine covering on identical machines, answering a question of Buchem, Rohwedder, Vredeveld, and Wiese~\cite{DBLP:conf/icalp/BuchemRVW21}.
For every $\varepsilon\in(0,1]$, we obtain a schedule with makespan at most $\mathrm{OPT}+\varepsilon p_{\max}$, or minimum load at least $\mathrm{OPT}-\varepsilon p_{\max}$, in time $2^{2^{O(1/\varepsilon)}}\cdot |I|^{O(1)}$.

We illustrate the reduction for makespan minimization.
Consider $n$ explicitly listed jobs with processing times $p_j>0$ and $m\ge1$ explicitly listed identical machines.
Set $\eta=1/\lceil2/\varepsilon\rceil$ and partition the jobs into $d=1/\eta$ classes $J_1,\ldots,J_d$ according to intervals of width $\eta p_{\max}$.
For a target makespan $T$, define the generator polytope
\[
    P_T=\left\{(\vect x,\vect u)\in
    \R_{\ge0}^{d}\times[0,1]^n:
    \sum_{j\in J_c}u_j=x_c\quad(c\in[d]),\qquad
    \sum_{j=1}^{n}p_j u_j\le T
    \right\}.
\]
A generator in $P_T\cap(\Z^d\times\R^n)$ describes one machine: $u_j$ is the fraction of job $j$ assigned to it, while the total assignment from each class must be integral.
Prescribing the generator count $N=m$ and the aggregate constraint $Q=\R^d\times\{1\}^n$ ensures that every job is assigned in total once.
With zero costs, this is precisely the slot-MILP relaxation of~\cite{DBLP:journals/corr/abs-2007-09333}.

Our theorem solves this relaxation with $\tau=1$, $d=O(1/\varepsilon)$, and $k=n$.
Thus the individual job assignments contribute only to the polynomial input factor.
Every schedule of makespan at most $T$ yields a feasible relaxation, and the rounding algorithm of~\cite[Lemma~11]{DBLP:journals/corr/abs-2007-09333} converts any feasible solution into a schedule of makespan at most $T+\eta p_{\max}$.
Binary search to additive precision $\varepsilon p_{\max}/2$ therefore gives the stated guarantee.
Replacing the upper load constraint by a lower bound and reversing the search gives machine covering.
The full argument, together with extensions to related machines and envy minimization, appears in \Cref{sec:additive-scheduling}.

\section{Preliminaries}\label{sec:preliminaries}

Since each $P^{(t)}$ is bounded, standard polyhedral theory (see, e.g.,~\cite[Theorem~10.3]{schrijver1986}) yields a polynomial-time computable integer $M\ge1$ with $\enc(M)=|I|^{O(1)}$ such that
\begin{equation}\label{eq:coordinate-bound}
    P^{(t)}\subseteq[-M,M]^{d+k}\qquad\qquad\forall t\in[\tau].
\end{equation}
Consequently, all integer-coordinate vectors $\vect x\in X^{(t)}$ have polynomial encoding length.
Our algorithm accesses the costs only at these vectors and uses only the following properties, which are assumed throughout the paper.

\begin{assumption}\label{ass:cost-properties}
    The cost functions $f^{(t)}:P_{\mathrm{x}}^{(t)}\to\R$ satisfy the following properties on $X^{(t)}=P_{\mathrm{x}}^{(t)}\cap\Z^d$:
    \begin{enumerate}
        \item \emph{Discrete Jensen inequality:} For every $\vect x \in X^{(t)}$ and every convex combination $\vect x = \sum_{i=1}^m \theta_i \vect x_i$ of points $\vect x_i \in X^{(t)}$, we have $f^{(t)}(\vect x) \le \sum_{i=1}^m \theta_i f^{(t)}(\vect x_i)$.
        \item \emph{Polynomial-time exact evaluation and magnitude bound:} For every $t\in[\tau]$ and $\vect x\in X^{(t)}$, the value $f^{(t)}(\vect x)\in\Q$ can be evaluated exactly in $|I|^{O(1)}$ time.
              There is an integer $0 \leq b_f \leq |I|^{O(1)}$, computable in $|I|^{O(1)}$ time, such that $|f^{(t)}(\vect x)|\le2^{b_f}$ for all such $t$ and $\vect x$.
        \item \emph{Common lattice denominator:} An integer $q\ge1$ with $\enc(q)=|I|^{O(1)}$ can be computed in $|I|^{O(1)}$ time such that $f^{(t)}(\vect x)\in q^{-1}\Z$ for all $t\in[\tau]$ and $\vect x\in X^{(t)}$.
    \end{enumerate}
\end{assumption}
Explicitly given convex polynomials with rational coefficients satisfy these properties under the dense encoding specified in the introduction.
Convexity gives the discrete Jensen inequality.
The polynomial descriptions and the coordinate bound~\eqref{eq:coordinate-bound} allow polynomial-time exact evaluation and the computation of a uniform magnitude bound $2^{b_f}$ with $b_f\le|I|^{O(1)}$.
Taking $q$ to be the least common multiple of all coefficient denominators gives the third property, with $\enc(q)=O(|I|)$, since the costs are evaluated only at integer vectors.
The common denominator allows exact binary search over the objective values in polynomially many steps.

Following~\cite[Definition~2.5.2]{dadushPhD}, a \emph{strong separation oracle} for a convex set $K\subseteq\R^\ell$, queried at $\vect y\in\Q^\ell$, either certifies $\vect y\in K$ or returns $\vect c\in\Q^\ell$ such that $\langle \vect c, \vect x\rangle<\langle \vect c, \vect y\rangle$ for every $\vect x\in K$.
The encoding length of its output is bounded by a fixed polynomial in the input and query lengths.

Our support argument extends the following Carath\'eodory-type bound for integer cones due to Eisenbrand and Shmonin~\cite{eisenbrand2006}.

\begin{lemma}[\cite{eisenbrand2006}]
    \label{lem:eisenbrand}
    Let $X \subseteq \Z^d$ be a finite set satisfying $\operatorname{conv}(X) \cap \Z^d = X$.
    Then, for every $\vect y \in \operatorname{int{.}cone}(X)$, there exists a subset $X' \subseteq X$ with $|X'| \le 2^d$ such that $\vect y \in \operatorname{int{.}cone}(X')$.
\end{lemma}
This support bound is asymptotically tight: Exponential support can be necessary even for the special case of high-multiplicity bin packing~\cite{jansen2025esa}.

\section{The Algorithm}\label{sec:algorithm}
We prove \Cref{thm:main} for generators with $d$ integer coordinates and $k$ continuous coordinates.
The overall framework follows the residue-class decomposition of Koana and Kumabe~\cite{koana2026}, but we improve their result in several aspects.
We partition the generator sets into residue classes $\vect r \in \{0, \ldots, 2^d - 1\}^d$ on the integer coordinates and define their convex hulls.
\begin{equation*}
    S^{(t)}_{\vect r}\coloneqq P^{(t)}\cap\bigl((\vect r+2^d\Z^d)\times\R^k\bigr),
    \qquad P^{(t)}_{\vect r}\coloneqq \operatorname{conv}(S^{(t)}_{\vect r}).
\end{equation*}
Since $P^{(t)}$ is convex, $P^{(t)}_{\vect r}\subseteq P^{(t)}$.
Therefore, every point of $P^{(t)}_{\vect r}$ with an integer first block is a valid generator:
\begin{equation*}
    P^{(t)}_{\vect r}\cap(\Z^d\times\R^k)\subseteq S^{(t)}.
\end{equation*}
Such a generator need not itself have residue $\vect r$ modulo $2^d$.

The core insight of~\cite{koana2026} is that this decomposition allows us to formulate an ILP with significantly fewer variables than the standard bin packing ILP.
Instead of assigning a multiplicity variable to every feasible generator (bin configuration), it is sufficient to record the aggregate item counts and the number of generators (bins) for each residue class.
Similarly, we aim to find an aggregate vector $(\vect z^{(t)}_{\vect r},\vect w^{(t)}_{\vect r})$ and a count $n^{(t)}_{\vect r}$, for each type and residue class:
\[
    (\vect z^{(t)}_{\vect r},\vect w^{(t)}_{\vect r})
    =\sum_{(\vect x,\vect u)\in S^{(t)}_{\vect r}}
    \lambda^{(t)}_{(\vect x,\vect u)}(\vect x,\vect u),
    \qquad
    n^{(t)}_{\vect r}
    =\sum_{(\vect x,\vect u)\in S^{(t)}_{\vect r}}
    \lambda^{(t)}_{(\vect x,\vect u)}.
\]
Only $\vect z^{(t)}_{\vect r}$ and the count $n_{\vect r}^{(t)}$ will be required to be integral; $\vect w^{(t)}_{\vect r}\in\R^k$ remains continuous.

We organize this section as follows:
First, we do not show how to find the tuple $(\vect z^{(t)}_{\vect r},\vect w^{(t)}_{\vect r},n^{(t)}_{\vect r})$, but instead define their required properties formally.
In the original problem~\eqref{eq:main-problem}, the constraints in $Q$ and $\mathcal A$ depend only on aggregate vectors and counts, so these tuples retain all information needed to enforce them.
Therefore, the remaining challenge is to recover actual generators from the tuples, without increasing the cost.
In \Cref{sec:construction} we then show how to construct $n^{(t)}_{\vect r}$ generators with exactly the given aggregate $(\vect z^{(t)}_{\vect r}, \vect w^{(t)}_{\vect r})$, and desired cost.
Only later, after restricting the number of active residue classes in \Cref{sec:active-residue-classes}, we show in \Cref{sec:master-formulation} how to find the tuples for all active residue classes and types simultaneously subject to the constraints of~\eqref{eq:main-problem}, using a master ILP formulation.

To search over these triples, we need a convex description of the aggregates and counts that can arise from one residue class.
We obtain it by first allowing fractional multiplicities.
For a positive count $n$, the average generator $(\vect z/n,\vect w/n)$ must lie in $P^{(t)}_{\vect r}$; a zero count forces $(\vect z,\vect w)=(\vect0,\vect0)$.
Thus the feasible aggregate and count triples form the homogenization cone $C^{(t)}_{\vect r}\subseteq\R^{d+k+1}$:
\begin{equation*}
    C^{(t)}_{\vect r}\coloneqq \{(\vect0,\vect0,0)\}\cup
    \left\{(\vect z,\vect w,n)\in\R^d\times\R^k\times\R_{>0}:
    \left(\frac{\vect z}{n},\frac{\vect w}{n}\right)\in P^{(t)}_{\vect r}\right\}.
\end{equation*}
To account for the objective, we associate each triple in this cone with the minimum cost of a fractional representation, which is defined by:
\[
    \min_{\theta_{(\vect x,\vect u)}\ge0}\Biggl\{ \sum_{(\vect x,\vect u)\in S^{(t)}_{\vect r}}
    \theta_{(\vect x,\vect u)}f^{(t)}(\vect x):
    \sum_{(\vect x,\vect u)\in S^{(t)}_{\vect r}} \theta_{(\vect x,\vect u)}(\vect x,\vect u)=(\vect z,\vect w),
    \sum_{(\vect x,\vect u)\in S^{(t)}_{\vect r}} \theta_{(\vect x,\vect u)}=n\Biggr\}.
\]

The problem is that $S^{(t)}_{\vect r}$ may contain infinitely many points, so this is not a finite linear program.
To fix this, observe that for each fixed integer vector $\vect x$, the feasible continuous coordinates form the rational polytope
\[
    \{\vect u\in\R^k:(\vect x,\vect u)\in P^{(t)}\},
\]
called the \emph{fiber} of $\vect x$.
Let $\mathcal V^{(t)}_{\vect r}$ consist of all pairs $(\vect x,\vect u)$ in $S^{(t)}_{\vect r}$ for which $\vect u$ is a vertex of its fiber.
There are finitely many possible integer vectors $\vect x$, and each fiber has finitely many vertices, so $\mathcal V^{(t)}_{\vect r}$ is finite.
Its points have polynomial encoding length by~\eqref{eq:coordinate-bound} and standard vertex bounds for rational polyhedra.
Every generator in $S^{(t)}_{\vect r}$ is a convex combination of points of $\mathcal V^{(t)}_{\vect r}$ with the same integer coordinates and hence the same cost.
In particular, $P^{(t)}_{\vect r}=\operatorname{conv}(\mathcal V^{(t)}_{\vect r})$ is a rational polytope.
Therefore, the minimal fractional cost of $(\vect z,\vect w,n)$ equals the value of the finite linear program
\begin{subequations}\label{eq:cost-envelope-pair}
    \begin{equation}\label{eq:cost-envelope}
        \begin{aligned}
            \Phi^{(t)}_{\vect r}(\vect z,\vect w,n)\coloneqq
            \min_{\theta_{(\vect x,\vect u)}\ge0}\Biggl\{ \sum_{(\vect x,\vect u)\in\mathcal V^{(t)}_{\vect r}} \mspace{-17mu}\theta_{(\vect x,\vect u)}f^{(t)}(\vect x):
            \mspace{-8mu}\sum_{(\vect x,\vect u)\in\mathcal V^{(t)}_{\vect r}} \mspace{-17mu}\theta_{(\vect x,\vect u)}(\vect x,\vect u)=(\vect z,\vect w),
            \mspace{-5mu}\sum_{(\vect x,\vect u)\in\mathcal V^{(t)}_{\vect r}} \mspace{-17mu}\theta_{(\vect x,\vect u)}=n\Biggr\}.
        \end{aligned}
    \end{equation}
    Its variables are the nonnegative coefficients $\theta_{(\vect x,\vect u)}$; the costs $f^{(t)}(\vect x)$ are objective coefficients.
    We set $\Phi^{(t)}_{\vect r}=+\infty$ outside $C^{(t)}_{\vect r}$, and for an empty residue class we set $C^{(t)}_{\vect r}=\{(\vect0,\vect0,0)\}$ and $\Phi^{(t)}_{\vect r}(\vect0,\vect0,0)=0$.
    Note that the set $\mathcal V^{(t)}_{\vect r}$ may still be large but the oracle in \Cref{lem:cost-oracle} computes the required value and representations of~\eqref{eq:cost-envelope} without enumerating all its variables.

    The dual of~\eqref{eq:cost-envelope} has $d+k+1$ variables $(\vect a,b)\in\R^{d+k}\times\R$:
    \begin{equation}\label{eq:envelope-dual}
        \begin{aligned}
            \max\bigl\{ \langle\vect a,(\vect z,\vect w)\rangle+bn:  \langle\vect a,(\vect x,\vect u)\rangle+b\le f^{(t)}(\vect x)
            \quad(\forall(\vect x,\vect u)\in\mathcal V^{(t)}_{\vect r})\bigr\}.
        \end{aligned}
    \end{equation}
\end{subequations}
Let $D^{(t)}_{\vect r}$ denote this dual feasible region.
It is a nonempty rational polyhedron: $\vect a=\vect0$ and sufficiently small $b$ satisfy all its inequalities.
For every $(\vect z,\vect w,n)\in C^{(t)}_{\vect r}$, the primal LP is feasible and has a finite optimal value, so strong duality gives
\[
    \Phi^{(t)}_{\vect r}(\vect z,\vect w,n)
    =\sup_{(\vect a,b)\in D^{(t)}_{\vect r}}
    \bigl(\langle\vect a,(\vect z,\vect w)\rangle+bn\bigr).
\]
Outside $C^{(t)}_{\vect r}$, the primal is infeasible and the dual, being feasible, is unbounded above.
Thus this identity holds throughout $\R^{d+k+1}$, and $\Phi^{(t)}_{\vect r}$ is the support function of $D^{(t)}_{\vect r}$.
Consequently, it is closed, convex, and rational polyhedral~\cite[Corollary~19.2.1]{rockafellar1970}.

The following lemma states that, given a rational point $(\vect z,\vect w,n)\in C^{(t)}_{\vect r}$, one can compute the optimal fractional cost $\Phi^{(t)}_{\vect r}(\vect z,\vect w,n)$ and an optimal solution to~\eqref{eq:cost-envelope} and~\eqref{eq:envelope-dual}.
We omit its proof here to aid readability as it is technical and lengthy, but include it in \Cref{sec:appendix-oracles} for completeness.
\begin{restatable}[\faCut, Fractional Cost Oracle]{lemma}{fractionalCostOracle}\label{lem:cost-oracle}
    Given a nonempty residue class $S^{(t)}_{\vect r}$ and a rational point $(\vect z,\vect w,n)\in C^{(t)}_{\vect r}$, one can compute $\Phi^{(t)}_{\vect r}(\vect z,\vect w,n)$, an optimal dual solution $(\vect a^\star,b^\star)$ to~\eqref{eq:envelope-dual} satisfying
    \[
        \Phi^{(t)}_{\vect r}(\vect z,\vect w,n)=\langle\vect a^\star,(\vect z,\vect w)\rangle+b^\star n,
    \]
    and a fractional representation of $(\vect z,\vect w,n)$ attaining $\Phi^{(t)}_{\vect r}(\vect z,\vect w,n)$, supported on at most $\min\{2^d,d+k+1\}$ rational generators in $S^{(t)}_{\vect r}$, in time
    \[
        2^{O(d^2\log(d+1))}\cdot(|I|+\enc(\vect z,\vect w,n))^{O(1)}.
    \]
    The generators and their multiplicities have polynomial encoding length in the input and query.
\end{restatable}
\begin{proof}[Proof Idea]
    The oracle combines integer pricing with the optimization-separation equivalence of Gr\"otschel, Lov\'asz, and Schrijver~\cite{grotschel1988}.
    A quadratic secondary objective and an adaptation of Eisenbrand and Shmonin's parity argument~\cite{eisenbrand2006} yield the support bound $2^d$.
    The full proof is given in \Cref{sec:appendix-oracles} for completeness.
\end{proof}

\subsection{Reconstructing Generators from a Given Aggregate}\label{sec:construction}

We first solve the local reconstruction problem: given a triple $(\vect z,\vect w,n) \in C^{(t)}_{\vect r}$, recover generators realizing its aggregate and count.
Specifically, we convert an aggregate with an integer first block and an integer count into a sparse representation with integer multiplicities, without increasing its fractional cost.

\begin{restatable}[\faCut, Constructive Residue Decomposition]{lemma}{constructiveResidueDecomposition}\label{lem:construction}
    Given a residue vector $\vect r\in\{0,\dots,2^d-1\}^d$ and a rational point $(\vect z_{\vect r},\vect w_{\vect r},n_{\vect r})\in C^{(t)}_{\vect r}$ with $\vect z_{\vect r}\in\Z^d$ and $n_{\vect r}\in\Z_{\ge0}$, an integer conic decomposition
    \begin{equation*}
        (\vect z_{\vect r},\vect w_{\vect r})
        =\sum_{j=1}^p\lambda_j(\vect x_j,\vect u_j),\qquad
        \sum_{j=1}^p\lambda_j=n_{\vect r},
    \end{equation*}
    with $p\le2^{d+1}-1$ distinct generators $(\vect x_j,\vect u_j)\in P^{(t)}_{\vect r}\cap(\Z^d\times\Q^k)\subseteq S^{(t)}$ and positive integer coefficients $\lambda_j\in\Z_{>0}$, satisfying
    \[
        \sum_{j=1}^p\lambda_j f^{(t)}(\vect x_j)
        \le\Phi^{(t)}_{\vect r}(\vect z_{\vect r},\vect w_{\vect r},n_{\vect r}),
    \]
    can be computed in time $2^{O(d^2\log(d+1))}\cdot(|I|+\enc(\vect z_{\vect r},\vect w_{\vect r},n_{\vect r}))^{O(1)}$.
\end{restatable}
\begin{proof}[Proof sketch]
    We obtain generators with fractional multiplicities from \Cref{lem:cost-oracle}.
    Round down every fractional multiplicity and take the corresponding generator with that multiplicity.
    Afterwards we construct the remaining generators from convex combinations of the generators with truly fractional multiplicities.
    This combination uses the fact that all generators are in the same residue class to ensure that the first \(d\) coordinates remain integral.
    Since the fractional multiplicity from every original generator used in this step is smaller than one, the total number of different generators at most doubles.
    Further, by only using convex combinations the objective value cannot increase.
    A complete self-contained proof is given in \Cref{sec:appendix-construction}.
\end{proof}

\subsection{Bounding the Number of Active Residue Classes}
\label{sec:active-residue-classes}

Having shown how to reconstruct generators from a given triple, we now turn to finding suitable triples jointly for all types.
To achieve our desired running time, we avoid using all $2^{d^2}$ residue classes for each type and instead reduce the number of active residue classes using integer cone Carath\'eodory bounds.
A similar technique was used in~\cite{DBLP:journals/siamdm/Jansen10}.
The following lemma extends the support bound in \Cref{lem:eisenbrand} to preserve the full typewise aggregates $(\vect y^{(t)},\vect v^{(t)})$ and counts $N^{(t)}$ without increasing the objective.
The proof is similar to the original argument in~\cite{eisenbrand2006}, with parity applied only to the integer coordinates and a special potential to cope with the objective function.

\begin{lemma}[Support Bound]\label{lem:sparse-support}
    Let $P^{(1)},\dots,P^{(\tau)}\subseteq\R^{d+k}$ be bounded convex sets, let $S^{(t)}\coloneqq P^{(t)}\cap(\Z^d\times\R^k)$, and let $f^{(t)}:P^{(t)}_{\mathrm{x}}\to\R$ satisfy the discrete Jensen inequality in \Cref{ass:cost-properties} for each $t\in[\tau]$.
    For any finitely supported collection of coefficients $\vect\lambda=(\lambda^{(t)}_{(\vect x,\vect u)})_{t\in[\tau],(\vect x,\vect u)\in S^{(t)}}$ with $\lambda^{(t)}_{(\vect x,\vect u)}\in\Z_{\ge0}$, there exists a finitely supported collection $\vect{\widetilde\lambda}=(\widetilde\lambda^{(t)}_{(\vect x,\vect u)})_{t\in[\tau],(\vect x,\vect u)\in S^{(t)}}$ with $\widetilde\lambda^{(t)}_{(\vect x,\vect u)}\in\Z_{\ge0}$ such that for every type $t\in[\tau]$:
    \begin{enumerate}
        \item $\sum_{(\vect x,\vect u)\in S^{(t)}}
                  \widetilde\lambda^{(t)}_{(\vect x,\vect u)}(\vect x,\vect u)
                  =\sum_{(\vect x,\vect u)\in S^{(t)}}
                  \lambda^{(t)}_{(\vect x,\vect u)}(\vect x,\vect u)$ and $\sum_{(\vect x,\vect u)\in S^{(t)}}
                  \widetilde\lambda^{(t)}_{(\vect x,\vect u)}
                  =\sum_{(\vect x,\vect u)\in S^{(t)}}
                  \lambda^{(t)}_{(\vect x,\vect u)}$,
        \item $\sum_{(\vect x,\vect u)\in S^{(t)}}
                  \widetilde\lambda^{(t)}_{(\vect x,\vect u)}f^{(t)}(\vect x)
                  \le\sum_{(\vect x,\vect u)\in S^{(t)}}
                  \lambda^{(t)}_{(\vect x,\vect u)}f^{(t)}(\vect x)$, and
        \item $|\{(\vect x,\vect u)\in S^{(t)}:
                  \widetilde\lambda^{(t)}_{(\vect x,\vect u)}>0\}|\le2^d$.
    \end{enumerate}
    In particular, the total support across all types satisfies $|\operatorname{supp}(\vect{\widetilde\lambda})|\le\tau2^d$.
\end{lemma}
\begin{proof}
    It suffices to prove the statement independently for each type.
    Fix a type $t\in[\tau]$ and suppress its index, writing $P=P^{(t)}$, $S=S^{(t)}$, $f=f^{(t)}$, and $\lambda_{(\vect x,\vect u)}=\lambda^{(t)}_{(\vect x,\vect u)}$.
    Define the aggregate vectors and total count:
    \[
        (\vect y,\vect v)\coloneqq \sum_{(\vect x,\vect u)\in S}
        \lambda_{(\vect x,\vect u)}(\vect x,\vect u),\qquad
        N\coloneqq \sum_{(\vect x,\vect u)\in S}\lambda_{(\vect x,\vect u)}.
    \]
    Among all finitely supported representations of $(\vect y,\vect v)$ with integer multiplicities, total count $N$, and cost at most $\sum_{(\vect x,\vect u)\in S}\lambda_{(\vect x,\vect u)}f(\vect x)$, choose one minimizing the potential
    \[
        \sum_{(\vect x,\vect u)\in S}\widetilde\lambda_{(\vect x,\vect u)}\|\vect x\|_2^2.
    \]
    Such a minimum exists: boundedness of $P$ makes $X=P_{\mathrm{x}}\cap\Z^d$ finite, and there are only finitely many multisets of $N$ points from $X$.
    Thus the potential takes only finitely many values on the nonempty set of admissible representations.
    If several supported generators have the same integer block $\vect x$, replace them by their average with multiplicity equal to the sum of their original multiplicities.
    This preserves the full aggregate, count, cost, and potential, and the new generator belongs to $S$ by convexity of $P$.
    We may therefore assume that distinct generators in the chosen representation have distinct integer blocks.
    By construction, this representation satisfies conditions~1 and~2.

    It remains to show that its support is bounded by $2^d$.
    For the sake of contradiction, suppose $|\operatorname{supp}(\vect{\widetilde\lambda})|>2^d$.
    By the pigeonhole principle, there exist two supported generators $(\vect x_a,\vect u_a)$ and $(\vect x_b,\vect u_b)$ whose distinct integer blocks have identical coordinate parity: $\vect x_a\equiv\vect x_b\pmod2$.
    Their midpoint
    \[
        (\vect x_c,\vect u_c)
        =\frac12\bigl((\vect x_a,\vect u_a)+(\vect x_b,\vect u_b)\bigr)
    \]
    belongs to $P$ by convexity and has integer first block $\vect x_c=(\vect x_a+\vect x_b)/2$, so $(\vect x_c,\vect u_c)\in S$.
    Now consider replacing one copy of $(\vect x_a,\vect u_a)$ and one copy of $(\vect x_b,\vect u_b)$ with two copies of $(\vect x_c,\vect u_c)$.
    This preserves the full aggregate $(\vect y,\vect v)$ and count $N$, while the discrete Jensen inequality gives $2f(\vect x_c)\le f(\vect x_a)+f(\vect x_b)$.
    However, the potential decreases by
    \[
        \|\vect x_a\|_2^2+\|\vect x_b\|_2^2-2\|\vect x_c\|_2^2
        =\tfrac12\|\vect x_a-\vect x_b\|_2^2>0,
    \]
    where the equality follows from the parallelogram law, contradicting the minimality of $\vect{\widetilde\lambda}$.
    Applying this argument independently to every type and summing the support sizes gives total support of at most $\tau2^d$.
\end{proof}
In particular, if $\vect\lambda$ is a feasible solution to~\eqref{eq:main-problem}, then $\vect{\widetilde\lambda}$ preserves all constraints in $Q$ and $\mathcal A$ without increasing the objective value.

The support bound also limits the number of active residue classes.
Each generator $(\vect x,\vect u)$ belongs to the class $\vect r=\vect x\bmod2^d\in\{0,\dots,2^d-1\}^d$, so any attainable objective budget can be met using a tuple $\mathcal R=(R^{(1)},\dots,R^{(\tau)})$ with $|R^{(t)}|\le2^d$ for every type.

\begin{corollary}[Candidate Support Bound]\label{lem:num-subsets}
    The number of candidate active residue tuples $\mathcal R=(R^{(1)},\dots,R^{(\tau)})$ with $R^{(t)}\subseteq\{0,\dots,2^d-1\}^d$ and $|R^{(t)}|\le2^d$ is at most $2^{\tau2^{O(d)}}$.
\end{corollary}
\begin{proof}
    The number of residue classes per type is $2^{d^2}$.
    Choosing at most $2^d$ active residue classes per type yields
    \[
        \sum_{j=0}^{2^d}\binom{2^{d^2}}{j}
        \le\sum_{j=0}^{2^d}(2^{d^2})^j
        \le(2^d+1)2^{d^2 2^d}
        \le2^{d+1+d^2 2^d}
        =2^{2^{O(d)}}.
    \]
    Choosing an active subset independently for each of the $\tau$ types yields at most
    \[
        \left(2^{2^{O(d)}}\right)^\tau=2^{\tau2^{O(d)}}
    \]
    candidate tuples.
    For $\tau=1$ this is $2^{2^{O(d)}}$.
\end{proof}

\subsection{Combining Aggregates Across Types and Residues}
\label{sec:master-formulation}

We first show how to solve the problem~\eqref{eq:main-problem} under a given bound $\sum_{t=1}^{\tau} N^{(t)} \le N_{\max}$ as in~\eqref{eq:count-bound}.
Whenever $N_{\max}$ is supplied as part of the input, we include its binary encoding length in $|I|$.
After establishing this algorithm, we show how to remove this bound to complete the proof of~\Cref{thm:main}.

For a fixed candidate tuple $\mathcal R = (R^{(1)}, \dots, R^{(\tau)})$, it remains to choose triples in the cones of its active classes so that their sums satisfy the constraints in $Q$ and $\mathcal A$ and their total fractional cost meets the budget.
By \Cref{lem:construction}, any such rational choice with integral $\vect z$ and $n$ can then be converted, block by block, into a solution of the original problem.
We first define a continuous relaxation: for each type $t \in [\tau]$ and active residue $\vect r \in R^{(t)}$, we introduce a variable block $(\vect z^{(t)}_{\vect r},\vect w^{(t)}_{\vect r}, n^{(t)}_{\vect r}) \in \R^d\times\R^k\times\R_{\ge 0}$ representing the aggregate vector and total multiplicity of a fractional combination of generators from $S^{(t)}_{\vect r}$.

Write $\vect z$, $\vect w$, and $\vect n$ for the families of all corresponding blocks.
Together they form a vector $(\vect z,\vect w,\vect n)\in\R^{\widetilde D_{\mathcal R}}$, where
\[
    \widetilde D_{\mathcal R}\coloneqq (d+k+1)\sum_{t=1}^{\tau}|R^{(t)}|.
\]
Only the coordinates in $\vect z$ and $\vect n$ must be integral; their total number is
\[
    D_{\mathcal R}\coloneqq (d+1)\sum_{t=1}^{\tau}|R^{(t)}|
    \le D\coloneqq \tau(d+1)2^d .
\]
The typewise aggregates and the total aggregate are linear combinations of these variables:
\begin{align*}
    (\vect y^{(t)},\vect v^{(t)}) \coloneqq \sum_{\vect r\in R^{(t)}}(\vect z^{(t)}_{\vect r},\vect w^{(t)}_{\vect r}),\quad
    (\vect y,\vect v) & \coloneqq \sum_{t=1}^{\tau}(\vect y^{(t)},\vect v^{(t)}),\quad
    N^{(t)} \coloneqq \sum_{\vect r\in R^{(t)}}n^{(t)}_{\vect r}.
\end{align*}

For a rational objective budget $\gamma$, define the continuous feasible region $K_{\mathcal R}(\gamma) \subseteq \R^{\widetilde D_{\mathcal R}}$ by:
\begin{subequations}\label{eq:master-formulation}
    \begin{align}
        B (\vect y,\vect v)                                                                                                                   & \le \vect c, \label{eq:master-polytope-q}                                                                   \\
        (\vect z^{(t)}_{\vect r},\vect w^{(t)}_{\vect r}, n^{(t)}_{\vect r})                                                                  & \in C^{(t)}_{\vect r} \quad (\forall t \in [\tau], \, \forall \vect r \in R^{(t)}), \label{eq:master-cones} \\
        \big((\vect y^{(t)},\vect v^{(t)}, N^{(t)})\big)_{t=1}^{\tau}                                                                         & \in \mathcal A, \label{eq:master-type-constraints}                                                          \\
        \sum_{t=1}^{\tau} \sum_{\vect r \in R^{(t)}} n^{(t)}_{\vect r}                                                                        & \le N_{\max}, \label{eq:master-count-bound}                                                                 \\
        \sum_{t=1}^{\tau} \sum_{\vect r \in R^{(t)}} \Phi^{(t)}_{\vect r}(\vect z^{(t)}_{\vect r},\vect w^{(t)}_{\vect r}, n^{(t)}_{\vect r}) & \le \gamma. \label{eq:master-budget}
    \end{align}
\end{subequations}
Here, \eqref{eq:master-polytope-q} enforces $(\vect y,\vect v) \in Q$ via its linear description $B(\vect y,\vect v) \le \vect c$, \eqref{eq:master-cones} ensures that each scaled aggregate lies in the corresponding residue hull $P^{(t)}_{\vect r}$, \eqref{eq:master-type-constraints} enforces the typewise linear constraints of $\mathcal A$, \eqref{eq:master-count-bound} imposes the auxiliary count bound~\eqref{eq:count-bound}, and \eqref{eq:master-budget} enforces the cost budget $\gamma$.
Because each cone $C^{(t)}_{\vect r}$ and each cost function $\Phi^{(t)}_{\vect r}$ is rational polyhedral, $K_{\mathcal R}(\gamma)$ is a rational polyhedron and hence a closed convex set.

Define the projected polyhedron $\widehat K_{\mathcal R}(\gamma)\subseteq\R^{D_{\mathcal R}}$ by
\[
    \widehat K_{\mathcal R}(\gamma)
    \coloneqq \{(\vect z,\vect n):\exists\vect w,\ (\vect z,\vect w,\vect n)\in K_{\mathcal R}(\gamma)\}.
\]
We impose integrality by seeking a point in
$\widehat K_{\mathcal R}(\gamma)\cap\Z^{D_{\mathcal R}}$.
A rational lift of such a point specifies the continuous aggregate coordinates.
The following lemma proves soundness and completeness for one type.

\begin{lemma}[One-Type Soundness and Completeness]\label{lem:kr-sound-complete}
    Suppose $\tau=1$, and write $P=P^{(1)}$, $S=P\cap(\Z^d\times\R^k)$, $f=f^{(1)}$, $K_R(\gamma)=K_{(R)}(\gamma)$, and $\widehat K_R(\gamma)=\widehat K_{(R)}(\gamma)$.
    The problem~\eqref{eq:main-problem} subject to~\eqref{eq:count-bound} has a feasible solution with cost at most $\gamma$ if and only if there is a set $R\subseteq\{0,\dots,2^d-1\}^d$ with $|R|\le 2^d$ such that
    \[
        \widehat K_R(\gamma)\cap\Z^{(d+1)|R|}\ne\emptyset.
    \]
\end{lemma}
\begin{proof}
    Suppose first that~\eqref{eq:main-problem} has a feasible solution subject to~\eqref{eq:count-bound} with cost at most $\gamma$.
    By~\Cref{lem:sparse-support}, there exists a representation supported on $p \le 2^d$ generators $(\vect x^{(1)},\vect u^{(1)}), \dots, (\vect x^{(p)},\vect u^{(p)}) \in S$:
    \[
        (\vect y,\vect v) = \sum_{j=1}^p \lambda_j (\vect x^{(j)},\vect u^{(j)}), \qquad N = \sum_{j=1}^p \lambda_j,
    \]
    with $\lambda_j \in \Z_{\ge 1}$, $N \le N_{\max}$, $(\vect y,\vect v, N) \in \mathcal A$, $B(\vect y,\vect v) \le \vect c$, and total cost $\sum_{j=1}^p \lambda_j f(\vect x^{(j)}) \le \gamma$.

    For each $j \in [p]$, let $\vect r(j) \coloneqq  \vect x^{(j)} \bmod 2^d$, and define the active residue set $R \coloneqq  \{\vect r(j) : j \in [p]\}$, with $|R| \le p \le 2^d$.
    For each $\vect r\in R$, set
    \begin{equation*}
        J_{\vect r} \coloneqq \{j\in[p]:\vect r(j) = \vect r\},\qquad
        n_{\vect r}                                                      \coloneqq \sum_{j\in J_{\vect r}}\lambda_j,\qquad
        (\vect z_{\vect r},\vect w_{\vect r})
        \coloneqq \sum_{j\in J_{\vect r}}\lambda_j(\vect x^{(j)},\vect u^{(j)}).
    \end{equation*}
    Then $n_{\vect r} \in \Z_{\ge 1}$ and $\vect z_{\vect r} \in \Z^d$.
    We now verify the constraints of $K_R(\gamma)$, given in~\eqref{eq:master-formulation}:
    First, by construction, $\sum_{\vect r \in R} (\vect z_{\vect r},\vect w_{\vect r}) = \sum_{j=1}^p {\lambda_j (\vect x^{(j)},\vect u^{(j)})} = (\vect y,\vect v) \in Q$, so~\eqref{eq:master-polytope-q} is satisfied.
    Because $(\vect z_{\vect r},\vect w_{\vect r}) / n_{\vect r} = \sum_{j \in J_{\vect r}} (\lambda_j / n_{\vect r}) (\vect x^{(j)},\vect u^{(j)})$ is a convex combination of points in $S_{\vect r}$, we have $(\vect z_{\vect r},\vect w_{\vect r}) / n_{\vect r} \in P_{\vect r}$, establishing $(\vect z_{\vect r},\vect w_{\vect r}, n_{\vect r}) \in C_{\vect r}$, satisfying~\eqref{eq:master-cones}.
    Further, since $\sum_{\vect r \in R} n_{\vect r} = N \leq N_{\max}$ and $(\vect y,\vect v, N) \in \mathcal A$,~\eqref{eq:master-type-constraints} and~\eqref{eq:master-count-bound} are satisfied.
    Finally, we have $\sum_{\vect r \in R} \Phi_{\vect r}(\vect z_{\vect r},\vect w_{\vect r}, n_{\vect r}) \le \sum_{j=1}^p \lambda_j f(\vect x^{(j)}) \le \gamma$, so~\eqref{eq:master-budget} is satisfied.
    Thus $(\vect z,\vect n)\in\widehat K_R(\gamma)\cap\Z^{(d+1)|R|}$.

    Conversely, suppose $\widehat K_R(\gamma)\cap\Z^{(d+1)|R|}\ne\emptyset$.
    Fix an integer point in this projection.
    Since $K_R(\gamma)$ is rational polyhedral, the point has a rational lift
    $(\vect z,\vect w,\vect n)\in K_R(\gamma)$
    with $\vect z_{\vect r}\in\Z^d$ and $n_{\vect r}\in\Z_{\ge0}$.
    A block with $n_{\vect r}=0$ has $(\vect z_{\vect r},\vect w_{\vect r})=(\vect0,\vect0)$ and $\Phi_{\vect r}(\vect0,\vect0,0)=0$, so it contributes no generators.
    For every $n_{\vect r} > 0$, $(\vect z_{\vect r},\vect w_{\vect r},n_{\vect r})\in C_{\vect r}$ gives $(\vect z_{\vect r},\vect w_{\vect r})/n_{\vect r}\in P_{\vect r}$.
    By \Cref{lem:construction}, we obtain rational generators in $P_{\vect r}\cap(\Z^d\times\R^k)\subseteq S$ with aggregate $(\vect z_{\vect r},\vect w_{\vect r})$, total multiplicity $n_{\vect r}$, and cost at most $\Phi_{\vect r}(\vect z_{\vect r},\vect w_{\vect r},n_{\vect r})$.
    Defining $(\vect y,\vect v) \coloneqq  \sum_{\vect r \in R} (\vect z_{\vect r},\vect w_{\vect r})$ and $N \coloneqq  \sum_{\vect r \in R} n_{\vect r}$, merging identical generators across all positive-count blocks yields a valid representation $\vect\lambda \in \Z_{\ge 0}^S$ with aggregate $(\vect y,\vect v)$, count $N \le N_{\max}$, and $(\vect y,\vect v, N) \in \mathcal A$.
    Moreover, $B(\vect y,\vect v) \le \vect c$ guarantees $(\vect y,\vect v) \in Q$, and the total cost satisfies $\sum_{(\vect x,\vect u) \in S} \lambda_{(\vect x,\vect u)} f(\vect x) \le \sum_{\vect r \in R} \Phi_{\vect r}(\vect z_{\vect r},\vect w_{\vect r}, n_{\vect r}) \le \gamma$.
    Hence,~\eqref{eq:main-problem} admits a feasible solution with cost at most $\gamma$.
\end{proof}

The result for multiple types follows immediately from~\Cref{lem:kr-sound-complete} and~\Cref{lem:sparse-support}.

\begin{corollary}[Soundness and Completeness for Multiple Types]\label{cor:master-sound-complete}
    The problem~\eqref{eq:main-problem} subject to~\eqref{eq:count-bound} has a feasible solution with cost at most $\gamma$ if and only if some candidate tuple $\mathcal R$ satisfies
    \[
        \widehat K_{\mathcal R}(\gamma)\cap\Z^{D_{\mathcal R}}\ne\emptyset.
    \]
\end{corollary}
\begin{proof}
    For the forward implication, let $\vect\lambda$ be a feasible solution to~\eqref{eq:main-problem} satisfying~\eqref{eq:count-bound} with cost at most $\gamma$.
    Applying~\Cref{lem:sparse-support} yields a solution supported on at most $2^d$ generators per type while preserving every typewise aggregate $(\vect y^{(t)},\vect v^{(t)})$ and count $N^{(t)}$ without increasing the objective.
    Grouping the generators by residue within each type $t \in [\tau]$ as in~\Cref{lem:kr-sound-complete} defines a candidate tuple $\mathcal R = (R^{(1)}, \dots, R^{(\tau)})$ with $|R^{(t)}| \le 2^d$ and an integer point in $\widehat K_{\mathcal R}(\gamma)$.

    Conversely, lift an integer point of $\widehat K_{\mathcal R}(\gamma)$ to a rational point of $K_{\mathcal R}(\gamma)$, as above.
    Apply \Cref{lem:construction} independently to each positive-count block and omit zero-count blocks.
    Because each typewise aggregate $(\vect y^{(t)},\vect v^{(t)})$ and count $N^{(t)}$ is preserved, all constraints in $Q$, $\mathcal A$, and~\eqref{eq:count-bound} are satisfied.
    Summing the cost bounds over all blocks ensures that the total objective is at most $\sum_{t, \vect r} \Phi^{(t)}_{\vect r}(\vect z^{(t)}_{\vect r},\vect w^{(t)}_{\vect r}, n^{(t)}_{\vect r}) \le \gamma$.
\end{proof}

It remains to implement integer feasibility of $\widehat K_{\mathcal R}(\gamma)$ using separation oracles.

\subsection{Solving the Formulation via Separation Oracles}\label{sec:separation}

For a fixed candidate tuple $\mathcal R=(R^{(1)},\dots,R^{(\tau)})$ and rational budget $\gamma$, we decide whether $\widehat K_{\mathcal R}(\gamma)$ contains an integer point and, if so, recover a rational lift in $K_{\mathcal R}(\gamma)$.
Dadush's algorithm~\cite[Theorem~7.1.1]{dadushPhD} finds an integer point in a convex set with a known bounding box using a strong separation oracle, or certifies that none exists.
We construct separation first for the full formulation and then for its projection, so the integer algorithm operates in $D_{\mathcal R}$ dimensions, independently of $k$.

Using the coordinate bound $M$ from~\eqref{eq:coordinate-bound}, every point $(\vect x,\vect u) \in P^{(t)}_{\vect r}$ satisfies $\|(\vect x,\vect u)\|_\infty\le M$, since $P^{(t)}_{\vect r}\subseteq P^{(t)}$.
A point $(\vect z^{(t)}_{\vect r},\vect w^{(t)}_{\vect r}, n^{(t)}_{\vect r})\in C^{(t)}_{\vect r}$ with $n^{(t)}_{\vect r}>0$ satisfies $(\vect z^{(t)}_{\vect r},\vect w^{(t)}_{\vect r})/n^{(t)}_{\vect r}\in P^{(t)}_{\vect r}$; if $n^{(t)}_{\vect r}=0$, then $(\vect z^{(t)}_{\vect r},\vect w^{(t)}_{\vect r})=(\vect0,\vect0)$.
Thus every cone point satisfies the coordinate bounds
\begin{equation}\label{eq:zero-count-cuts}
    -M n^{(t)}_{\vect r}  \le z^{(t)}_{\vect r,j}\le M n^{(t)}_{\vect r}
    \quad \forall j\in[d], \text{ and}\qquad
    -M n^{(t)}_{\vect r}  \le w^{(t)}_{\vect r,j}\le M n^{(t)}_{\vect r}\quad \forall j\in[k].
\end{equation}
The following lemma shows that these bounds, together with the count bound~\eqref{eq:count-bound}, bound every coordinate in the full formulation and its projection.

\begin{lemma}[Bounding Box]\label{lem:bounding-box}
    Every point $(\vect z,\vect w, n)\in K_{\mathcal R}(\gamma)$ satisfies
    \[
        \max\big( \|(\vect z^{(t)}_{\vect r},\vect w^{(t)}_{\vect r})\|_\infty, \, n^{(t)}_{\vect r} \big) \le T \coloneqq  \max\{1, N_{\max} M\}, \text{ with}\qquad \enc(T)=|I|^{O(1)}.
    \]
\end{lemma}
\begin{proof}
    The non-negative counts satisfy $n^{(t)}_{\vect r} \le \sum_{t', \vect r'} n^{(t')}_{\vect r'} \le N_{\max}$ by~\eqref{eq:master-count-bound}.
    The bounds in~\eqref{eq:zero-count-cuts} then imply $\|(\vect z^{(t)}_{\vect r},\vect w^{(t)}_{\vect r})\|_\infty \le M n^{(t)}_{\vect r} \le M N_{\max} \le T$.
    The encoding bound follows because $M$ and $N_{\max}$ have an encoding length of at most $|I|^{O(1)}$.
\end{proof}

By \Cref{lem:bounding-box}, the full formulation lies in $[-T,T]^{\widetilde D_{\mathcal R}}$ and its projection lies in $[-T,T]^{D_{\mathcal R}}$.
To construct the required oracles, we combine separation over the residue hulls with the fractional-cost oracle from \Cref{lem:cost-oracle}.

For the next lemma, fix a type $t\in[\tau]$ and write $P=P^{(t)}$ and $S=S^{(t)}$, suppressing the type index on residue sets and hulls.

\begin{restatable}[Strong Separation for Residue Hulls]{lemma}{residueHullOracle}\label{lem:pr-oracle}
    For $\vect r\in\{0,\dots,2^d-1\}^d$, a strong separation oracle for $P_{\vect r}=\operatorname{conv}(S_{\vect r})$ can be implemented deterministically in time
    \begin{equation*}
        2^{O(d^2\log(d+1))} \cdot (|I| + \enc(\vect x,\vect u))^{O(1)}
    \end{equation*}
    per query $(\vect x,\vect u)\in\Q^{d+k}$.
    It either certifies $(\vect x,\vect u)\in P_{\vect r}$ or returns $\vect c\in\Q^{d+k}$ and $\beta\in\Q$ satisfying
    \[
        \langle \vect c, (\vect x',\vect u')\rangle\le\beta<\langle \vect c, (\vect x,\vect u)\rangle
        \qquad\text{for every }(\vect x',\vect u')\in P_{\vect r}.
    \]
\end{restatable}
\begin{proof}
    For $\vect a\in\Q^{d+k}$, linear optimization over $P_{\vect r}$ is equivalent to linear optimization over $S_{\vect r}$.
    Applying \Cref{lem:integer-pricing} with zero cost and pricing vector $\vect a$ therefore computes
    \[
        \max\{\langle\vect a,(\vect x',\vect u')\rangle:(\vect x',\vect u')\in S_{\vect r}\}
    \]
    in time $2^{O(d^2\log(d+1))}(\enc(P)+\enc(\vect a))^{O(1)}$ and detects an empty residue class.
    In the empty case, $\vect c=\vect e_1$ and $\beta=x_1-1$ give a valid separator for a query $(\vect x,\vect u)$, where $\vect e_1$ is the first coordinate vector.

    Otherwise, $P_{\vect r}=\operatorname{conv}(\mathcal V_{\vect r})$, where $\mathcal V_{\vect r}$ is the finite set of fiber vertices defined in \Cref{sec:algorithm}.
    These vertices have polynomial encoding length, so $P_{\vect r}$ has a known polynomial bound on its facet complexity~\cite[Lemma~6.2.4]{grotschel1988}.
    The equivalence of strong optimization and separation~\cite[Theorem~6.4.9]{grotschel1988} therefore certifies $(\vect x,\vect u)\in P_{\vect r}$ or returns a strictly separating normal $\vect c$, using polynomially many optimization queries of polynomial encoding length.
    In the latter case, one further optimization computes
    \[
        \beta\coloneqq \max_{(\vect x',\vect u')\in S_{\vect r}}\langle\vect c,(\vect x',\vect u')\rangle
        =\max_{(\vect x',\vect u')\in P_{\vect r}}\langle\vect c,(\vect x',\vect u')\rangle
        <\langle\vect c,(\vect x,\vect u)\rangle.
    \]
    This proves the separator guarantee and the claimed running time.
\end{proof}

\begin{lemma}[Integer Feasibility over $\widehat K_{\mathcal R}(\gamma)$~\cite{dadushPhD}]\label{lem:sep-kr}
    One can deterministically decide whether $\widehat K_{\mathcal R}(\gamma)$ contains an integer point and, if so, find one together with a rational lift in $K_{\mathcal R}(\gamma)$, in time
    \[
        2^{\tau\log(\tau+1)\,2^{O(d)}}
        \cdot (|I|+\enc(\gamma))^{O(1)}.
    \]
\end{lemma}
\begin{proof}
    If $R^{(t)}=\emptyset$ for every $t\in[\tau]$, all typewise aggregates and counts are zero, and the objective value is zero.
    This candidate is feasible precisely when
    \[
        \vect0\in Q,\qquad
        \vect0\in\mathcal A,\qquad
        \gamma\ge0.
    \]
    These conditions can be checked directly in polynomial time.
    In the feasible case, return the empty representation.
    From now on, assume $\sum_t|R^{(t)}|\ge1$, so that $D_{\mathcal R}\ge d+1$.

    By \Cref{lem:bounding-box}, the closed convex set $K \coloneqq  K_{\mathcal R}(\gamma)$ is contained in the bounding box $[-T, T]^{\widetilde D_{\mathcal R}}$ with $\enc(T) = |I|^{O(1)}$.
    We first construct a strong separation oracle for $K$.
    Let $\bar{\vect\xi}\coloneqq (\bar{\vect z},\bar{\vect w},\bar{\vect n})\in\Q^{\widetilde D_{\mathcal R}}$ be a rational query point.
    We initially evaluate the explicit linear constraints: the bounding box $[-T, T]^{\widetilde D_{\mathcal R}}$, the non-negativity conditions $\bar n^{(t)}_{\vect r} \ge 0$, the coordinate bounds~\eqref{eq:zero-count-cuts}, the total-count bound~\eqref{eq:master-count-bound}, and the linear systems defining $Q$ and $\mathcal A$.
    Any violated inequality directly yields a separating hyperplane.
    In addition, for any residue class that is empty ($S^{(t)}_{\vect r} = \emptyset$, precomputed upfront), \eqref{eq:master-cones} forces $(\bar{\vect z}^{(t)}_{\vect r},\bar{\vect w}^{(t)}_{\vect r}, \bar n^{(t)}_{\vect r}) = (\vect0,\vect0,0)$; any non-zero coordinate in such a block immediately provides a coordinate separator.

    If all explicit linear constraints hold, we test cone membership for each block $(\bar{\vect z}^{(t)}_{\vect r},\bar{\vect w}^{(t)}_{\vect r}, \bar n^{(t)}_{\vect r})$.
    If $\bar n^{(t)}_{\vect r} = 0$, then $(\bar{\vect z}^{(t)}_{\vect r},\bar{\vect w}^{(t)}_{\vect r})=(\vect0,\vect0)$.
    Otherwise, we query \Cref{lem:pr-oracle} at the scaled vector $(\bar{\vect x}^{(t)}_{\vect r},\bar{\vect u}^{(t)}_{\vect r}) \coloneqq  (\bar{\vect z}^{(t)}_{\vect r},\bar{\vect w}^{(t)}_{\vect r}) / \bar n^{(t)}_{\vect r}$.
    If $(\bar{\vect x}^{(t)}_{\vect r},\bar{\vect u}^{(t)}_{\vect r}) \notin P^{(t)}_{\vect r}$, the oracle returns a separator $\langle \vect c, (\vect x,\vect u)\rangle \le \beta$ valid on $P^{(t)}_{\vect r}$ such that $\langle \vect c, (\bar{\vect x}^{(t)}_{\vect r},\bar{\vect u}^{(t)}_{\vect r})\rangle > \beta$.
    Homogenizing this inequality yields
    \[
        \langle \vect c, (\vect z^{(t)}_{\vect r},\vect w^{(t)}_{\vect r})\rangle - \beta n^{(t)}_{\vect r} \le 0,
    \]
    which is valid for every cone point: divide by the count when it is positive, and use the zero aggregate when the count is zero.
    It is strictly violated at the query point because
    \[
        \langle\vect c,(\bar{\vect z}^{(t)}_{\vect r},\bar{\vect w}^{(t)}_{\vect r})\rangle-\beta\bar n^{(t)}_{\vect r}
        =\bar n^{(t)}_{\vect r}\bigl(\langle\vect c,(\bar{\vect x}^{(t)}_{\vect r},\bar{\vect u}^{(t)}_{\vect r})\rangle-\beta\bigr)>0.
    \]

    Finally, if $\bar{\vect\xi}$ satisfies all cone constraints, we verify the budget constraint~\eqref{eq:master-budget}.
    For each block with $S^{(t)}_{\vect r}\ne\emptyset$, we evaluate the fractional cost $\Phi^{(t)}_{\vect r}(\bar{\vect z}^{(t)}_{\vect r},\bar{\vect w}^{(t)}_{\vect r}, \bar n^{(t)}_{\vect r})$ and obtain an optimal dual certificate $(\vect a^{(t)}_{\vect r}, b^{(t)}_{\vect r})$ using \Cref{lem:cost-oracle}.
    An empty residue class has a zero block and contributes zero cost; for it we set $(\vect a^{(t)}_{\vect r},b^{(t)}_{\vect r})=(\vect0,0)$.
    If the budget inequality~\eqref{eq:master-budget} is satisfied, the oracle accepts $\bar{\vect\xi}$.
    Otherwise, it returns the linear cut
    \begin{equation*}
        \sum_{t=1}^{\tau} \sum_{\vect r \in R^{(t)}} \left( \langle \vect a^{(t)}_{\vect r}, (\vect z^{(t)}_{\vect r},\vect w^{(t)}_{\vect r})\rangle + b^{(t)}_{\vect r} n^{(t)}_{\vect r} \right) \le \gamma.
    \end{equation*}
    By dual feasibility~\eqref{eq:envelope-dual}, $\langle \vect a^{(t)}_{\vect r}, (\vect z,\vect w)\rangle + b^{(t)}_{\vect r} n \le \Phi^{(t)}_{\vect r}(\vect z,\vect w, n)$ throughout $C^{(t)}_{\vect r}$, making this inequality valid for all points in $K$.
    Because each dual bound is tight at the query point $\bar{\vect\xi}$, the left-hand side evaluates to $\sum_{t, \vect r} \Phi^{(t)}_{\vect r}(\bar{\vect z}^{(t)}_{\vect r},\bar{\vect w}^{(t)}_{\vect r}, \bar n^{(t)}_{\vect r}) > \gamma$, strictly separating $\bar{\vect\xi}$ from $K$.

    Since $\sum_{t=1}^{\tau}|R^{(t)}|\le\tau 2^d$, a query to this full separation oracle uses at most $\tau 2^d$ invocations of \Cref{lem:pr-oracle,lem:cost-oracle} and takes time
    \[
        \tau 2^d\,2^{O(d^2\log(d+1))}
        (|I|+\enc(\gamma)+\enc(\bar{\vect\xi}))^{O(1)}.
    \]
    The full polyhedron $K$ has a rational inequality description of encoding length per inequality bounded by a polynomial in $\widetilde D_{\mathcal R}+|I|+\enc(\gamma)$.
    Indeed, the residue hulls and cost envelopes have such bounds by their finite fiber-vertex descriptions and LP duality; combining at most $\tau 2^d$ blocks preserves a polynomial bound in these quantities.
    Standard rational-polyhedral projection bounds give the same type of bound for $\widehat K\coloneqq \widehat K_{\mathcal R}(\gamma)$.

    The optimization--separation equivalence~\cite[Theorem~6.4.9]{grotschel1988} therefore supplies a strong optimization oracle for $K$.
    Strong linear optimization over $\widehat K$ reduces to strong linear optimization over $K$: give the discarded continuous coordinates objective coefficient zero.
    Another appliction of~\cite{grotschel1988} gives a strong separation oracle for $\widehat K$ from the strong optimization oracle for $\widehat K$, with only polynomial overhead in $\widetilde D_{\mathcal R}$, the input and budget lengths, and the query length.
    Dadush's algorithm~\cite[Theorem~7.1.1]{dadushPhD} now determines whether
    $\widehat K\cap\Z^{D_{\mathcal R}}\ne\emptyset$ and, if so, returns an integer point, using
    \[
        2^{O(D_{\mathcal R}\log(D_{\mathcal R}+2))}
            (\widetilde D_{\mathcal R}+|I|+\enc(\gamma))^{O(1)}
    \]
    arithmetic operations and calls to the projected oracle, with query lengths polynomial in the same quantities.
    Fixing the returned integer coordinates in $K$ gives a nonempty rational polyhedron.
    Exact linear programming with the full separation oracle recovers a rational point of this fiber whose encoding length is polynomial in $\widetilde D_{\mathcal R}+|I|+\enc(\gamma)$.

    Since $D_{\mathcal R}\le\tau(d+1)2^d$, we have
    \[
        D_{\mathcal R}\log(D_{\mathcal R}+2)
        =O\bigl(\tau d2^d(d+\log(\tau+1))\bigr)
        \le\tau\log(\tau+1)\,2^{O(d)}.
    \]
    The dependence on $\widetilde D_{\mathcal R}\le\tau(d+k+1)2^d$ is polynomial, and $k$ is bounded by the input length.
    Combining these bounds with the separation time gives the claimed running time
    \[
        2^{\tau\log(\tau+1)\,2^{O(d)}}
            (|I|+\enc(\gamma))^{O(1)}. \qedhere
    \]
\end{proof}

The randomized algorithm of Reis and Rothvoss~\cite{reis2023} can likewise be applied to the projected lattice-feasibility problems over $\widehat K_{\mathcal R}(\gamma)$.
This gives a randomized variant with the same stated double-exponential dependence on $d$.

\subsection{Finding the Optimal Cost}\label{sec:exact-optimization}
By \Cref{ass:cost-properties}, we have a computable integer $q\ge1$ with $\enc(q)=|I|^{O(1)}$ such that every generator has cost in $q^{-1}\Z$, because its cost depends only on its integer coordinates. Every representation with integer multiplicities has its objective in the same lattice.
This allows us to optimize exactly by searching integer budgets for the scaled objective.

We first complete the algorithm with a supplied multiplicity bound.
\begin{proposition}[Optimization with a Multiplicity Bound]\label{prop:counted-optimization}
    Given $N_{\max}\in\Z_{\ge0}$, one can solve~\eqref{eq:main-problem} subject to~\eqref{eq:count-bound} deterministically in time
    \[
        2^{\tau\log(\tau+1)\,2^{O(d)}}\cdot |I|^{O(1)}.
    \]
    The algorithm reports infeasibility or returns an optimum with at most $2^{2d+1}$ generators per type.
\end{proposition}

\begin{proof}
    Fix a rational objective budget $\gamma$.
    Enumerate the candidate tuples $\mathcal{R}$ from \Cref{lem:num-subsets}.
    For each tuple, use \Cref{lem:sep-kr} to find an integer point in $\widehat K_{\mathcal R}(\gamma)$ and a rational lift in $K_{\mathcal R}(\gamma)$, or certify that no such integer point exists.
    When a lift is found, apply \Cref{lem:construction} to its blocks, obtaining at most $2^{d+1}-1$ generators per residue class and type.
    There are at most $2^d$ active classes per type, so at most $2^d(2^{d+1}-1)<2^{2d+1}$ generators per type are returned.
    By~\Cref{cor:master-sound-complete}, a feasible representation of cost at most $\gamma$ exists if and only if at least one projected candidate contains an integer point.

    By \Cref{lem:sep-kr} and~\Cref{lem:num-subsets}, solving all candidate subproblems for one budget takes time
    \[
        2^{\tau\log(\tau+1)\,2^{O(d)}}
            (|I|+\enc(\gamma))^{O(1)},
    \]
    since the enumeration factor $2^{\tau 2^{O(d)}}$ is absorbed into this bound.
    The returned rational lift has encoding length polynomial in $\widetilde D_{\mathcal R}+|I|+\enc(\gamma)$, as established in \Cref{lem:sep-kr}.
    Applying \Cref{lem:construction} to at most $\tau 2^d$ blocks and merging within types therefore takes $2^{\tau\log(\tau+1)2^{O(d)}}(|I|+\enc(\gamma))^{O(1)}$ time as well.

    It remains to find the optimal budget exactly.
    By \Cref{ass:cost-properties}, every generator cost has absolute value at most $2^{b_f}$.
    With total count at most $N_{\max}$, a valid absolute bound on every objective is
    \[
        V\coloneqq N_{\max}2^{b_f}.
    \]
    Test budget $V$ to detect infeasibility; otherwise binary-search the integers $b\in[-qV,qV]$ for the smallest feasible budget $\gamma=b/q$.
    Since every attainable objective lies in $q^{-1}\Z$, a witness for the smallest feasible budget is optimal.
    The number of iterations and the budget encoding lengths are $O(\log(V+2)+\log q)=|I|^{O(1)}$, so binary search preserves the claimed running-time bound.
\end{proof}

\subsection{Removing the Multiplicity Bound and Detecting Unboundedness}
\label{sec:main-proof}

Combine the constraints of $Q$ and $\mathcal A$ into a system
\[
    W\vect y_{\mathrm{agg}}\le\vect c_{\mathrm{agg}},
    \qquad
    \vect y_{\mathrm{agg}}\coloneqq \big((\vect y^{(t)},\vect v^{(t)},N^{(t)})\big)_{t=1}^{\tau}.
\]
Clear denominators in this system and in the local descriptions of the polytopes $P^{(t)}$.
Choose a computable integer $U\ge1$ bounding all resulting coefficients and right-hand sides, including the coefficients obtained by substituting an integer-coordinate vector $\vect x\in[-M,M]^d$ into a local or aggregate constraint.
Only the integer coordinates are substituted; the continuous coordinates remain variables.
For example, a local row $\vect a^\top\vect x+\vect b^\top\vect u\le c$ contributes a coefficient of absolute value at most
\[
    |\vect a^\top\vect x-c|\le M\|\vect a\|_1+|c|.
\]
The analogous bound applies to the integer-coordinate and count columns of each type block of $W$.
Thus $U$ is computed without enumerating generators and satisfies $\log U=|I|^{O(1)}$.
Set
\begin{equation}\label{eq:automatic-count-bound}
    s\coloneqq \tau 2^d,\qquad R\coloneqq sk,\qquad
    H\coloneqq (R+1)!U^{R+1},\qquad
    N_\star\coloneqq s(s+1)s!H^s.
\end{equation}
The factor $H$ bounds the coefficients after eliminating at most $R$ continuous aggregate variables.
For $k=0$, it reduces to $H=U$.

\begin{restatable}[\faCut, Uniform Multiplicity Bound]{lemma}{uniformMultiplicityBound}\label{lem:uniform-count-bound}
    Every feasible instance has a feasible representation of total multiplicity at most $N_\star$.
    If its objective is bounded below, it has an optimal representation with the same bound.
    Moreover, if $W\vect y_{\mathrm{agg}}\le\vect0$ admits a generator representation with integer multiplicities and negative cost, it admits one of total multiplicity at most $N_\star$.
\end{restatable}
\begin{proof}[Proof Sketch]
    By \Cref{lem:sparse-support}, it suffices to fix at most $s$ type--integer-coordinate pairs.
    For each pair, retain its count and its continuous aggregate as variables.
    Eliminating the continuous aggregates gives a rational polyhedron in the counts with integral coefficients bounded by $H$.
    Cramer's rule and integer-programming proximity then bound the counts of an optimum.
    A normalized recession polytope similarly yields a bounded negative-cost representation.
    The coefficient bounds and the full proof are given in \Cref{sec:appendix-multiplicity}.
\end{proof}

\begin{proof}[Proof of \Cref{thm:main}]
    Compute $N_\star$ from~\eqref{eq:automatic-count-bound} and apply \Cref{prop:counted-optimization} to the original instance with $N_{\max}=N_\star$.
    If this instance is infeasible, \Cref{lem:uniform-count-bound} shows that the original instance is infeasible as well.
    Otherwise retain its optimal representation $\vect\lambda^0$.

    Next replace the right-hand sides in the descriptions of $Q$ and $\mathcal A$ by zero, obtaining polyhedra $\widehat Q$ and $\widehat{\mathcal A}$.
    Apply \Cref{prop:counted-optimization} to this altered instance, with the same generator domains and costs and again $N_{\max}=N_\star$.
    The zero representation is feasible for this instance.

    If the returned optimum $\widehat{\vect\rho}$ has negative cost, then $\vect\lambda^0+\ell\widehat{\vect\rho}$ is feasible for the original instance for every $\ell\in\Z_{\ge0}$ and its cost tends to $-\infty$.
    Hence the two returned representations certify unboundedness below.

    Otherwise, \Cref{lem:uniform-count-bound} excludes every negative-cost homogeneous representation, even without a multiplicity bound.
    \Cref{sec:appendix-multiplicity} shows that an unbounded feasible instance would admit such a representation.
    The original objective is therefore bounded below, and \Cref{lem:uniform-count-bound} implies that $\vect\lambda^0$ is globally optimal.
    Its support satisfies the bound in \Cref{prop:counted-optimization}.

    Finally,
    \[
        \log H=O\bigl((sk+1)(\log(sk+1)+\log U)\bigr),
        \qquad
        \enc(N_\star)=s^{O(1)}|I|^{O(1)}
        =\tau^{O(1)}2^{O(d)}|I|^{O(1)}.
    \]
    Forming the altered instance does not increase the input length.
    Substituting this encoding bound into the running time of the two calls to \Cref{prop:counted-optimization} proves the theorem.
\end{proof}

\section{Declaration of Generative AI Use}
GPT-6 Astra and Gemini 3.8 Flash were used in formulating theorems and proofs, language editing, notation review, and literature searches.
The authors had the ideas of using a support bound to obtain an ETH-tight running time and extending the formulation to include a cost function.
GPT-6 Astra was used to work out the details and assist in proving the resulting statements.

In a different context, GPT-6 Astra suggested a framework similar to that of~\cite{koana2026} for obtaining additive approximations for scheduling problems.
This led us to recognize that our algorithm also solves these problems.
The authors independently verified the mathematical proofs for correctness and take full responsibility for the content of this paper.

\printbibliography

@phdthesis{dadushPhD,
  author             = {Daniel Nicolas Dadush},
  title              = {Integer Programming, Lattice Algorithms, and Deterministic Volume Estimation},
  school             = {Georgia Institute of Technology},
  year               = {2012},
  timestamp          = {Tue, 29 Apr 2025 01:00:00 +0200},
  biburl             = {https://dblp.org/rec/phd/basesearch/Dadush12.bib},
  bibsource          = {dblp computer science bibliography, https://dblp.org},
  url                = {http://hdl.handle.net/1853/44807},
  _bib2doi_selected  = {dblp:/rec/phd/basesearch/Dadush12.bib},
  _bib2doi_confirmed = {true},
  _bib2doi_finished  = {true}
}

@article{eisenbrand2006,
  author             = {Friedrich Eisenbrand and Gennady Shmonin},
  title              = {Carath{\'e}odory bounds for integer cones},
  journal            = {Operations Research Letters},
  volume             = {34},
  number             = {5},
  year               = {2006},
  doi                = {10.1016/j.orl.2005.09.008},
  timestamp          = {Wed, 14 Nov 2018 00:00:00 +0100},
  biburl             = {https://dblp.org/rec/journals/orl/EisenbrandS06.bib},
  bibsource          = {dblp computer science bibliography, https://dblp.org},
  _bib2doi_selected  = {dblp:/rec/journals/orl/EisenbrandS06.bib},
  _bib2doi_confirmed = {true}
}

@article{goemans2020,
  author             = {Michel X. Goemans and Thomas Rothvoss},
  title              = {Polynomiality for Bin Packing with a Constant Number of Item Types},
  journal            = {Journal of the ACM},
  volume             = {67},
  number             = {6},
  year               = {2020},
  doi                = {10.1145/3421750},
  timestamp          = {Sat, 15 Aug 2026 01:00:00 +0200},
  biburl             = {https://dblp.org/rec/journals/jacm/GoemansR20.bib},
  bibsource          = {dblp computer science bibliography, https://dblp.org},
  _bib2doi_selected  = {dblp:/rec/journals/jacm/GoemansR20.bib},
  _bib2doi_confirmed = {true}
}

@book{grotschel1988,
  author             = {Martin Gr{\"o}tschel and L{\'a}szl{\'o} Lov{\'a}sz and Alexander Schrijver},
  title              = {Geometric Algorithms and Combinatorial Optimization},
  publisher          = {Springer-Verlag},
  year               = {1988},
  doi                = {10.1007/978-3-642-97881-4},
  timestamp          = {Tue, 06 Aug 2019 01:00:00 +0200},
  biburl             = {https://dblp.org/rec/books/sp/GLS1988.bib},
  bibsource          = {dblp computer science bibliography, https://dblp.org},
  _bib2doi_selected  = {dblp:/rec/books/sp/GLS1988.bib},
  _bib2doi_confirmed = {true}
}

@inproceedings{jansen2025ciac,
  author             = {Klaus Jansen and Kai Kahler and Esther Zwanger},
  title              = {Exact and Approximate High-Multiplicity Scheduling on Identical Machines},
  booktitle          = {Algorithms and Complexity (CIAC)},
  series             = {Lecture Notes in Computer Science},
  publisher          = {Springer},
  year               = {2025},
  doi                = {10.1007/978-3-031-92932-8_1},
  timestamp          = {Wed, 11 Jun 2025 01:00:00 +0200},
  _bib2doi_selected  = {dblp:/rec/conf/ciac/JansenKZ25.bib},
  _bib2doi_confirmed = {true},
  biburl             = {https://dblp.org/rec/conf/ciac/JansenKZ25.bib},
  bibsource          = {dblp computer science bibliography, https://dblp.org}
}

@inproceedings{jansen2025esa,
  author             = {Klaus Jansen and Lis Pirotton and Malte Tutas},
  title              = {The Support of Bin Packing Is Exponential},
  booktitle          = {33rd Annual European Symposium on Algorithms (ESA)},
  year               = {2025},
  doi                = {10.4230/LIPIcs.ESA.2025.48},
  timestamp          = {Thu, 02 Oct 2025 01:00:00 +0200},
  biburl             = {https://dblp.org/rec/conf/esa/JansenPT25.bib},
  bibsource          = {dblp computer science bibliography, https://dblp.org},
  _bib2doi_selected  = {dblp:/rec/conf/esa/JansenPT25.bib},
  _bib2doi_confirmed = {true}
}

@inproceedings{jansen2026,
  author             = {Klaus Jansen and Felix Ohnesorge and Lis Pirotton},
  title              = {A Tight Double-Exponential Lower Bound for High-Multiplicity Bin Packing},
  booktitle          = {53rd International Colloquium on Automata, Languages, and Programming (ICALP)},
  year               = {2026},
  doi                = {10.4230/LIPIcs.ICALP.2026.116},
  timestamp          = {Mon, 06 Jul 2026 01:00:00 +0200},
  biburl             = {https://dblp.org/rec/conf/icalp/JansenOP26.bib},
  bibsource          = {dblp computer science bibliography, https://dblp.org},
  _bib2doi_selected  = {dblp:/rec/conf/icalp/JansenOP26.bib},
  _bib2doi_confirmed = {true}
}

@article{jansenKlein2020,
  author             = {Klaus Jansen and Kim-Manuel Klein},
  title              = {About the Structure of the Integer Cone and Its Application to Bin Packing},
  journal            = {Mathematics of Operations Research},
  volume             = {45},
  number             = {4},
  year               = {2020},
  doi                = {10.1287/moor.2019.1040},
  timestamp          = {Thu, 16 Sep 2021 01:00:00 +0200},
  biburl             = {https://dblp.org/rec/journals/mor/JansenK20.bib},
  bibsource          = {dblp computer science bibliography, https://dblp.org},
  _bib2doi_selected  = {dblp:/rec/journals/mor/JansenK20.bib},
  _bib2doi_confirmed = {true}
}

@article{knop2021,
  author             = {Du{\v{s}}an Knop and Martin Kouteck{\'y} and Asaf Levin and Matthias Mnich and Shmuel Onn},
  title              = {Parameterized complexity of configuration integer programs},
  journal            = {Operations Research Letters},
  volume             = {49},
  number             = {6},
  year               = {2021},
  doi                = {10.1016/j.orl.2021.11.005},
  timestamp          = {Tue, 21 Mar 2023 00:00:00 +0100},
  biburl             = {https://dblp.org/rec/journals/orl/KnopKLMO21.bib},
  bibsource          = {dblp computer science bibliography, https://dblp.org},
  _bib2doi_selected  = {dblp:/rec/journals/orl/KnopKLMO21.bib},
  _bib2doi_confirmed = {true}
}

@misc{koana2026,
  author            = {Tomohiro Koana and Soh Kumabe},
  title             = {High-Multiplicity Bin Packing is {FPT}},
  year              = {2026},
  eprint            = {2609.16923},
  doi               = {10.48550/arXiv.2609.16923},
  archiveprefix     = {arXiv},
  primaryclass      = {cs.DS},
  _bib2doi_finished = {true}
}

@inproceedings{kowalik2024,
  author             = {{\L}ukasz Kowalik and Alexandra Lassota and Konrad Majewski and Micha{\l} Pilipczuk and Marek Soko{\l}owski},
  title              = {Detecting Points in Integer Cones of Polytopes Is Double-Exponentially Hard},
  booktitle          = {ACM-SIAM Symposium on Simplicity in Algorithms (SOSA)},
  year               = {2024},
  doi                = {10.1137/1.9781611977936.25},
  timestamp          = {Mon, 03 Mar 2025 00:00:00 +0100},
  biburl             = {https://dblp.org/rec/conf/sosa/KowalikLMPS24.bib},
  bibsource          = {dblp computer science bibliography, https://dblp.org},
  _bib2doi_selected  = {dblp:/rec/conf/sosa/KowalikLMPS24.bib},
  _bib2doi_confirmed = {true}
}

@article{mccormick2001,
  author             = {S. Thomas McCormick and Scott R. Smallwood and Frits C. R. Spieksma},
  title              = {A Polynomial Algorithm for Multiprocessor Scheduling with Two Job Lengths},
  journal            = {Mathematics of Operations Research},
  volume             = {26},
  number             = {1},
  year               = {2001},
  doi                = {10.1287/moor.26.1.31.10590},
  timestamp          = {Sun, 28 May 2017 01:00:00 +0200},
  biburl             = {https://dblp.org/rec/journals/mor/McCormickSS01.bib},
  bibsource          = {dblp computer science bibliography, https://dblp.org},
  _bib2doi_selected  = {dblp:/rec/journals/mor/McCormickSS01.bib},
  _bib2doi_confirmed = {true}
}

@article{mnich2018,
  author             = {Matthias Mnich and Ren{\'e} van Bevern},
  title              = {Parameterized complexity of machine scheduling: 15 open problems},
  journal            = {Computers \& Operations Research},
  volume             = {100},
  pages              = {254--261},
  year               = {2018},
  doi                = {10.1016/j.cor.2018.07.020},
  timestamp          = {Sun, 19 Jan 2025 00:00:00 +0100},
  biburl             = {https://dblp.org/rec/journals/cor/MnichB18.bib},
  bibsource          = {dblp computer science bibliography, https://dblp.org},
  _bib2doi_selected  = {dblp:/rec/journals/cor/MnichB18.bib},
  _bib2doi_confirmed = {true}
}

@inproceedings{reis2023,
  author             = {Victor Reis and Thomas Rothvoss},
  title              = {The Subspace Flatness Conjecture and Faster Integer Programming},
  booktitle          = {64th Annual IEEE Symposium on Foundations of Computer Science (FOCS)},
  year               = {2023},
  doi                = {10.1109/FOCS57990.2023.00060},
  timestamp          = {Tue, 02 Jan 2024 00:00:00 +0100},
  biburl             = {https://dblp.org/rec/conf/focs/ReisR23.bib},
  bibsource          = {dblp computer science bibliography, https://dblp.org},
  _bib2doi_selected  = {dblp:/rec/conf/focs/ReisR23.bib},
  _bib2doi_confirmed = {true},
  _bib2doi_finished  = {true}
}

@book{rockafellar1970,
  author             = {R. Tyrrell Rockafellar},
  title              = {Convex Analysis},
  publisher          = {Princeton University Press},
  year               = {1970},
  doi                = {10.1515/9781400873173},
  timestamp          = {Thu, 25 Jul 2019 01:00:00 +0200},
  biburl             = {https://dblp.org/rec/books/degruyter/Rockafellar70.bib},
  bibsource          = {dblp computer science bibliography, https://dblp.org},
  isbn               = {978-1-4008-7317-3},
  _bib2doi_selected  = {dblp:/rec/books/degruyter/Rockafellar70.bib},
  _bib2doi_confirmed = {true}
}

@book{schrijver1986,
  author             = {Alexander Schrijver},
  title              = {Theory of Linear and Integer Programming},
  publisher          = {John Wiley \& Sons},
  year               = {1986},
  timestamp          = {Wed, 20 Apr 2011 01:00:00 +0200},
  biburl             = {https://dblp.org/rec/books/daglib/0090562.bib},
  bibsource          = {dblp computer science bibliography, https://dblp.org},
  isbn               = {978-0-471-98232-6},
  _bib2doi_selected  = {dblp:/rec/books/daglib/0090562.bib},
  _bib2doi_confirmed = {true},
  _bib2doi_finished  = {true}
}

@article{DBLP:journals/siamdm/Jansen10,
  author             = {Klaus Jansen},
  title              = {An {EPTAS} for Scheduling Jobs on Uniform Processors: Using an {MILP} Relaxation with a Constant Number of Integral Variables},
  journal            = {{SIAM} J. Discret. Math.},
  volume             = {24},
  number             = {2},
  year               = {2010},
  doi                = {10.1137/090749451},
  timestamp          = {Sat, 25 Apr 2020 01:00:00 +0200},
  biburl             = {https://dblp.org/rec/journals/siamdm/Jansen10.bib},
  bibsource          = {dblp computer science bibliography, https://dblp.org},
  _bib2doi_selected  = {dblp:/rec/journals/siamdm/Jansen10.bib},
  _bib2doi_confirmed = {true}
}

@article{cook1986sensitivity,
  author             = {William J. Cook and A. M. H. Gerards and Alexander Schrijver and {\'E}va Tardos},
  title              = {Sensitivity theorems in integer linear programming},
  journal            = {Mathematical Programming},
  volume             = {34},
  number             = {3},
  year               = {1986},
  doi                = {10.1007/BF01582230},
  timestamp          = {Sat, 27 May 2017 01:00:00 +0200},
  biburl             = {https://dblp.org/rec/journals/mp/CookGST86.bib},
  bibsource          = {dblp computer science bibliography, https://dblp.org},
  _bib2doi_selected  = {dblp:/rec/journals/mp/CookGST86.bib},
  _bib2doi_confirmed = {true}
}

@misc{knop2019mimo,
  author             = {Du{\v{s}}an Knop and Martin Kouteck{\'y} and Asaf Levin and Matthias Mnich and Shmuel Onn},
  title              = {Multitype Integer Monoid Optimization and Applications},
  year               = {2019},
  eprint             = {1909.07326},
  archiveprefix      = {arXiv},
  primaryclass       = {cs.DS},
  doi                = {10.48550/arXiv.1909.07326},
  timestamp          = {Sat, 05 Sep 2020 01:00:00 +0200},
  biburl             = {https://dblp.org/rec/journals/corr/abs-1909-07326.bib},
  bibsource          = {dblp computer science bibliography, https://dblp.org},
  _bib2doi_selected  = {dblp:/rec/journals/corr/abs-1909-07326.bib},
  _bib2doi_confirmed = {true}
}

@misc{brandao2013vector,
  author             = {Filipe Brand{\~a}o and Jo{\~a}o Pedro Pedroso},
  title              = {Multiple-choice Vector Bin Packing: Arc-flow Formulation with Graph Compression},
  year               = {2013},
  eprint             = {1312.3836},
  archiveprefix      = {arXiv},
  primaryclass       = {math.OC},
  doi                = {10.48550/arXiv.1312.3836},
  timestamp          = {Tue, 17 Sep 2019 01:00:00 +0200},
  biburl             = {https://dblp.org/rec/journals/corr/BrandaoP13a.bib},
  bibsource          = {dblp computer science bibliography, https://dblp.org},
  _bib2doi_selected  = {dblp:/rec/journals/corr/BrandaoP13a.bib},
  _bib2doi_confirmed = {true}
}

@article{knop2018scheduling,
  author             = {Du{\v{s}}an Knop and Martin Kouteck{\'y}},
  title              = {Scheduling meets {$n$}-fold integer programming},
  journal            = {Journal of Scheduling},
  volume             = {21},
  number             = {5},
  year               = {2018},
  doi                = {10.1007/s10951-017-0550-0},
  timestamp          = {Wed, 01 Apr 2020 01:00:00 +0200},
  biburl             = {https://dblp.org/rec/journals/scheduling/KnopK18.bib},
  bibsource          = {dblp computer science bibliography, https://dblp.org},
  _bib2doi_selected  = {dblp:/rec/journals/scheduling/KnopK18.bib},
  _bib2doi_confirmed = {true}
}

@article{knop2023huge,
  author             = {Du{\v{s}}an Knop and Martin Kouteck{\'y} and Asaf Levin and Matthias Mnich and Shmuel Onn},
  title              = {High-multiplicity {$N$}-fold {IP} via configuration {LP}},
  journal            = {Mathematical Programming},
  volume             = {200},
  number             = {1},
  year               = {2023},
  doi                = {10.1007/s10107-022-01882-9},
  timestamp          = {Tue, 12 Sep 2023 01:00:00 +0200},
  biburl             = {https://dblp.org/rec/journals/mp/KnopKLMO23.bib},
  bibsource          = {dblp computer science bibliography, https://dblp.org},
  _bib2doi_selected  = {dblp:/rec/journals/mp/KnopKLMO23.bib},
  _bib2doi_confirmed = {true}
}

@article{kouteckyZink2025,
  author             = {Martin Kouteck{\'y} and Johannes Zink},
  title              = {Complexity of scheduling few types of jobs on related and unrelated machines},
  journal            = {Journal of Scheduling},
  volume             = {28},
  number             = {1},
  year               = {2025},
  doi                = {10.1007/s10951-024-00827-8},
  timestamp          = {Sun, 15 Jun 2025 01:00:00 +0200},
  biburl             = {https://dblp.org/rec/journals/scheduling/KouteckyZ25.bib},
  bibsource          = {dblp computer science bibliography, https://dblp.org},
  _bib2doi_selected  = {dblp:/rec/journals/scheduling/KouteckyZ25.bib},
  _bib2doi_confirmed = {true}
}

@inproceedings{rohwedder2025uniform,
  author             = {Lars Rohwedder},
  title              = {{ETH}-Tight {FPT} Algorithm for Makespan Minimization on Uniform Machines},
  booktitle          = {52nd International Colloquium on Automata, Languages, and Programming (ICALP 2025)},
  year               = {2025},
  doi                = {10.4230/LIPIcs.ICALP.2025.126},
  timestamp          = {Sat, 06 Sep 2025 01:00:00 +0200},
  biburl             = {https://dblp.org/rec/conf/icalp/Rohwedder25.bib},
  bibsource          = {dblp computer science bibliography, https://dblp.org},
  _bib2doi_selected  = {dblp:/rec/conf/icalp/Rohwedder25.bib},
  _bib2doi_confirmed = {true}
}

@inproceedings{jansen2025combinatorial,
  author             = {Klaus Jansen and Kai Kahler and Lis Pirotton and Malte Tutas},
  title              = {New Algorithm for Combinatorial {$n$}-Folds and Applications},
  booktitle          = {20th International Symposium on Parameterized and Exact Computation (IPEC 2025)},
  year               = {2025},
  doi                = {10.4230/LIPIcs.IPEC.2025.15},
  timestamp          = {Mon, 15 Dec 2025 00:00:00 +0100},
  biburl             = {https://dblp.org/rec/conf/iwpec/JansenKPT25.bib},
  bibsource          = {dblp computer science bibliography, https://dblp.org},
  _bib2doi_selected  = {dblp:/rec/conf/iwpec/JansenKPT25.bib},
  _bib2doi_confirmed = {true}
}

@article{veselov2020,
  author  = {S. I. Veselov and D. V. Gribanov and N. Yu. Zolotykh and A. Yu. Chirkov},
  title   = {A polynomial algorithm for minimizing discrete convic functions in fixed dimension},
  journal = {Discrete Applied Mathematics},
  volume  = {283},
  pages   = {11--19},
  year    = {2020},
  doi     = {10.1016/j.dam.2019.10.006}
}

@article{heinz2005,
  author  = {Sebastian Heinz},
  title   = {Complexity of integer quasiconvex polynomial optimization},
  journal = {Journal of Complexity},
  volume  = {21},
  number  = {4},
  pages   = {543--556},
  year    = {2005},
  doi     = {10.1016/j.jco.2005.04.004}
}

@article{hildebrandKoeppe2013,
  author  = {Robert Hildebrand and Matthias K{\"o}ppe},
  title   = {A new {Lenstra}-type algorithm for quasiconvex polynomial integer minimization with complexity {$2^{O(n\log n)}$}},
  journal = {Discrete Optimization},
  volume  = {10},
  number  = {1},
  pages   = {69--84},
  year    = {2013},
  doi     = {10.1016/j.disopt.2012.11.003}
}

@article{DBLP:journals/mor/JansenS11,
  author    = {Klaus Jansen and
               Roberto Solis{-}Oba},
  title     = {A Polynomial Time \emph{OPT} + 1 Algorithm for the Cutting Stock Problem
               with a Constant Number of Object Lengths},
  journal   = {Math. Oper. Res.},
  volume    = {36},
  number    = {4},
  pages     = {743--753},
  year      = {2011},
  url       = {https://doi.org/10.1287/moor.1110.0515},
  doi       = {10.1287/MOOR.1110.0515},
  bibsource = {dblp computer science bibliography, https://dblp.org}
}

@article{brandaoPedroso2016arcflow,
  author  = {Filipe Brand{\~a}o and Jo{\~a}o Pedro Pedroso},
  title   = {Bin packing and related problems: General arc-flow formulation with graph compression},
  journal = {Computers \& Operations Research},
  volume  = {69},
  pages   = {56--67},
  year    = {2016},
  doi     = {10.1016/j.cor.2015.11.009}
}

@article{gilmoreGomory1961,
  author  = {P. C. Gilmore and R. E. Gomory},
  title   = {A Linear Programming Approach to the Cutting-Stock Problem},
  journal = {Operations Research},
  volume  = {9},
  number  = {6},
  pages   = {849--859},
  year    = {1961},
  doi     = {10.1287/opre.9.6.849}
}

@article{assmann1984,
  author  = {Susan F. Assmann and David S. Johnson and Daniel J. Kleitman and Joseph Y.-T. Leung},
  title   = {On a dual version of the one-dimensional bin packing problem},
  journal = {Journal of Algorithms},
  volume  = {5},
  number  = {4},
  pages   = {502--525},
  year    = {1984},
  doi     = {10.1016/0196-6774(84)90004-X}
}

@inproceedings{bannach2020small,
  author    = {Max Bannach and Sebastian Berndt and Marten Maack and Matthias Mnich and Alexandra Lassota and Malin Rau and Malte Skambath},
  title     = {Solving Packing Problems with Few Small Items Using Rainbow Matchings},
  booktitle = {45th International Symposium on Mathematical Foundations of Computer Science (MFCS 2020)},
  year      = {2020},
  doi       = {10.4230/LIPIcs.MFCS.2020.11}
}

@inproceedings{lassota2022small,
  author    = {Alexandra Lassota and Aleksander {\L}ukasiewicz and Adam Polak},
  title     = {Tight Vector Bin Packing with Few Small Items via Fast Exact Matching in Multigraphs},
  booktitle = {49th International Colloquium on Automata, Languages, and Programming (ICALP 2022)},
  year      = {2022},
  doi       = {10.4230/LIPIcs.ICALP.2022.87}
}

@misc{brinkop2022structural,
  author        = {Hauke Brinkop and David Fischer and Klaus Jansen},
  title         = {Structural Results for High-Multiplicity Scheduling on Uniform Machines},
  year          = {2022},
  note          = {Revised September 2024},
  eprint        = {2203.01741},
  archiveprefix = {arXiv},
  primaryclass  = {cs.DS}
}

@inproceedings{DBLP:conf/icalp/BuchemRVW21,
  author    = {Moritz Buchem and
               Lars Rohwedder and
               Tjark Vredeveld and
               Andreas Wiese},
  title     = {Additive Approximation Schemes for Load Balancing Problems},
  booktitle = {48th International Colloquium on Automata, Languages, and Programming,
               {ICALP} 2021},
  year      = {2021},
  url       = {https://doi.org/10.4230/LIPIcs.ICALP.2021.42},
  doi       = {10.4230/LIPICS.ICALP.2021.42},
  bibsource = {dblp computer science bibliography, https://dblp.org}
}

@article{DBLP:journals/corr/abs-2007-09333,
  author     = {Moritz Buchem and
                Lars Rohwedder and
                Tjark Vredeveld and
                Andreas Wiese},
  title      = {Additive Approximation Schemes for Load Balancing Problems},
  journal    = {CoRR},
  volume     = {abs/2007.09333},
  year       = {2020},
  url        = {https://doi.org/10.48550/arXiv.2007.09333},
  doi        = {10.48550/ARXIV.2007.09333},
  eprinttype = {arXiv},
  eprint     = {2007.09333},
  bibsource  = {dblp computer science bibliography, https://dblp.org}
}

@article{nguyenRothe2026,
  author  = {Trung Thanh Nguyen and J{\"o}rg Rothe},
  title   = {Mathematical programming approaches for social welfare maximization},
  journal = {4OR},
  volume  = {24},
  pages   = {259--312},
  year    = {2026},
  doi     = {10.1007/s10288-026-00628-z},
  url     = {https://doi.org/10.1007/s10288-026-00628-z}
}

@misc{nguyenElbassioni2026,
  author        = {Trung Thanh Nguyen and Khaled Elbassioni},
  title         = {Maximizing {$p$}-Mean Social Welfare in the High-Multiplicity Setting: Few Agent Types and Few Item Types},
  year          = {2026},
  eprint        = {2610.04417},
  archiveprefix = {arXiv},
  primaryclass  = {cs.GT},
  url           = {https://arxiv.org/abs/2610.04417},
  note          = {Version 1, October 3, 2026}
}

@article{burnikel2009separation,
  author  = {Christoph Burnikel and Stefan Funke and Kurt Mehlhorn and Stefan Schirra and Susanne Schmitt},
  title   = {A Separation Bound for Real Algebraic Expressions},
  journal = {Algorithmica},
  volume  = {55},
  pages   = {14--28},
  year    = {2009},
  doi     = {10.1007/s00453-007-9132-4}
}

\newpage
\appendix

\section{Applications}\label{sec:applications}
All multiplicities in this section are encoded in binary, and $|I|$ denotes the binary encoding length of the respective instance.
Solutions are returned as configurations together with their multiplicities, so their descriptions need not list every item or machine individually.
Except for the mixed-integer formulation in \Cref{sec:huge-nfold} and the additive scheduling scheme in \Cref{sec:additive-scheduling}, the applications use $k=0$.
Except for generalized-mean welfare in \Cref{sec:welfare}, the cost functions $f^{(t)}$ are explicitly given convex polynomials with rational coefficients, so \Cref{ass:cost-properties} holds.
The welfare application uses the extension in Appendix~\ref{sec:general-convex-costs}.

\subsection{Huge \texorpdfstring{$n$}{n}-Fold Optimization}\label{sec:huge-nfold}
In huge $n$-fold programming, the bricks are specified by their types and binary-encoded multiplicities, rather than listed individually.
A complete brick type specifies its local constraints, bounds, linking matrices, integrality pattern, and objective.
Our formulation allows the linking matrices to differ between types, and hence covers generalized $n$-fold programs.
The following corollary also permits continuous variables within each brick; its objective depends only on the integer coordinates.

\begin{corollary}[Huge $n$-Fold Optimization with Convex Polynomial Costs]\label{cor:huge-nfold}
    For every $t\in[\tau]$, let $P^{(t)}\subseteq\R^{d+k}$ be a bounded rational polytope given by explicit linear inequalities, let $m^{(t)}\in\Z_{\ge0}$ be encoded in binary, and let $f^{(t)}$ be a rational polynomial in $d$ variables given in dense encoding by total degree and convex on $P_{\mathrm{x}}^{(t)}$.
    Given rational matrices $A^{(t)}\in\Q^{r\times d}$ and $C^{(t)}\in\Q^{r\times k}$ and a vector $\vect b\in\Q^r$, the problem
    \[
        \begin{aligned}
            \min\quad
             & \sum_{t=1}^{\tau}\sum_{j=1}^{m^{(t)}}
            f^{(t)}(\vect x^{(t,j)})                                                \\
            \text{s.t.}\quad
             & \sum_{t=1}^{\tau}\sum_{j=1}^{m^{(t)}}
            \bigl(A^{(t)}\vect x^{(t,j)}+C^{(t)}\vect u^{(t,j)}\bigr)=\vect b,      \\
             & (\vect x^{(t,j)},\vect u^{(t,j)})
            \in P^{(t)}\cap(\Z^d\times\R^k)
             &                                       & (t\in[\tau],\ j\in[m^{(t)}])
        \end{aligned}
    \]
    can be solved deterministically in time
    \[
        2^{\tau\log(\tau+1)\,2^{O(d)}}|I|^{O(1)},
    \]
    where $|I|$ is the compact binary input length, including the multiplicities and cost descriptions.
\end{corollary}
\begin{proof}
    Apply \Cref{thm:main} with the given polytopes and costs and $Q=\R^{d+k}$.
    Let $\mathcal A$ impose
    \[
        N^{(t)}=m^{(t)}\quad(t\in[\tau]),
        \qquad
        \sum_{t=1}^{\tau}
        \bigl(A^{(t)}\vect y^{(t)}+C^{(t)}\vect v^{(t)}\bigr)=\vect b.
    \]
    Grouping equal brick configurations within each type gives a feasible multiplicity representation of the same cost.
    Conversely, repeating each returned configuration according to its multiplicity gives a feasible brick assignment of the same cost; the algorithm outputs this assignment in its compact representation.
    The reduction has polynomial encoding length and does not expand the binary multiplicities.
    Since the local polytopes are bounded and the brick counts are prescribed, every feasible instance has an attained finite optimum.
    The running time follows from \Cref{thm:main}.
\end{proof}

Taking $k=0$ and $P^{(t)}=\{\vect x:B^{(t)}\vect x=\vect b^{(t)},\ \vect\ell^{(t)}\le\vect x\le\vect u^{(t)}\}$ recovers bounded huge generalized $n$-fold IP as introduced in~\cite{knop2019mimo}.
Local and linking inequalities are allowed as well, without introducing integer slack coordinates.
More generally, any rational linear constraints on the typewise aggregates can be included in $\mathcal A$.
In particular, the parameter $\tau$ counts complete types, not merely distinct pairs of constraint matrices.

For integer constraint matrices and coordinatewise separable convex objectives, Knop, Kouteck\'y, Levin, Mnich, and Onn~\cite[Theorem~1]{knop2023huge} give a huge $n$-fold algorithm that is \ac{FPT} in $r+s+\Delta$ and polynomial in $d$ and $\tau$, where $r$ is the number of global equations, $s$ is the maximum number of local equations per type, and $\Delta$ is the maximum absolute matrix entry.
\Cref{cor:huge-nfold} instead parameterizes $d+\tau$, allowing arbitrary binary coefficients and unrestricted row counts; these parameterizations are incomparable.
Our convex polynomial costs may also couple coordinates within a brick.

\subsection{Machine Scheduling}\label{sec:scheduling}
We apply our algorithm to high-multiplicity scheduling with convex load costs on uniform and unrelated machines, and then to makespan minimization.

Consider $d$ job types with processing requirements $\vect p\in\mathbb Q_{>0}^d$ and multiplicities $\vect a\in\Z_{\ge0}^d$.
There are $\tau$ machine types: type $t$ consists of $m^{(t)}\in\Z_{\ge0}$ identical machines of speed $s_t\in\mathbb Q_{>0}$.
Assume $\sum_{t=1}^\tau {m^{(t)}} \geq 1$.
A job of type $j$ takes time $p_j/s_t$ on a machine of type $t$.
If the $\ell$-th machine of type $t$ receives the configuration $\vect x^{(t,\ell)}$, its load is
\[
    L_{t,\ell}\coloneqq \frac{\langle \vect p, \vect x^{(t,\ell)}\rangle}{s_t}.
\]
For a fixed integer $q\ge1$, we minimize
\begin{equation}\label{eq:load-power-objective}
    \sum_{t=1}^{\tau}\sum_{\ell=1}^{m^{(t)}}L_{t,\ell}^{q}.
\end{equation}
This is equivalent to minimizing the $\ell_q$-norm of the load vector; in particular, $q=2$ gives the standard sum-of-squared-loads objective.
For $q\ge2$, these costs are nonlinear and are not directly handled by the polyhedral XP-time framework of Goemans and Rothvoss~\cite{goemans2020}.
Knop, Kouteck\'y, Levin, Mnich, and Onn~\cite[Section~5.1.4]{knop2019mimo} obtain \ac{FPT} algorithms for $\ell_q$-load minimization with fixed integer $q\ge2$, parameterized by $d+m$ or $d+p_{\max}$, where $m\coloneqq \sum_t m^{(t)}$ and $p_{\max}$ is the largest integral processing requirement.
They explicitly identify extending their framework to convex objectives beyond small processing requirements as an open direction~\cite[Section~6]{knop2019mimo}.

For a more detailed comparison with coefficient-based algorithms, suppose that processing requirements are integral.
Building on the $n$-fold IP formulations of Knop and Kouteck\'y~\cite{knop2018scheduling}, Kouteck\'y and Zink~\cite{kouteckyZink2025} obtain $p_{\max}^{O(d^2)}$ parameter dependence for makespan and $\ell_2$-load minimization on uniform and unrelated machines.
The high-multiplicity $n$-fold IP algorithm of Knop, Kouteck\'y, Levin, Mnich, and Onn~\cite{knop2023huge} extends these results to binary machine multiplicities, as explained in~\cite[Section~3.2]{kouteckyZink2025}.
More generally, for fixed $q$, introducing a load variable with a separable convex power cost in that algorithm gives a $(d p_{\max})^{O(d^2)}|I|^{O(1)}$ bound for~\eqref{eq:load-power-objective}.

Our result instead parameterizes $d+\tau$, with only polynomial dependence on the binary encoding of processing times and speeds.

\begin{corollary}[Uniform Scheduling with Convex Load Costs]\label{cor:uniform-load-costs}
    For every fixed integer $q\ge1$, high-multiplicity scheduling with objective~\eqref{eq:load-power-objective}, $d$ job types, and $\tau$ distinct machine speeds can be solved deterministically in time
    \[
        2^{\tau\log(\tau+1)\,2^{O(d)}}|I|^{O(1)}.
    \]
    The algorithm returns an optimal schedule using at most $2^{2d+1}$ distinct configurations for each speed.
\end{corollary}
\begin{proof}
    For every machine type $t$, set
    \[
        P^{(t)}\coloneqq \{\vect x\in\mathbb R^d:\vect0\le\vect x\le\vect a\},
        \qquad f^{(t)}(\vect x)\coloneqq \left(\frac{\langle \vect p, \vect x\rangle}{s_t}\right)^q.
    \]
    Set $Q\coloneqq \{\vect a\}$ and let $\mathcal A$ impose $N^{(t)}=m^{(t)}$ for every $t$.
    The zero configuration represents an idle machine and has cost zero.
    Every schedule gives a feasible generator representation by recording how many machines of each type use each configuration.
    Conversely, a feasible representation assigns exactly $a_j$ jobs of each type $j$ to exactly $m^{(t)}$ machines of each type $t$, including idle machines.
    Its additive generator cost is precisely~\eqref{eq:load-power-objective}.

    Each domain $P^{(t)}$ is a bounded rational polytope.
    Since $\langle \vect p, \vect x\rangle/s_t\ge0$ on this domain, $f^{(t)}$ is convex there.
    For fixed integer $q$, expanding $f^{(t)}$ gives a rational polynomial of degree $q$ whose dense coefficient description has polynomial encoding length.
    Thus the reduction has explicitly given convex polynomial costs of polynomial encoding length, so \Cref{ass:cost-properties} holds.
    Applying \Cref{thm:main} proves the claimed running time and support bound.
\end{proof}

Here $\tau$ counts distinct speeds, whereas the total number of machines $\sum_t m^{(t)}$ may be exponentially large in its encoding length.
For identical machines, $\tau=1$ and the running time becomes $2^{2^{O(d)}}\cdot |I|^{O(1)}$.
Fixing an integer $q$ ensures polynomial encoding length of the explicit polynomial costs.

The scheduling reduction also permits a processing-time vector $\vect p^{(t)}\in\mathbb Q_{\ge0}^d$ for each machine type, with load $\langle \vect p^{(t)}, \vect x\rangle$ and cost $(\langle \vect p^{(t)}, \vect x\rangle)^q$.
The convexity and polynomial encoding arguments above apply unchanged, giving the same running time in $d+\tau$.

One common alternative objective is the makespan $C_{\max}\coloneqq \max_{t,\ell}L_{t,\ell}$.
For uniform-machine makespan minimization, Brinkop, Fischer, and Jansen~\cite{brinkop2022structural} give a proximity-based algorithm with $p_{\max}^{O(d^2)}$ parameter dependence.
Rohwedder~\cite{rohwedder2025uniform} improves this dependence to $p_{\max}^{O(d)}$.
Jansen, Kahler, Pirotton, and Tutas~\cite{jansen2025combinatorial} obtain a $p_{\max}^{O(d)}|I|^{O(1)}$ bound with binary multiplicities of both jobs and machines, as required by the input model used here.

Our $d+\tau$ parameterization also applies to makespan minimization.
Uniform-machine makespan minimization is NP-hard already for six job types when the number of speeds is unrestricted~\cite[Theorem~4.13]{kouteckyZink2025}.
Our algorithm applies as follows:
For a prescribed makespan $T$, replace the domains by
\[
    P^{(t)}_T\coloneqq \{\vect x\in\mathbb R^d:\vect0\le\vect x\le\vect a,
    \ \langle \vect p^{(t)}, \vect x\rangle\le T\}.
\]
Zero costs decide makespan feasibility; retaining the power costs minimizes the load objective subject to the makespan bound.
Exact makespan minimization follows by binary search: after multiplying all finite processing times by a common denominator $D$, every load is integral and the search interval is
$[0,\sum_{j:a_j>0} a_j\max_t Dp_j^{(t)}]$.
Both $D$ and the interval endpoints have polynomial encoding length, so the search uses polynomially many calls.

\subsection{Point in Cone and High-Multiplicity Bin Packing}

We now specialize to $\tau=1$, writing $P=P^{(1)}$, $f=f^{(1)}$, and $\lambda_{\vect x}=\lambda^{(1)}_{\vect x}$, with no additional typewise constraints unless specified. The zero-objective case of \Cref{thm:main} implies \Cref{thm:point-in-cone,thm:binpacking} via known polynomial-time reductions.
In each case, the reduction preserves dimension $d$ and has polynomial encoding length, and a sparse representation immediately yields the desired solution. For constant costs, the pricing oracle is linear integer optimization, as in \Cref{lem:pr-oracle}.
We provide the proofs for completeness, but the reductions are well-known in the literature~\cite{goemans2020,kowalik2024}.

For the Point-in-Cone problem, the algorithm of Goemans and Rothvoss~\cite{goemans2020} gives polynomial-time solvability in every fixed dimension.
\Cref{thm:main} yields a constructive \ac{FPT} algorithm matching the lower bound of Kowalik, Lassota, Majewski, Pilipczuk, and Soko{\l}owski~\cite{kowalik2024}.

\begin{corollary}[Point-in-Cone]\label{thm:point-in-cone}
    Given a bounded rational polyhedron $P \subseteq \R^d$ and a target vector $\vect y \in \Z^d$, one can determine whether $\vect y \in \operatorname{int{.}cone}(P \cap \Z^d)$ in time $2^{2^{O(d)}} \cdot (\enc(P) + \enc(\vect y))^{O(1)}$. If so, one can construct a certificate $\vect \lambda \in \Z_{\ge 0}^{P \cap \Z^d}$ with $\vect y = \sum_{\vect x \in P \cap \Z^d} \lambda_{\vect x} \vect x$ and $|\operatorname{supp}(\vect \lambda)| \le 2^{2d+1}$.
\end{corollary}

\begin{proof}[Proof of \Cref{thm:point-in-cone}]
    Let $P \subseteq \R^d$ be a bounded rational polyhedron and $\vect y \in \Z^d$ a target vector. We define the target polyhedron $Q \subseteq \R^d$ as the singleton:
    \begin{equation*}
        Q \coloneqq  \{\vect y\} = \left\{ \vect z \in \R^d : \vect z \le \vect y, \; -\vect z \le -\vect y \right\}.
    \end{equation*}
    The ambient dimension is $d$, and the explicit inequality system for $Q$ has encoding length $O(d^2+\enc(\vect y))$.
    Thus $\enc(P)+\enc(Q)=(\enc(P)+\enc(\vect y))^{O(1)}$.
    By definition, $\operatorname{int{.}cone}(P \cap \Z^d) \cap Q \ne \emptyset$ if and only if $\vect y \in \operatorname{int{.}cone}(P \cap \Z^d)$. Set $f=0$. Applying \Cref{thm:main} determines whether $\vect y \in \operatorname{int{.}cone}(P \cap \Z^d)$ and, if so, outputs a certificate $\vect \lambda \in \Z_{\ge 0}^{P \cap \Z^d}$ such that $\sum_{\vect x \in P \cap \Z^d} \lambda_{\vect x} \vect x = \vect y$ with $|\operatorname{supp}(\vect \lambda)| \le 2^{2d+1}$ in deterministic time $2^{2^{O(d)}} \cdot (\enc(P) + \enc(\vect y))^{O(1)}$.
\end{proof}

McCormick, Smallwood, and Spieksma~\cite{mccormick2001} gave a polynomial-time algorithm for high-multiplicity bin packing with two item types.
Goemans and Rothvoss~\cite{goemans2020} established polynomial-time solvability for every fixed number of types.
Jansen and Klein~\cite{jansenKlein2020} gave an \ac{FPT} algorithm parameterized by the number of vertices of its integer hull and Jansen and Solis-Oba~\cite{DBLP:journals/mor/JansenS11} gave an \ac{FPT} algorithm parameterized by the number of item sizes $d$ using at most one additional bin.
The question of fixed-parameter tractability in the number of types remained open~\cite{mnich2018,jansen2025ciac,kowalik2024,jansen2026} until Koana and Kumabe~\cite{koana2026} gave a deterministic $2^{2^{O(d\log d)}}\cdot |I|^{O(1)}$-time algorithm.
We improve this dependence to $2^{2^{O(d)}}$ and construct a packing represented by generators and their multiplicities.
This running time is \ac{ETH}-tight in $d$: Jansen, Ohnesorge, and Pirotton~\cite{jansen2026} rule out a $2^{2^{o(d)}}\cdot |I|^{O(1)}$-time algorithm for high-multiplicity bin packing.

\begin{corollary}[High-Multiplicity Bin Packing]\label{thm:binpacking}
    Given a bin capacity $B \in \Z_{\ge 1}$, $d$ item types with sizes $\vect s = (s_1, \dots, s_d)^\top \in \Z_{\ge 1}^d$ and multiplicities $\vect a = (a_1, \dots, a_d)^\top \in \Z_{\ge 0}^d$, and a bin bound $b \in \Z_{\ge 0}$, one can determine whether all items can be packed into at most $b$ bins of capacity $B$ in time $2^{2^{O(d)}} \cdot |I|^{O(1)}$, where $|I| = \Theta(\enc(B, \vect s, \vect a, b))$ is the total binary input length. If so, one can also construct a packing with at most $2^{2d+1}$ distinct generators.
\end{corollary}

\begin{proof}[Proof of \Cref{thm:binpacking}]
    Let an instance of high-multiplicity bin packing be given by bin capacity $B \in \Z_{\ge 1}$, item sizes $\vect s = (s_1, \dots, s_d)^\top \in \Z_{\ge 1}^d$, item multiplicities $\vect a = (a_1, \dots, a_d)^\top \in \Z_{\ge 0}^d$, and bin bound $b \in \Z_{\ge 0}$.
    A valid bin generator is an integer vector $\vect x \in \Z_{\ge 0}^d$ satisfying $\langle \vect s, \vect x\rangle \le B$.
    Set
    \[
        P\coloneqq \{\vect x\in\mathbb R_{\ge0}^d:\langle \vect s, \vect x\rangle\le B\},
        \qquad Q\coloneqq \{\vect a\},
        \qquad f\coloneqq 0,
    \]
    and include $N^{(1)}\le b$ in $\mathcal A$.
    Since $s_i\ge1$, every coordinate satisfies $0\le x_i\le B/s_i\le B$, so $P$ is bounded.
    This formulation has dimension $d$ and encoding length $|I|^{O(1)}$.
    A feasible representation is exactly a collection of valid bin generators with
    $\sum_{\vect x}\lambda_{\vect x}\vect x=\vect a$ and $\sum_{\vect x}\lambda_{\vect x}\le b$, hence a packing of all items into at most $b$ bins.
    Applying \Cref{thm:main} decides feasibility in time $2^{2^{O(d)}}|I|^{O(1)}$ and returns a packing supported on at most $2^{2d+1}$ distinct generators.

    To minimize the number of bins, omit the constraint $N^{(1)}\le b$ and set $f=1$.
    Then the objective is $\sum_{\vect x}\lambda_{\vect x}$, so \Cref{thm:main} directly returns a minimum-bin packing with the same running time and support bound, or reports infeasibility.
\end{proof}

\subsection{Vector Bin Packing}

In vector bin packing, each item consumes several resources simultaneously, and the total consumption in every bin must respect every resource capacity; see Brand\~ao and Pedroso~\cite{brandaoPedroso2016arcflow}.
Let $d$ be the number of distinct item types and $r$ the number of resource dimensions.
The size vectors are $\vect s_1,\dots,\vect s_d\in\mathbb Q_{\ge0}^r$, their multiplicities are $\vect a\in\Z_{>0}^d$, and each bin has capacity vector $\vect B\in\mathbb Q_{>0}^r$ with $\vect s_j \leq \vect B$ for all $j\in[d]$.
Writing $A\coloneqq [\vect s_1\mid\cdots\mid\vect s_d]$, a bin configuration $\vect x\in\Z_{\ge0}^d$ is feasible if $A\vect x\le\vect B$.

The cone theorem of Goemans and Rothvoss~\cite{goemans2020} yields an XP algorithm parameterized by $d$.
After clearing denominators in each resource constraint, let $\Delta$ denote the largest entry of the resulting integer size matrix.
The configuration-IP framework of Knop, Kouteck\'y, Levin, Mnich, and Onn~\cite[Section~5]{knop2021} gives an \ac{FPT} algorithm parameterized by $d+r+\Delta$.
There are also algorithms with different parameterizations, including randomized \ac{FPT} algorithms for explicitly listed items parameterized by the number of small items~\cite{bannach2020small,lassota2022small}.
The following corollary gives an \ac{FPT} algorithm parameterized by $d$ alone, with polynomial dependence on the binary encoding of all resource constraints.

\begin{corollary}[High-Multiplicity Vector Bin Packing]\label{cor:vector-bin-packing}
    High-multiplicity vector bin packing with $d$ distinct item types can be solved deterministically in time $2^{2^{O(d)}}|I|^{O(1)}$.
    The algorithm returns a minimum-bin packing with at most $2^{2d+1}$ distinct bin configurations.
    The number $r$ of resource dimensions is part of the input and need not be a parameter.
\end{corollary}
\begin{proof}
    Use one generator type and set
    \[
        P\coloneqq \{\vect x\in\mathbb R^d:\vect0\le\vect x\le\vect a,
        \ A\vect x\le\vect B\},
        \qquad Q\coloneqq \{\vect a\},
        \qquad f\coloneqq 1,
    \]
    with no additional constraints in $\mathcal A$.
    The bounds $x_j\le a_j$ exclude no configuration occurring in a packing of the given items.
    Every packing determines a feasible representation of $\vect a$ over $P\cap\Z^d$, and every such representation determines a packing.
    The objective $\sum_{\vect x}\lambda_{\vect x}$ counts its bins.
    Empty configurations have positive cost and can be omitted from an optimum.
    The constant objective is a convex polynomial, and the description of $P$ has $d$ variables and $r+2d$ inequalities with polynomial encoding length.
    Applying \Cref{thm:main} with $\tau=1$ proves the result.
    In particular, increasing $r$ increases the input length but not the dimension of the generator space.
\end{proof}

A prescribed bin bound $b$ can be imposed by adding $N^{(1)}\le b$ to $\mathcal A$ and setting $f=0$.
Since ordinary bin packing is the case $r=1$, the double-exponential dependence on $d$ is \ac{ETH}-tight by~\cite{jansen2026}.

\subsection{Several Bin Types and Cutting Stock}
We next allow $\tau$ bin types with different capacities and costs; this setting is also included in the multiple-choice vector-packing formulation of Brand\~ao and Pedroso~\cite{brandao2013vector}.
Suppose bin type $t$ has capacity vector $\vect B^{(t)}$ and nonnegative rational cost $c_t$.
Set
\[
    P^{(t)}\coloneqq \{\vect x\in\mathbb R^d:\vect0\le\vect x\le\vect a,
    \ A\vect x\le\vect B^{(t)}\},
    \qquad f^{(t)}\coloneqq c_t,
    \qquad Q\coloneqq \{\vect a\}.
\]
Finite supplies of bins are expressed by $N^{(t)}\le m^{(t)}$ in $\mathcal A$.
The objective is total purchasing cost, so \Cref{thm:main} yields a minimum-cost packing in time $2^{\tau\log(\tau+1)\,2^{O(d)}}|I|^{O(1)}$, with at most $2^{2d+1}$ configurations per bin type.
In one resource dimension, this is cutting stock with $d$ requested piece lengths and $\tau$ stock lengths and prices.
Its configuration formulation goes back to Gilmore and Gomory~\cite{gilmoreGomory1961}.
Goemans and Rothvoss~\cite[Section~6.1]{goemans2020} give an XP algorithm parameterized by $d+\tau$.
The number of item types alone does not yield fixed-parameter tractability unless $\mathrm{P}=\mathrm{NP}$: Kouteck\'y and Zink~\cite[Theorem~4.17]{kouteckyZink2025} prove that cutting stock is NP-hard already for eight item types when the number of stock types is unrestricted.

\subsection{Bin Covering}\label{sec:bin-covering}
Given positive rational item sizes $\vect s$, multiplicities $\vect a$, and a threshold $B>0$, bin covering asks for as many disjoint groups of items of total size at least $B$ as possible~\cite{assmann1984}.
Unused items are allowed.

An XP algorithm parameterized by $d$ follows from the cone theorem of Goemans and Rothvoss~\cite{goemans2020}.
Indeed, append a coordinate equal to one to each covering configuration and require the sum to belong to $[\vect0,\vect a]\times[b,\sum_j a_j]$.
This decides whether at least $b$ bins can be covered, so binary search finds the optimum.
For explicitly listed items,~\cite{bannach2020small,lassota2022small} give randomized \ac{FPT} algorithms for vector bin packing and vector bin covering parameterized by the number of small items.
That parameter counts individual items and may be large even when the number of item types is fixed.
Our result applies directly to binary multiplicities and is parameterized by the number of types.

Use a single generator type with
\[
    P\coloneqq \{\vect x\in\mathbb R^d:\vect0\le\vect x\le\vect a,
    \ \langle \vect s, \vect x\rangle\ge B\},
    \qquad Q\coloneqq \{\vect y:\vect0\le\vect y\le\vect a\},
    \qquad f\coloneqq -1.
\]
Each generator describes a covered bin, the constraints in $Q$ enforce item availability, and minimizing the objective maximizes the number of covered bins.
Every generator contains at least one item, so the number of bins is at most $\sum_j a_j$ despite the negative generator cost.
Thus \Cref{thm:main} computes an optimum in time $2^{2^{O(d)}}|I|^{O(1)}$, represented by at most $2^{2d+1}$ distinct covered-bin configurations.

The double-exponential dependence on $d$ is \ac{ETH}-tight by a padding reduction from high-multiplicity bin packing~\cite{jansen2026}.
Consider an instance with positive integral item sizes, capacity $B$, and a target of at most $b$ bins, and let $S$ be the total item size.
After rejecting the trivial case $S>bB$, add $bB-S$ items of size one, introducing at most one new item type.
A packing into at most $b$ bins can be padded to exactly $b$ bins of load $B$.
Conversely, since the augmented instance has total size exactly $bB$, covering at least $b$ bins requires exactly $b$ bins of load $B$ and uses all items; removing the added unit items gives a packing of the original instance.
The reduction has polynomial binary encoding length and increases the number of item types by at most one, so the lower bound transfers.

\subsection{Generalized-Mean Social Welfare}\label{sec:welfare}

Consider allocating $b_j\in\Z_{\ge1}$ items of each type $j\in[d]$ among $n_t\in\Z_{\ge1}$ agents of each type $t\in[\tau]$.
Both item and agent multiplicities are encoded in binary, and $n\coloneqq\sum_t n_t\ge1$.
A type-$t$ agent receiving bundle $\vect x\in\Z_{\ge0}^d$ has additive utility
$u_t(\vect x)=\sum_{j=1}^d v_{tj}x_j$, where $v_{tj}\in\Q_{\ge0}$.
For a fixed rational $p<1$, $p\ne0$, the generalized-mean welfare of an allocation with individual utilities $U_1,\ldots,U_n$ is
\[
    M_p\coloneqq\left(\frac1n\sum_{i=1}^n U_i^p\right)^{1/p}.
\]
For $p<0$, we set $M_p=0$ if any $U_i=0$.
Nguyen and Elbassioni~\cite[Theorems~6 and~9]{nguyenElbassioni2026} give an FPTAS for fixed $d$ and a \ac{PTAS} for fixed $\tau$ for each fixed finite rational $p<1$.
We obtain an exact algorithm when both type counts are parameters.

\begin{corollary}[Generalized-mean welfare]\label{cor:mean-welfare}
    For every fixed rational $p<1$, $p\ne0$, an allocation maximizing $M_p$ can be computed deterministically in time
    \[
        2^{\tau\log(\tau+1)\,2^{O_p(d)}}|I|^{O(1)}.
    \]
    The output is a compressed allocation using at most $2^{2d+1}$ distinct bundles per agent type.
\end{corollary}

\begin{proof}
    Choose an integer $q\ge1$ clearing the denominators of all $v_{tj}$, and an integer $A\ge1$ such that
    $q u_t(\vect x)\in\{0,\ldots,A\}$ for every integer $\vect x$ in the box $0\le\vect x\le\vect b$.
    For example, take $A=\max\{1,\max_t\sum_j qv_{tj}b_j\}$; both $q$ and $A$ have polynomial encoding length.
    Prescribe $N^{(t)}=n_t$ in $\mathcal A$ and take $Q=\{\vect b\}$, so that generators represent individual bundles and all items are allocated.

    For $0<p<1$, use the generator polytope $P^{(t)}=\{\vect x:0\le\vect x\le\vect b\}$ and the convex cost
    $f^{(t)}(\vect x)=-u_t(\vect x)^p$.
    Minimizing its sum is equivalent to maximizing $M_p$.
    For $p<0$, use $P^{(t)} = \set{\vect x}[0 \leq \vect x \leq \vect b, u_t(\vect x) \geq \frac{1}{q}]$, and test feasibility using \Cref{prop:counted-optimization} with zero costs and $N_{\max}=n$.
    On integer bundles this restriction is equivalent to positive utility.
    If infeasible, every allocation has welfare zero; giving all items to one agent and empty bundles to the others is optimal and has the claimed support bound.
    Otherwise use $f^{(t)}(\vect x)=u_t(\vect x)^p$, which is convex on the restricted polytope.
    Since $1/p<0$, minimizing its sum again maximizes $M_p$.

    The costs need not be rational.
    Since every feasible representation has total multiplicity $n$, we apply \Cref{prop:general-convex-costs} with the supplied bound $N_{\max}=n$.
    We verify \Cref{ass:general-convex-costs} below.

    We first verify the magnitude bound and the evaluation-oracle requirement of \Cref{ass:general-convex-costs}.
    Write $p=a/b$ in lowest terms, where $a\in\mathbb{Z}\setminus\{0\}$ and $b\in\mathbb{Z}_{\ge 1}$ are fixed.
    Every positive utility at an integer generator has the form $c/q$, with $1\le c\le A$.
    Consequently,
    \[
        |f^{(t)}(\vect x)|\le V:=\max\{A,q\}^{|a|},
    \]
    and $V$ has polynomial encoding length.

    Given an integer generator $\vect x$ and a precision $r\in\mathbb{Z}_{\ge 0}$, a zero utility, which is allowed only when $p>0$, is handled exactly.
    Otherwise, write $u_t(\vect x)=c/q>0$ and compute the largest integer
    $h\in\{0,\ldots,2^rV\}$ satisfying
    \[
        \left(\frac{h}{2^r}\right)^b \le \left(\frac{c}{q}\right)^a.
    \]
    This can be done by binary search using exact rational comparisons.
    Since $a$ and $b$ are fixed, each comparison involves numbers of encoding length $O(|I|+r)$, and the number of comparisons is $O(1+\log V+r)$.
    Thus the computation takes $(|I|+r)^{O(1)}$ time.
    By the choice of $h$,
    \[
        0\le u_t(\vect x)^p-\frac{h}{2^r}<2^{-r}.
    \]
    Returning $-h/2^r$ when $0<p<1$, and $h/2^r$ when $p<0$, therefore gives a value in $2^{-r}\mathbb{Z}$ approximating $f^{(t)}(\vect x)$ within $2^{-r}$.

    It remains to bound nonzero differences between objective values of sparse feasible allocations.
    Put $K=\tau2^{2d+1}$ as in \Cref{sec:general-convex-costs}.
    Any such difference, after dropping zero terms and absorbing the common cost sign, has the form
    \[
        \delta=\sum_{i=1}^{J}\mu_i(c_i/q)^{a/b},
        \qquad J\le2K,\quad \mu_i\in\bZ,\quad 1\le c_i\le A,\quad
        \sum_i|\mu_i|\le2n.
    \]
    We use the following specialization of the separation bound of~\cite[Theorem~1]{burnikel2009separation}.\begin{claim*}[{Application of~\cite[Theorem~1]{burnikel2009separation}}]
        Let \(J, b\in\bZ_{\geq 1}\), and let \(r_1, \ldots, r_J > 0\) be rational numbers whose numerators and denominators are at most $R\ge1$.
        Let $\mu_1, \ldots, \mu_J\in\bZ$ satisfy $\sum_i|\mu_i|\le C$ for some $C\ge1$.
        Then
        \[
            s=\sum_{i=1}^{J}\mu_i\sqrt[b]{r_i}\ne0
            \quad\implies\quad
            |s|\ge(CR^J)^{-b^J}.
        \]
    \end{claim*}
    \begin{proof}\renewcommand{\qedsymbol}{\ensuremath{\blacksquare}}
        Write $r_i=p_i/q_i$ with $1\le p_i,q_i\le R$.
        For $b\ge2$, let $E_i=\sqrt[b]{p_i/q_i}$.
        The root rules in Theorem 1 of \cite{burnikel2009separation} assign
        \[
            (u(E_i),l(E_i))=
            \begin{cases}
                ((p_iq_i^{b-1})^{1/b},q_i), & p_i\ge q_i, \\
                (p_i,(p_i^{b-1}q_i)^{1/b}), & p_i<q_i,
            \end{cases}
        \]
        so $u(E_i),l(E_i)\le R$.
        Represent $s$ by the expression tree $E=\sum_{i=1}^{J}\mu_i E_i$, with the coefficients $\mu_i$ as integer operands.
        Its $J$ root operations give weight $D(E)=b^J$, and the multiplication and addition rules yield
        \[
            l(E)=\prod_i l(E_i)\le R^J,
            \qquad
            u(E)=\sum_i|\mu_i|u(E_i)\prod_{j\ne i}l(E_j)\le CR^J.
        \]
        Theorem 1 in \cite{burnikel2009separation} therefore gives, whenever $s\ne0$,
        \[
            |s|\ge\frac{1}{l(E)u(E)^{D(E)-1}}
            \ge\frac{1}{R^J(CR^J)^{b^J-1}}
            \ge(CR^J)^{-b^J}.
        \]
        For $b=1$, multiplying $s$ by $\prod_iq_i\le R^J$ gives a nonzero integer, so $|s|\ge R^{-J}\ge(CR^J)^{-1}$.
    \end{proof}
    Put $M=\max\{2,A,q\}$ and write $(c_i/q)^{a/b}=\sqrt[b]{r_i}$ with $r_i=(c_i/q)^a$.
    Each $r_i$ is rational, with numerator and denominator at most $M^{|a|}$, including when $a<0$.
    Applying the claim with $C=2n$ and $R=M^{|a|}$, and using $J\le2K$, gives
    \[
        \delta\ne0
        \quad\Longrightarrow\quad
        |\delta|\ge\bigl(2nM^{2K|a|}\bigr)^{-b^{2K}}.
    \]
    Hence the computable integer
    \[
        B=1+b^{2K}\bigl(\enc(2n)+2K|a|\enc(M)\bigr)
        \le2^{O_p(K)}|I|^{O(1)}
    \]
    satisfies the required gap condition, since $a,b$ are fixed and $n,M$ have polynomial encoding length.
    \Cref{prop:general-convex-costs} now gives exact optimization and the stated running time and support bound.
\end{proof}
For every fixed rational \(p<1\), \(p\ne0\), \Cref{cor:mean-welfare} resolves the generalized-mean components of Open Problems 2 and 4 in Nguyen and Rothe~\cite[Sections 4.3.1 and 4.3.3]{nguyenRothe2026}.
Indeed, in their notation, the number of agent types is at most \(n\), and \(q\le(d+1)^n\), yielding fixed-parameter tractability under both parameterizations \(n+q\) and \(n+d\) from Open Problem 2.

\subsection{Additive Approximation for Makespan Balancing}\label{sec:additive-scheduling}
Buchem, Rohwedder, Vredeveld, and Wiese~\cite{DBLP:conf/icalp/BuchemRVW21} gave additive approximation schemes with error $\varepsilon p_{\max}$ for makespan minimization, machine covering, and envy minimization on identical machines, where $p_{\max}$ is the largest processing time.
Their approach combines a mixed-integer relaxation with a polynomial-time local-search rounding algorithm.
For $n$ jobs, $m$ machines, and $\tau$ distinct target load intervals, their full version~\cite[Theorem~12]{DBLP:journals/corr/abs-2007-09333} obtains additive error $\varepsilon p_{\max}$ in time $m^{\tau+1}n^{O(1/\varepsilon)}$.
They explicitly left open whether makespan minimization and machine covering on identical machines admit additive schemes with running time $f(1/\varepsilon)\cdot |I|^{O(1)}$~\cite[Section~5]{DBLP:conf/icalp/BuchemRVW21}.
We retain their relaxation and rounding algorithm, and use \Cref{thm:main} to solve the slot-MILP exactly in \ac{FPT} time parameterized by $1/\varepsilon+\tau$.
This yields additive \acp{EPTAS} for these two objectives and envy minimization on identical machines, as well as additive schemes for these three objectives on uniformly related machines parameterized by the accuracy and the number of distinct speeds.

The continuous coordinates in \Cref{thm:main} allow us to obtain additive approximation schemes for scheduling.
Consider $n$ explicitly listed jobs with positive rational processing times $p_1,\ldots,p_n$ and $m$ explicitly listed machines with rational target intervals $[\ell_i,u_i]$.
We seek an assignment $\sigma:[n]\to[m]$ satisfying
\[
    \ell_i\le L_i\coloneqq \sum_{j\in\sigma^{-1}(i)}p_j\le u_i
    \qquad \forall i\in[m].
\]
Write $p_{\max}\coloneqq \max_j p_j$.
\cite{DBLP:conf/icalp/BuchemRVW21} obtain additive error at most $\eta p_{\max}$ by rounding a solution of the following relaxation.
Fix $\eta\in(0,1]$ with $d\coloneqq \eta^{-1}\in\Z_{\ge1}$ and partition the jobs into classes
\[
    J_c\coloneqq \{j\in[n]:(c-1)\eta p_{\max}<p_j\le c\eta p_{\max}\}
    \qquad \forall c\in[d].
\]
The slot-MILP from~\cite{DBLP:journals/corr/abs-2007-09333} has continuous assignment variables $u_{j,i}$ and integer class counts $x_{i,c}$:
\begin{equation}\label{slotMILP}\tag{slot-MILP}
    \begin{aligned}
         & \sum_{i\in[m]}u_{j,i}=1                  &  & j\in[n],           \\
         & \ell_i\le\sum_{j\in[n]}p_ju_{j,i}\le u_i &  & i\in[m],           \\
         & \sum_{j\in J_c}u_{j,i}=x_{i,c}           &  & i\in[m],\ c\in[d], \\
         & u_{j,i}\ge0                              &  & j\in[n],\ i\in[m], \\
         & x_{i,c}\in\Z_{\ge0}                      &  & i\in[m],\ c\in[d].
    \end{aligned}
\end{equation}
Thus jobs may be split, but every machine receives an integral number of jobs from each size class.
Every integral schedule satisfying the target intervals is feasible for~\eqref{slotMILP} (see~\cite{DBLP:journals/corr/abs-2007-09333}).
\begin{lemma}[{\cite[Lemma 11]{DBLP:journals/corr/abs-2007-09333}}]\label{prop:local-search-rounding}
    Given a solution to~\eqref{slotMILP}, one can compute in polynomial time an integral assignment satisfying
    \[
        \ell_i-\eta p_{\max}\le\sum_{j\in\sigma^{-1}(i)}p_j\le u_i+\eta p_{\max}
        \qquad \forall i\in[m].
    \]
\end{lemma}

Let $(\ell^{(t)},u^{(t)})$, $t\in[\tau]$, be the distinct target intervals, and let $m_t$ be the number of machines with interval $t$.
For each type define the bounded rational polytope
\[
    P^{(t)}\coloneqq \left\{(\vect x,\vect u)\in\R_{\ge0}^{d}\times[0,1]^n:
    \begin{array}{l}
        \ell^{(t)}\le\sum_{j\in[n]}p_ju_j\le u^{(t)}, \\
        \sum_{j\in J_c}u_j=x_c\quad \forall c\in[d]
    \end{array}\right\},
    \qquad
    S^{(t)}\coloneqq P^{(t)}\cap(\Z^d\times\R^n).
\]
A generator $(\vect x, \vect u)$ describes the fractional assignment to one machine; only its $d$ class counts must be integral.
In particular, the $n$ continuous coordinates need not be bounded in terms of the parameters.

\begin{lemma}\label{prop:convert-slot-solutions}
    The system~\eqref{slotMILP} is feasible if and only if there are finitely supported nonnegative integer multiplicities $\lambda^{(t)}_{(\vect x,\vect u)}$, $(\vect x,\vect u)\in S^{(t)}$, such that
    \[
        \sum_{(\vect x,\vect u)\in S^{(t)}}\lambda^{(t)}_{(\vect x,\vect u)}=m_t
        \qquad(t\in[\tau]),
    \]
    and
    \[
        \sum_{t\in[\tau]}\sum_{(\vect x,\vect u)\in S^{(t)}}
        \lambda^{(t)}_{(\vect x,\vect u)}(\vect x,\vect u)
        \in\R_{\ge0}^{d}\times\{1\}^n.
    \]
    Rational solutions in either formulation can be converted to the other in polynomial time.
\end{lemma}
\begin{proof}
    A solution of~\eqref{slotMILP} assigns a vector $(x_{i,1},\ldots,x_{i,d},u_{1,i},\ldots,u_{n,i})$ to each machine $i$.
    Group identical vectors within each machine type and take their multiplicities.
    Conversely, assign $\lambda^{(t)}_{(\vect x,\vect u)}$ machines of type $t$ the vector $(\vect x,\vect u)$.
    The aggregate continuous coordinates enforce that every job is assigned in total once.
    The total multiplicity is $m$, so these conversions take polynomial time for explicitly listed machines.
\end{proof}
\begin{corollary}\label{prop:solve-slot-milp}
    Feasibility of~\eqref{slotMILP} can be decided, and a solution constructed when one exists, in time
    \[
        2^{\tau\log(\tau+1)\,2^{O(1/\eta)}}\cdot|I|^{O(1)}.
    \]
\end{corollary}
\begin{proof}
    Apply \Cref{thm:main} with $d=1/\eta$ integer coordinates, $k=n$ continuous coordinates, $Q\coloneqq \R_{\ge0}^{d}\times\{1\}^n$, and the constraints $N^{(t)}=m_t$ in $\mathcal A$.
    Set $f^{(t)}(\vect x)\coloneqq 0$ on the integer coordinates for every type.
    All polytopes have polynomial encoding length, and \Cref{prop:convert-slot-solutions} converts the resulting certificate to a solution of~\eqref{slotMILP}.
\end{proof}

\subsubsection{Consequences for Scheduling Problems}
For uniformly related machines, job $j$ takes time $p_j/s_i$ on machine $i$, where the speeds $s_i$ are positive rationals.
The problems $Q||C_{\max}$ and $Q||C_{\min}$ minimize the maximum completion load and maximize the minimum completion load, respectively.
\(Q||C_{\mathrm{envy}}\) asks to minimize the difference between the maximum and minimum load.
Let $\tau$ denote the number of distinct speeds, and write $s_{\min}\coloneqq \min_i s_i$ and $s_{\max}\coloneqq \max_i s_i$.
\begin{theorem}\label{prop:q-peptas}
    For every rational $\varepsilon\in(0,1]$, there is an algorithm running in time
    \[
        2^{\tau\log(\tau+1)\,2^{O(1/\varepsilon)}}\cdot|I|^{O(1)}
    \]
    that returns a schedule with completion load at most $\mathrm{OPT}+\varepsilon p_{\max}/s_i$ on each machine $i$ for $Q||C_{\max}$, and at least $\mathrm{OPT}-\varepsilon p_{\max}/s_i$ for $Q||C_{\min}$.
    Hence both problems admit an additive approximation scheme with error $\varepsilon p_{\max}/s_{\min}$, parameterized by $1/\varepsilon+\tau$.
\end{theorem}
\begin{proof}
    Use the slot-MILP with $\eta\coloneqq 1/\lceil2/\varepsilon\rceil\le\varepsilon/2$.
    Put $P \coloneqq \sum_jp_j$, $H\coloneqq P/s_{\min}$, and $\delta\coloneqq \varepsilon p_{\max}/(2s_{\max})$.
    For makespan minimization, a guess $T\in[0,H]$ gives target intervals $[0,s_iT]$.
    Feasibility of the slot-MILP is monotone in $T$, and every $T\ge\mathrm{OPT}$ is feasible.
    Binary search with \Cref{prop:solve-slot-milp}, starting with the feasible upper endpoint $H$, finds a feasible value $T\le\mathrm{OPT}+\delta$ once the interval has width at most $\delta$.
    By \Cref{prop:local-search-rounding}, the resulting schedule satisfies
    \[
        \sum_{j\in\sigma^{-1}(i)}\frac{p_j}{s_i}
        \le T+\frac{\eta p_{\max}}{s_i}
        \le\mathrm{OPT}+\delta+\frac{\eta p_{\max}}{s_i}
        \le\mathrm{OPT}+\frac{\varepsilon p_{\max}}{s_i}.
    \]
    For $Q||C_{\min}$, use the intervals $[s_iT,P]$ instead.
    Feasibility is now monotone decreasing in $T$, and every $T\le\mathrm{OPT}$ is feasible.
    Starting with the feasible lower endpoint $0$, binary search finds a feasible $T\ge\mathrm{OPT}-\delta$.
    Rounding gives completion load at least $T-\eta p_{\max}/s_i\ge\mathrm{OPT}-\varepsilon p_{\max}/s_i$ on every machine.
    Both searches use
    \[
        O\!\left(1+\log\frac{H}{\delta}\right)
        =O\!\left(1+\log\frac{n s_{\max}}{\varepsilon s_{\min}}\right)
    \]
    oracle calls with polynomially encoded rational bounds.
    There are at most $\tau$ target intervals at every call, which proves the running-time bound.
\end{proof}
\begin{theorem}\label{prop:q-peptasenvy}
    For every rational $\varepsilon\in(0,1]$, there is an algorithm running in time
    \[
        2^{\tau\log(\tau+1)\,2^{O(1/\varepsilon)}}\cdot|I|^{O(1)}
    \]
    that returns a schedule with difference between the maximum and minimum load at most $\mathrm{OPT}+\varepsilon p_{\max}/s_{\min}$.
    Hence \(Q||C_{\mathrm{envy}}\) admits an additive approximation scheme with error $\varepsilon p_{\max}/s_{\min}$, parameterized by $1/\varepsilon+\tau$.
\end{theorem}
\begin{proof}
    Define the completion load of a machine in the schedule \(\sigma : [n] \to [m]\) as
    \[
        C_i \coloneqq \frac{\sum_{j\in\sigma^{-1}(i)}p_j}{s_i}
    \]
    and write \(E\coloneqq\max_i C_i-\min_i C_i\).
    The average completion load is
    \[
        A \coloneqq \frac{\sum_jp_j}{\sum_is_i}
    \]
    and lies between the maximum and minimum completion load of any schedule.
    Building a schedule by assigning jobs consecutively to a currently least-loaded machine gives $E\le \frac{p_{\max}}{s_{\min}}$, so in an optimal schedule every load is in $[A - \frac{p_{\max}}{s_{\min}}, A + \frac{p_{\max}}{s_{\min}}]$.
    Set $\eta\coloneqq1/\lceil4/\varepsilon\rceil$ and enumerate all intervals $[\ell,u]$ with endpoints in
    \[
        G\coloneqq\set[\Big]{A+j\eta \frac{p_{\max}}{s_{\min}}}[j\in\Z,\ -1/\eta\le j\le1/\eta].
    \]
    For each interval, solve~\eqref{slotMILP} with the target interval scaled by the machine speeds (as in the proof of \Cref{prop:q-peptas}) using \Cref{prop:solve-slot-milp}, and round each feasible solution using \Cref{prop:local-search-rounding}.
    Return the resulting schedule with minimum envy.
    Rounding the minimum and maximum loads of an optimal schedule outwards to $G$ gives one feasible interval of width at most $\mathrm{OPT}+2\eta \frac{ p_{\max}}{s_{\min}}$.
    The rounding lemma increases its width by at most another $2\eta \frac{p_{\max}}{s_{\min}}$, so the returned schedule satisfies
    \[
        E\le\mathrm{OPT}+4\eta \frac{p_{\max}}{s_{\min}}\le\mathrm{OPT}+\varepsilon \frac{p_{\max}}{s_{\min}}.
    \]
    There are $O(1/\varepsilon^2)$ candidate intervals.
    Each call has at most \(\tau\) distinct target intervals, with polynomially encoded rational endpoints.
    Since \(\frac{1}{\eta} = O(\frac{1}{\varepsilon})\), \Cref{prop:solve-slot-milp} gives the claimed running time.
\end{proof}
For makespan minimization, the factor involving the minimum speed cannot in general be omitted when speeds below one are allowed.
\begin{lemma}
    It is NP-hard to approximate $Q||C_{\max}$ with additive error $p_{\max}/8$, even when $\tau=2$, $s_{\max}=1$, and all processing times are integral.
\end{lemma}
\begin{proof}
    Take a strongly NP-hard instance of \textsc{3-Partition} with positive integers $a_1,\ldots,a_{3m}$ and $B$, where $\sum_ja_j=mB$ and $B/4<a_j<B/2$.
    We may restrict to instances with polynomially bounded $B$.
    Create these $3m$ jobs and $B$ further jobs of processing time $2B$, together with $m$ machines of speed $1/(2B)$ and one machine of speed $1$.
    Put $T\coloneqq 2B^2$; here $p_{\max}=2B$.
    A valid 3-partition gives a schedule of makespan $T$: assign each triple to a slow machine and all additional jobs to the fast machine.

    Conversely, consider a schedule of makespan at most $T+p_{\max}/8=T+B/4$.
    No job of size $2B$ fits on a slow machine, since its completion time there is $4B^2$.
    Thus all such jobs occupy the fast machine, giving it load $T$.
    No original job can join them because $a_j>B/4$.
    Each slow machine receives only original jobs and has integral total processing requirement at most $B$: a requirement of $B+1$ would give completion load $T+2B>T+B/4$.
    Their total processing requirement is $mB$, so every slow machine receives exactly $B$ and hence exactly three original jobs.
    This recovers a valid 3-partition and proves the claim.
\end{proof}

\begin{corollary}\label{cor:identical-additive}
    Makespan minimization, machine covering, and envy minimization on identical machines admit additive \acp{EPTAS} with error $\varepsilon p_{\max}$ and running time $2^{2^{O(1/\varepsilon)}}\cdot |I|^{O(1)}$.
\end{corollary}
\begin{proof}
    Apply \Cref{prop:q-peptas,prop:q-peptasenvy} with a single machine type of speed one.
\end{proof}

\section{Omitted Proofs}\label{sec:omitted-proofs}

\subsection{Separation and Cost Oracles}\label{sec:appendix-oracles}

We use the notation from \Cref{sec:algorithm,sec:separation} and the bounds from \Cref{sec:preliminaries}.
Fix a type $t$ and write
\[
    P=P^{(t)}=\{(\vect x,\vect u):A\vect x+B\vect u\le\vect c\},
    \qquad S=P\cap(\Z^d\times\R^k),\qquad f=f^{(t)}.
\]
For $\vect r\in\{0,\ldots,2^d-1\}^d$, write $S_{\vect r}=P\cap((\vect r+2^d\Z^d)\times\R^k)$ and $P_{\vect r}=\operatorname{conv}(S_{\vect r})$.
The fractional cost $\Phi_{\vect r}$ is defined by~\eqref{eq:cost-envelope} with the type index suppressed.
We retain the coordinate bound $P\subseteq[-M,M]^{d+k}$ from~\eqref{eq:coordinate-bound} and the cost magnitude bound $|f(\vect x)|\le2^{b_f}$ from \Cref{ass:cost-properties}.

We first implement integer pricing over a residue class.
For a fixed integer block, minimizing over its fiber is a linear program, which allows us to perform the integer search in dimension $d$.

\begin{lemma}[Integer Pricing]\label{lem:integer-pricing}
    Given $\vect r\in\{0,\dots,2^d-1\}^d$ and $\vect a\in\Q^{d+k}$, one can determine that $S_{\vect r}=\emptyset$ or return a rational minimizer of
    \[
        \min_{(\vect x,\vect u)\in S_{\vect r}}\bigl(f(\vect x)-\langle\vect a,(\vect x,\vect u)\rangle\bigr)
    \]
    deterministically in time $2^{O(d^2\log(d+1))}\cdot(|I|+\enc(\vect a))^{O(1)}$.
    The returned point can be chosen in $\mathcal V_{\vect r}$ and has encoding length polynomial in $|I|$.
\end{lemma}
\begin{proof}
    Write $\vect a=(\vect a_x,\vect a_u)\in\Q^d\times\Q^k$ and substitute $\vect x=\vect r+2^d\vect\xi$.
    Let
    \[
        P'\coloneqq \{\vect\xi\in\R^d:\text{there exists }\vect u\in\R^k
        \text{ with }A(\vect r+2^d\vect\xi)+B\vect u\le\vect c\}
    \]
    and, for $\vect\xi\in P'$, define
    \[
        \ell(\vect\xi)\coloneqq \min\{-\langle\vect a_u,\vect u\rangle:
        A(\vect r+2^d\vect\xi)+B\vect u\le\vect c\}.
    \]
    This minimum is attained because each nonempty fiber is a polytope.
    The function $\ell$ is convex on $P'$: convex combinations of feasible fiber points remain feasible for the corresponding convex combination of their first coordinates.
    Consequently, the integer pricing function
    \[
        g(\vect\xi)\coloneqq f(\vect r+2^d\vect\xi)
        -\langle\vect a_x,\vect r+2^d\vect\xi\rangle+\ell(\vect\xi)
        \qquad(\vect\xi\in P'\cap\Z^d)
    \]
    satisfies the discrete Jensen inequality from \Cref{ass:cost-properties}.
    At an integer query in $P'$, its value and an optimal fiber vertex can be computed by exact evaluation of $f$ and rational linear programming.

    Since $P\subseteq[-M,M]^{d+k}$ and $\|\vect r\|_\infty\le 2^d-1$, every point of $P'$ has infinity norm at most $M+1$.
    Hence the Euclidean ball $B_\rho$ centered at the origin with integer radius $\rho\coloneqq d(M+1)$ contains $P'$, and $\enc(\rho)=|I|^{O(1)}$.
    We extend the ordering of feasible integer points by their pricing costs to all of $E\coloneqq B_\rho\cap\Z^d$ without constructing a description of $P'$.
    Write $m$ for the number of rows of $(A,B)$ and define
    \[
        \nu(\vect\xi)\coloneqq \min_{\vect u\in[-M,M]^k}
        \max\bigl(\{0\}\cup
        \{(A(\vect r+2^d\vect\xi)+B\vect u-\vect c)_i:i\in[m]\}\bigr).
    \]
    This function is convex, is exactly computable by a linear program, and is zero precisely on $P'$.
    Indeed, the bounded minimization attains its minimum, and the coordinate bound on $P$ ensures that restricting $\vect u$ to this box does not remove any feasible fiber point.
    Set $G\coloneqq 2^{b_f}+M\|\vect a\|_1+1$ and define $\psi:E\to\Q$ by
    \[
        \psi(\vect\xi)\coloneqq
        \begin{cases}
            g(\vect\xi)-G, & \vect\xi\in P',    \\
            \nu(\vect\xi), & \vect\xi\notin P'.
        \end{cases}
    \]
    At feasible integer points, $|g(\vect\xi)|<G$, so all feasible points have smaller $\psi$-values than all infeasible points.
    Their relative order is exactly that of $g$.
    Two values of $\psi$ can be compared using polynomial-time rational linear programming and at most two evaluations of $f$ at integer-coordinate vectors in $X$.

    A function $\psi$ on a finite lattice set $E$ is \emph{discrete convic} if
    \[
        \vect z\in\vect y+\operatorname{cone}
        \{\vect y-\vect x:\vect x\in E,\ \psi(\vect x)\le\psi(\vect y)\}
        \quad\Longrightarrow\quad\psi(\vect z)\ge\psi(\vect y)
        \qquad(\vect y,\vect z\in E).
    \]
    \cite[Theorem~1]{veselov2020} give a deterministic algorithm minimizing such a function on $B_\rho\cap\Z^d$ using
    $2^{O(d^2\log(d+1))}\log(\rho+1)$ comparison-oracle calls and $2^{O(d^2\log(d+1))}(1+\log\rho)^{O(1)}$ additional bit operations.
    We verify their hypothesis for $\psi$.

    Let $\vect y,\vect z\in E$ and write, for points $\vect x_1,\dots,\vect x_p\in E$,
    \[
        \vect z=\vect y+\sum_{j=1}^p\lambda_j(\vect y-\vect x_j),
        \qquad\lambda_j\ge0,\quad\psi(\vect x_j)\le\psi(\vect y),
        \qquad\Lambda\coloneqq \sum_{j=1}^p\lambda_j.
    \]
    Suppose first that $\vect y\notin P'$.
    Every $\vect x_j$ then satisfies $\nu(\vect x_j)\le\nu(\vect y)$, whether or not it is feasible.
    Since $(1+\Lambda)\vect y=\vect z+\sum_j\lambda_j\vect x_j$, convexity of $\nu$ gives
    \[
        (1+\Lambda)\nu(\vect y)
        \le\nu(\vect z)+\sum_j\lambda_j\nu(\vect x_j)
        \le\nu(\vect z)+\Lambda\nu(\vect y).
    \]
    Thus $\nu(\vect z)\ge\nu(\vect y)>0$, so $\vect z$ is infeasible and $\psi(\vect z)\ge\psi(\vect y)$.
    Suppose now that $\vect y\in P'$.
    Then all $\vect x_j$ are feasible and $g(\vect x_j)\le g(\vect y)$.
    If $\vect z$ is infeasible, $\psi(\vect z)>\psi(\vect y)$ by definition.
    Otherwise all points in the same convex-combination identity belong to $P'\cap\Z^d$, and the discrete Jensen inequality for $g$ yields
    \[
        (1+\Lambda)g(\vect y)
        \le g(\vect z)+\sum_j\lambda_jg(\vect x_j)
        \le g(\vect z)+\Lambda g(\vect y).
    \]
    Therefore $g(\vect z)\ge g(\vect y)$, proving that $\psi$ is discrete convic.

    Apply the cited algorithm to minimize $\psi$ over $E$.
    If its minimizer $\vect\xi^\star$ is infeasible, then $P'\cap\Z^d=\emptyset$, since every feasible integer point lies in $E$ and would have a smaller value.
    Otherwise, set $\vect x^\star\coloneqq \vect r+2^d\vect\xi^\star$.
    An optimal fiber vertex $\vect u^\star$ then gives a pricing minimizer $(\vect x^\star,\vect u^\star)\in S_{\vect r}$.
    Every queried point belongs to $E$ and has polynomial encoding length.
    The linear programs in each comparison have polynomially encoded data, and the evaluations of $f$ take polynomial time by \Cref{ass:cost-properties}.
    Thus each comparison takes $(|I|+\enc(\vect a))^{O(1)}$ time, with all returned numbers of polynomial encoding length.
    Combining this with the comparison count proves the claimed running time.
\end{proof}

\fractionalCostOracle*
\begin{proof}
    For this proof, fix a type $t\in[\tau]$ and write $P=P^{(t)}$, $S=S^{(t)}$, and $f=f^{(t)}$, suppressing the type index on residue sets, cones, and cost functions.

    For $(\vect z,\vect w,n)=(\vect0,\vect0,0)$, the empty representation has cost zero.
    Setting $\vect a^\star\coloneqq \vect0$ and computing $b^\star\coloneqq \min_{(\vect x,\vect u)\in S_{\vect r}}f(\vect x)$ by \Cref{lem:integer-pricing} gives a dual feasible pair attaining this value.
    Hence assume $n>0$.
    The primal feasible region in~\eqref{eq:cost-envelope} is nonempty and compact.
    By strong LP duality, its optimal value is achieved by the dual linear program~\eqref{eq:envelope-dual}, which has $d+k+1$ variables and $|\mathcal V_{\vect r}|$ inequalities.
    We solve this dual LP via the ellipsoid method using a strong separation oracle: a query pair $(\bar{\vect a},\bar b)\in\Q^{d+k}\times\Q$ is dual feasible if and only if
    \begin{equation}\label{eq:pricing-min}
        \min_{(\vect x,\vect u)\in S_{\vect r}}\bigl(f(\vect x)-\langle\bar{\vect a},(\vect x,\vect u)\rangle\bigr)\ge\bar b.
    \end{equation}
    By \Cref{lem:integer-pricing}, an exact minimizer $(\vect x^\star,\vect u^\star)\in\mathcal V_{\vect r}$ can be computed in time $2^{O(d^2\log(d+1))}(|I|+\enc(\bar{\vect a}))^{O(1)}$.
    If its pricing value is smaller than $\bar b$, the inequality $\langle\vect a,(\vect x^\star,\vect u^\star)\rangle+b\le f(\vect x^\star)$ separates the query from the dual feasible region.

    Each inequality of the dual LP has coefficients $(\vect x,\vect u,1,f(\vect x))$ of polynomial encoding length.
    Thus this LP has a known polynomial bound on the encoding length of its inequalities, despite their possibly exponential number.
    Applying the strong optimization and primal recovery algorithms of Gr\"otschel, Lov\'asz, and Schrijver~\cite[Theorem~6.4.9 and Lemma~6.5.15]{grotschel1988} uses polynomially many separation queries of polynomial encoding length and yields an optimal dual pair and an optimal basic primal solution supported on at most $d+k+1$ generators.
    The dual pair, generator coordinates, and primal coefficients have encoding length polynomial in $|I|+\enc(\vect z,\vect w,n)$, and this LP procedure has the stated running-time bound.

    To obtain the support bound $2^d$, we minimize a quadratic secondary objective on the integer coordinates.
    Among the minimum-cost representations, choose one minimizing
    \[
        \sum_{(\vect x,\vect u)}\theta_{(\vect x,\vect u)}p(\vect x),
        \qquad
        p(\vect x)\coloneqq \left\|\frac{\vect x-\vect r}{2^d}\right\|_2^2.
    \]
    Both objectives depend only on the integer coordinates of each generator.
    Expanding any generator into vertices of its fiber therefore preserves both objectives, the aggregate, and the count.
    Thus minimizing the two objectives lexicographically over the finite LP is equivalent to doing so over all fractional representations on $S_{\vect r}$.
    Its feasible region is compact, and the two successive minimizations have an optimal basic solution with at most $d+k+1$ positive coefficients.
    We implement this lexicographic optimization by a rational perturbation of the cost.

    Every column of this LP consists of a fiber vertex, a count coordinate equal to one, and an original cost coefficient.
    The fiber vertices have polynomial encoding length by the bounds for rational polytopes.
    The magnitude bound and common denominator in \Cref{ass:cost-properties} give the computable bound $\enc(f(\vect x))=O(b_f+\enc(q))$ on the cost coefficients.
    Consequently, Cramer's rule gives a computable bound $E=(|I|+\enc(\vect z,\vect w,n))^{O(1)}$ on the encoding lengths of the original objective values at every basic feasible solution.
    This bound uses only the uniform column bound and the $d+k+1$ equality constraints; computing it does not require enumerating $\mathcal V_{\vect r}$.
    In particular, distinct original objective values at basic feasible solutions differ by at least $2^{-2E}$.

    The secondary objective is nonnegative and at most
    \[
        H\coloneqq nd(M+1)^2,
    \]
    since $\| (\vect x-\vect r)/2^d\|_\infty\le M+1$ for every generator.
    Choose
    \[
        \varepsilon\coloneqq \frac{2^{-2E-1}}{H+1},
        \qquad
        \widetilde f(\vect x)\coloneqq f(\vect x)+\varepsilon p(\vect x).
    \]
    The perturbation has polynomial encoding length.
    At a basic feasible solution whose original cost is above the minimum, the original cost gap is at least $2^{-2E}$, whereas the difference between any two secondary contributions after scaling by $\varepsilon$ is less than $2^{-2E-1}$.
    Such a basic feasible solution cannot minimize the perturbed objective.
    The feasible region of the finite LP is a nonempty polytope, since $\theta\ge0$ and $\sum\theta_{(\vect x,\vect u)}=n$.
    It follows that every optimal solution for the perturbed objective first minimizes the original cost and, among those solutions, minimizes the secondary objective.

    The preceding LP procedure also applies to the perturbed cost.
    The quadratic $p$ is convex, so $\widetilde f$ satisfies the discrete Jensen inequality, and its values can be evaluated exactly in polynomial time.
    The coordinate and cost bounds give a computable magnitude bound of polynomial encoding length.
    If $\varepsilon=a/b$ with integers $a,b>0$, then $qb2^{2d}$ is a common denominator for $\widetilde f(\vect x)$ at integer $\vect x$; its encoding length is polynomial as well.
    Thus all pricing and LP operations for the perturbed cost have the same running-time bound, with the polynomially encoded perturbation included in the input.

    Retain the envelope value and optimal dual certificate computed for the original cost.
    Applying the LP procedure to the perturbed cost supplies a lexicographically optimal basic representation on at most $d+k+1$ fiber vertices.

    Combine generators with the same integer coordinates by averaging their continuous coordinates with their fractional multiplicities as weights.
    This preserves the full aggregate, count, and both objectives, and does not increase the support.
    Suppose the resulting representation uses more than $2^d$ generators.
    Two distinct normalized integer vectors
    $\vect a_i=(\vect x_i-\vect r)/2^d$ and
    $\vect a_j=(\vect x_j-\vect r)/2^d$ then have the same parity.
    Consequently, their full midpoint
    \[
        (\vect x_0,\vect u_0)
        =\frac12\bigl((\vect x_i,\vect u_i)+(\vect x_j,\vect u_j)\bigr)
    \]
    belongs to $S_{\vect r}$.
    For $\alpha\coloneqq \min\{\theta_i,\theta_j\}>0$, replace amount $\alpha$ of each generator by amount $2\alpha$ of $(\vect x_0,\vect u_0)$.
    This preserves the full aggregate and count, and the discrete Jensen inequality ensures that the primary cost does not increase.
    Since the primary cost was optimal, it remains unchanged.
    However, the secondary objective decreases by
    \[
        \alpha\left(\|\vect a_i\|_2^2+\|\vect a_j\|_2^2
        -2\left\|\frac{\vect a_i+\vect a_j}{2}\right\|_2^2\right)
        =\frac{\alpha}{2}\|\vect a_i-\vect a_j\|_2^2>0,
    \]
    a contradiction.
    Thus at most $2^d$ generators remain, while the original bound of $d+k+1$ is preserved.
    Since the averaging combines at most $d+k+1$ generators of polynomial encoding length, the resulting rational generators and multiplicities also have polynomial encoding length.
    The averaging takes polynomial time, so the two applications of the LP procedure and the subsequent merging take total time
    \[
        2^{O(d^2\log(d+1))}(|I|+\enc(\vect z,\vect w,n))^{O(1)}.\qedhere
    \]
\end{proof}

\subsection{Reconstruction}
\label{sec:appendix-construction}

\constructiveResidueDecomposition*
\begin{proof}
    For this proof, fix a type $t$ and write $P=P^{(t)}$, $S=S^{(t)}$, and $f=f^{(t)}$, suppressing the type index on residue hulls, cones, and cost functions.
    Write $\vect z=\vect z_{\vect r}$, $\vect w=\vect w_{\vect r}$, and $n=n_{\vect r}$.
    In the trivial case $n=0$, we have $(\vect z,\vect w)=(\vect0,\vect0)$ and the empty representation is valid.
    Otherwise, by \Cref{lem:cost-oracle}, we can compute $m\leq \min\{2^d, d+k+1\}$ rational generators $(\vect x_i,\vect u_i)\in S_{\vect r}$ and positive rational multiplicities $\theta_i$ satisfying
    \[
        \sum_i\theta_i=n,\qquad
        \sum_i\theta_i(\vect x_i,\vect u_i)=(\vect z,\vect w),\qquad
        \sum_i\theta_i f(\vect x_i)=\Phi_{\vect r}(\vect z,\vect w,n).
    \]
    Set
    \[
        q_i\coloneqq \lfloor\theta_i\rfloor,\qquad
        \delta_i\coloneqq \theta_i-q_i\in[0,1),\qquad
        \ell\coloneqq \sum_i\delta_i=n-\sum_iq_i.
    \]
    Keep $(\vect x_i,\vect u_i)$ with multiplicity $q_i$ whenever $q_i>0$.
    These are valid generators with integer multiplicities, and $\ell$ is the number of generators that still need to be constructed.
    Since $\ell$ is an integer and $0\le\ell<m\le2^d$, we have $0\le\ell\le2^d-1$.
    If $\ell=0$, all $\delta_i$ vanish, so the retained generators $(\vect x_i,\vect u_i)$ already have aggregate $(\vect z,\vect w)$, total count $n$, and cost $\Phi_{\vect r}(\vect z,\vect w,n)$.
    Therefore, assume $\ell\ge1$.

    We now construct $\ell-1$ \emph{averaged} generators and then one \emph{final} generator.
    Since all source generators belong to the same residue class, we can compute vectors $\vect a_i\in\Z^d$ such that
    \[
        \vect x_i=\vect r+2^d\vect a_i.
    \]
    Each averaged generator will use $2^d$ source generators with equal weights.
    We can use at most $\rho_i\coloneqq \lfloor2^d\delta_i\rfloor$ copies of $(\vect x_i,\vect u_i)$ without exceeding its remaining multiplicity $\delta_i$.
    Put exactly this many copies of each source generator $(\vect x_i,\vect u_i)$ into a pool.
    For each of the $\ell-1$ averaged generators, take any $2^d$ generators $(\vect x_{i_1},\vect u_{i_1}),\ldots,(\vect x_{i_{2^d}},\vect u_{i_{2^d}})$ from the pool and set
    \begin{equation}\label{eq:averaged-generator}
        (\vect x_{\mathrm{avg}},\vect u_{\mathrm{avg}})
        \coloneqq \frac1{2^d}\sum_{j=1}^{2^d}(\vect x_{i_j},\vect u_{i_j}).
    \end{equation}
    Its integer-coordinate block satisfies
    \[
        \vect x_{\mathrm{avg}}
        =\frac1{2^d}\sum_{j=1}^{2^d}(\vect r+2^d\vect a_{i_j})
        =\vect r+\sum_{j=1}^{2^d}\vect a_{i_j}\in\Z^d.
    \]
    We want to repeat this process $\ell-1$ times, each time removing the used generators from the pool, therefore we need at least $2^d(\ell-1)$ copies in the pool.
    Observe that the available quantities $2^d\delta_i$ sum to $2^d\ell$.
    Rounding them down loses
    \[
        2^d\ell-\sum_i\rho_i=\sum_i(2^d\delta_i-\lfloor2^d\delta_i\rfloor).
    \]
    This loss is an integer because $2^d\ell$ and all $\rho_i$ are integers.
    Each summand lies in $[0,1)$, so the loss is smaller than $m\le2^d$ and hence at most $2^d-1$.
    Thus $\sum_i\rho_i\ge2^d\ell-(2^d-1)\ge2^d(\ell-1)$, as required.

    We now need to show that the averaged generators constructed in~\eqref{eq:averaged-generator} belong to $P_{\vect r}\cap(\Z^d\times\Q^k)$.
    Their first blocks are integral by the calculation above.
    Each full generator is an average of rational points in $S_{\vect r}$, so it is rational and belongs to $P_{\vect r}$ by convexity.
    Write $(\vect x_{\mathrm{avg},j},\vect u_{\mathrm{avg},j})$ for the $j$-th averaged generator, and use each with multiplicity one.

    We now proceed to construct the final generator $(\vect x_{\mathrm{final}},\vect u_{\mathrm{final}})$.
    Let $c_i$ be the number of copies of $(\vect x_i,\vect u_i)$ used by the averaged generators.
    Then $0\le c_i\le\rho_i$ and $\sum_i c_i=2^d(\ell-1)$.
    Summing their defining expressions gives
    \[
        \sum_{j=1}^{\ell-1}(\vect x_{\mathrm{avg},j},\vect u_{\mathrm{avg},j})
        =\sum_i\frac{c_i}{2^d}(\vect x_i,\vect u_i).
    \]
    The multiplicity of $(\vect x_i,\vect u_i)$ left for the final generator is $\beta_i\coloneqq \delta_i-c_i/2^d$.
    Since $c_i\le\lfloor2^d\delta_i\rfloor\le2^d\delta_i$, each $\beta_i$ is nonnegative, and
    \[
        \sum_i\beta_i=\ell-\frac1{2^d}\sum_i c_i=\ell-(\ell-1)=1.
    \]
    Thus the remainder has exactly the multiplicity needed for one generator.
    Define the final generator by
    \[
        (\vect x_{\mathrm{final}},\vect u_{\mathrm{final}})
        \coloneqq \sum_i\beta_i(\vect x_i,\vect u_i).
    \]
    Since the $\beta_i$ are nonnegative and sum to one, this is a rational convex combination of points in $S_{\vect r}$, proving $(\vect x_{\mathrm{final}},\vect u_{\mathrm{final}})\in P_{\vect r}\cap\Q^{d+k}$.
    The integrality of its first block requires a separate argument, because the coefficients $\beta_i$ may be fractional.
    After removing the generators with multiplicities $q_i$ and the averaged generators, the remaining aggregate is
    \begin{align*}
         & (\vect z,\vect w)-\sum_iq_i(\vect x_i,\vect u_i)
        -\sum_{j=1}^{\ell-1}(\vect x_{\mathrm{avg},j},\vect u_{\mathrm{avg},j})    \\
         & \qquad=\sum_i\left(\delta_i-\frac{c_i}{2^d}\right)(\vect x_i,\vect u_i) \\
         & \qquad=\sum_i\beta_i(\vect x_i,\vect u_i)
        =(\vect x_{\mathrm{final}},\vect u_{\mathrm{final}}).
    \end{align*}
    The first block of the left-hand side is integral, so $(\vect x_{\mathrm{final}},\vect u_{\mathrm{final}})\in P_{\vect r}\cap(\Z^d\times\Q^k)$.
    This identity also proves that all constructed generators preserve the full aggregate $(\vect z,\vect w)$.
    The output has total multiplicity $\sum_iq_i+(\ell-1)+1=n$ by construction, and at most $m+\ell\le2^{d+1}-1$ distinct generators occur.

    It remains to check the cost.
    Each $\vect x_{\mathrm{avg},j}$ is the average of the integer blocks of its $2^d$ source generators, while $\vect x_{\mathrm{final}}=\sum_i\beta_i\vect x_i$.
    The discrete Jensen inequality (\Cref{ass:cost-properties}) therefore gives
    \[
        \sum_{j=1}^{\ell-1}f(\vect x_{\mathrm{avg},j})\le\sum_i\frac{c_i}{2^d}f(\vect x_i),
        \qquad
        f(\vect x_{\mathrm{final}})\le\sum_i\beta_i f(\vect x_i).
    \]
    Including the retained generators, the total cost is at most
    \[
        \sum_i\left(q_i+\frac{c_i}{2^d}+\beta_i\right)f(\vect x_i)
        =\sum_i\theta_i f(\vect x_i)
        =\Phi_{\vect r}(\vect z,\vect w,n).
    \]

    Finally, we prove the claimed running time.
    The generators $(\vect x_i,\vect u_i)$ and multiplicities $\theta_i$ returned by \Cref{lem:cost-oracle} have encoding length polynomial in $|I|+\enc(\vect z,\vect w,n)$.
    Since $\rho_i\le2^d-1$, constructing the averaged generators requires processing at most $m(2^d-1)\le2^d(2^d-1)=2^{O(d)}$ copies of the source generators.
    For each averaged generator, group repeated copies of the same source generator into one rational coefficient with denominator $2^d$.
    By \Cref{lem:cost-oracle}, we also have $m\le d+k+1$, so every resulting coordinate is a sum of polynomially many rational numbers of polynomial encoding length.
    The final generator and all intermediate coefficients likewise have polynomial encoding length.
    Computing $(\vect x_{\mathrm{final}},\vect u_{\mathrm{final}})$ and recording the multiplicities $q_i$ in binary therefore takes polynomial time.
    Thus the construction after the oracle call takes $2^{O(d)}(|I|+\enc(\vect z,\vect w,n))^{O(1)}$ time, which is absorbed into the running time of \Cref{lem:cost-oracle}.
\end{proof}

\subsection{Uniform Multiplicity Bound}\label{sec:appendix-multiplicity}

Retain the aggregate system $W\vect y_{\mathrm{agg}}\le\vect c_{\mathrm{agg}}$ and the coefficient bound $U$ from \Cref{sec:main-proof}, and use $s$, $R$, $H$, and $N_\star$ from~\eqref{eq:automatic-count-bound}.

\uniformMultiplicityBound*
\begin{proof}
    The costs depend only on integer generator coordinates, of which there
    are finitely many. Their common denominator implies that a feasible
    objective bounded below attains a minimum.
    If the empty representation is optimal, it already satisfies the bound.
    Otherwise, sparsify an optimum using \Cref{lem:sparse-support} and fix
    its $p\le s$ type--integer-coordinate pairs
    $(t_j,\vect x_j)$, but leave their continuous coordinates free.
    Let $\lambda_j\in\Z_{\ge0}$ be their counts and let
    $\vect w_j\in\R^k$ be their aggregate continuous coordinates.
    Writing the local polytope as
    $A^{(t)}\vect x+B^{(t)}\vect u\le\vect b^{(t)}$, its constraints become
    \[
        B^{(t_j)}\vect w_j
        \le(\vect b^{(t_j)}-A^{(t_j)}\vect x_j)\lambda_j.
    \]
    Together with the aggregate constraints, these form a rational system
    \[
        C\vect\lambda+E\vect w\le\vect c,
        \qquad \vect\lambda\ge\vect0,
    \]
    with at most $pk\le R$ continuous variables and integral coefficients
    bounded by $U$ in absolute value.
    The boundedness of each local polytope forces $\vect w_j=\vect0$ when
    $\lambda_j=0$; when $\lambda_j>0$, division by $\lambda_j$ recovers a
    valid generator $(\vect x_j,\vect w_j/\lambda_j)$.
    Consequently, this formulation is exact, and its objective is
    $\sum_j\lambda_jf^{(t_j)}(\vect x_j)$.

    Projecting out $\vect w$ gives a polyhedron in $\vect\lambda$ with
    an integral inequality description whose coefficients and right-hand
    sides are bounded by $H$.
    Indeed, Farkas' lemma describes this projection by
    $\vect\eta^\top C\vect\lambda\le\vect\eta^\top\vect c$ for
    $\vect\eta\ge\vect0$ satisfying $\vect\eta^\top E=\vect0$.
    Each extreme ray of this nonnegative kernel has support at most
    $R+1$ and can be scaled to an integral vector with entries at most
    $R!U^R$, by the cofactor formula applied to a minimal dependence.
    Its resulting coefficients are therefore at most
    $(R+1)R!U^{R+1}=H$.
    This description need not be explicitly enumerated.

    The projected LP is bounded below: a rational improving recession
    direction could otherwise be scaled to have integral counts and added
    to a feasible integer count vector, contradicting optimality.
    Its feasible region lies in the nonnegative orthant and is therefore pointed, so it has an optimal vertex $\vect\lambda^{\mathrm{LP}}$.
    At this vertex, $p$ independent constraints form an invertible subsystem whose matrix and right-hand side have entries of absolute value at most $H$.
    Cramer's rule consequently gives
    $\|\vect\lambda^{\mathrm{LP}}\|_\infty\le p!H^p$.
    The Leibniz formula bounds every square subdeterminant of the full constraint matrix by $p!H^p$, including the nonnegativity rows because $H\ge1$.
    The proximity theorem~\cite[Theorem~1(i)]{cook1986sensitivity} therefore gives an integer optimum $\vect\lambda^\star$ with
    \[
        \|\vect\lambda^\star\|_\infty\le(p+1)p!H^p,
        \qquad
        \sum_{j=1}^p\lambda_j^\star\le p(p+1)p!H^p\le N_\star.
    \]
    Lifting this count vector to continuous aggregates and dividing each positive aggregate by its count recovers a representation of the same cost.
    Applying the same argument with zero costs proves the feasibility
    assertion.

    For a negative-cost homogeneous representation, sparsify first and fix
    its $p\le s$ type--integer-coordinate pairs. Apply the same projection
    argument to the homogeneous constraints, and intersect the resulting
    count cone with $\sum_j\lambda_j\le1$.
    The resulting bounded rational polytope has a negative-cost vertex
    $\vect\lambda^\star$ whose coordinates have a common denominator
    $q\le p!H^p$.
    Lift it to rational continuous aggregates $\vect w^\star$; such a lift exists because the fiber is a nonempty rational polyhedron.
    Then $(q\vect\lambda^\star,q\vect w^\star)$ gives a homogeneous
    representation with integral counts, negative cost, and total count
    at most $q\le N_\star$.
\end{proof}

To justify the unboundedness test, assume that a feasible representation has already been found.
Introduce a count and a continuous
aggregate block for every type--integer-coordinate pair that occurs
in a local polytope. There are finitely many such pairs, and the
perspective formulation above is an exact rational mixed-integer
linear program. If its objective is unbounded below, its LP relaxation
has a rational improving recession direction. Scaling the direction
makes all counts integral; its continuous blocks then define a
negative-cost representation satisfying the homogeneous aggregate
constraints. Conversely, adding any such representation repeatedly to
a feasible representation proves unboundedness.

\section{Approximable Convex Costs with a Multiplicity Bound}\label{sec:general-convex-costs}

In this appendix we extend \Cref{prop:counted-optimization} to approximable convex costs in the integer-generator case $k=0$.
We impose a supplied bound $N_{\max}\in\Z_{\ge0}$ on the total multiplicity as in~\eqref{eq:count-bound}, and include its binary encoding length in $|I|$.
This suffices for generalized-mean welfare in \Cref{sec:welfare}, where the number of agents fixes the total multiplicity.

As stated in \Cref{ass:cost-properties}, our algorithm and analysis access the cost functions $f^{(t)}$ only at the integer generators $S^{(t)} = P^{(t)} \cap \Z^d$ through the discrete Jensen inequality, polynomial-time exact evaluation and magnitude bounds, and a common lattice denominator.
Therefore, the algorithm applies directly to any cost functions given by a rational evaluation oracle that satisfies \Cref{ass:cost-properties}.

Write $F(\vect\lambda)\coloneqq \sum_{t,\vect x}\lambda^{(t)}_{\vect x}f^{(t)}(\vect x)$, and set $K\coloneqq \tau2^{2d+1}$.
Retain the coordinate bound $M$ from~\eqref{eq:coordinate-bound}.

\begin{assumption}\label{ass:general-convex-costs}
    Each $f^{(t)}$ is convex on $P^{(t)}$, with a known bound $|f^{(t)}(\vect x)|\le2^b$ on integer generators, where $b\le|I|^{O(1)}$.
    Given an integer generator $\vect x$ and $r\in\Z_{\ge0}$, an oracle deterministically returns $\widehat f^{(t)}_{r}(\vect x)\in2^{-r}\Z$ with $|\widehat f^{(t)}_{r}(\vect x)-f^{(t)}(\vect x)|\le2^{-r}$ in time $(|I|+r)^{O(1)}$.

    A supplied integer $B\ge1$ satisfies
    \[
        F(\vect\lambda)\ne F(\vect\mu)
        \quad\Longrightarrow\quad
        |F(\vect\lambda)-F(\vect\mu)|\ge2^{-B}
    \]
    for any two feasible representations $\vect\lambda,\vect\mu$ satisfying~\eqref{eq:count-bound}, each with total support at most $K$.
\end{assumption}

\begin{proposition}\label{prop:general-convex-costs}
    Under \Cref{ass:general-convex-costs}, the problem~\eqref{eq:main-problem} for $k=0$ subject to~\eqref{eq:count-bound} can be solved deterministically in time
    \[
        2^{\tau\log(\tau+1)\,2^{O(d)}}(|I|+B)^{O(1)}.
    \]
    The algorithm reports infeasibility or returns an optimum with at most $2^{2d+1}$ generators per type.
    In particular, the running-time bound of \Cref{prop:counted-optimization} is preserved if such a $B\le2^{O(K)}|I|^{O(1)}$ can be computed within that time.
\end{proposition}

\begin{proof}
    We construct rational costs satisfying \Cref{ass:cost-properties} and choose their approximation precision so that optimizing them with \Cref{prop:counted-optimization} also minimizes the original objective.

    For $\varepsilon\coloneqq 2^{-r}$, define rational lattice costs
    \[
        g_r^{(t)}(\vect x)\coloneqq \widehat f^{(t)}_{r}(\vect x)
        +\varepsilon(1+3\|\vect x\|_2^2),
        \qquad C\coloneqq 2+3dM^2.
    \]
    Since $\|\vect x\|_\infty\le M$ for all $\vect x\in S^{(t)}$, they satisfy $0\le g_r^{(t)}(\vect x)-f^{(t)}(\vect x)\le\varepsilon C$.
    The quadratic correction ensures the discrete Jensen inequality despite the approximation error.
    To see this, suppress the type index and consider a nontrivial convex combination $\vect x=\sum_{i=1}^m\theta_i\vect x_i$ with all points in $S$.
    Removing occurrences of $\vect x$ and normalizing reduces to $\vect x_i\ne\vect x$ for every $i$.
    Convexity of $f$, the bound $|\widehat f_r-f|\le\varepsilon$, and the quadratic identity for this convex combination give
    \[
        \sum_{i=1}^m\theta_i g_r(\vect x_i)-g_r(\vect x)
        \ge 3\varepsilon\sum_{i=1}^m\theta_i\|\vect x_i-\vect x\|_2^2-2\varepsilon
        \ge\varepsilon>0,
    \]
    since distinct integer points have squared distance at least one.
    Moreover, $g_r^{(t)}(\vect x)$ is exactly evaluable in $(|I|+r)^{O(1)}$ time, has magnitude at most $2^b+C$, and lies in $2^{-r}\Z$.
    Since $\enc(M)=|I|^{O(1)}$, these costs satisfy \Cref{ass:cost-properties} with size parameter $|I|+r$.

    Write $G_r(\vect\lambda)\coloneqq\sum_{t,\vect x}\lambda^{(t)}_{\vect x}g_r^{(t)}(\vect x)$.
    Every representation satisfying~\eqref{eq:count-bound} obeys
    \[
        0\le G_r(\vect\lambda)-F(\vect\lambda)\le\varepsilon C N_{\max}.
    \]
    Choose $r=B+\lceil\log_2(C(N_{\max}+1))\rceil+1$, so that $\varepsilon C N_{\max}<2^{-B}$, and apply \Cref{prop:counted-optimization} to the costs $g_r$ with the supplied bound $N_{\max}$.
    Feasibility is unchanged, so an infeasibility report is valid for the original instance.
    Otherwise, let $\vect\lambda$ be the returned optimum for $G_r$.
    Each $S^{(t)}$ is finite, and the multiplicity bound leaves only finitely many representations, so the nonempty feasible set has an optimum for $F$.
    Applying \Cref{lem:sparse-support} to this optimum gives an optimum $\vect\mu$ with total support at most $\tau2^d\le K$, preserving all typewise aggregates and counts.
    Thus both $\vect\lambda$ and $\vect\mu$ are feasible, satisfy~\eqref{eq:count-bound}, and have total support at most $K$.
    Consequently,
    \[
        F(\vect\lambda)\le G_r(\vect\lambda)
        \le G_r(\vect\mu)
        \le F(\vect\mu)+\varepsilon C N_{\max}
        <F(\vect\mu)+2^{-B}.
    \]
    The gap assumption forces $F(\vect\lambda)=F(\vect\mu)$.
    Since $r=O(B+\log(N_{\max}+1)+\log C)=B+|I|^{O(1)}$, the running time and support bound follow from \Cref{prop:counted-optimization}.
\end{proof}

For rational costs evaluable exactly in polynomial time, the evaluation time supplies a bound $L=|I|^{O(1)}$ on the bit length of the returned values.
Then $B=2KL$ suffices: an objective difference involves at most $2K$ denominators, whose product is at most $2^{2KL}$.
Examples include maxima of explicitly given rational affine functions and $1/(1+\ell(\vect x))$ for rational affine $\ell\ge0$ on the generator domain.

\end{document}